\documentclass[11pt]{article}
\usepackage[left=3cm,right=5cm,top=4cm,bottom=3cm]{geometry}

\usepackage{bbm}
\usepackage{amsmath,amsfonts}
\usepackage{mathtools}
\usepackage{color,enumerate,graphicx}
\usepackage{mathrsfs}
\usepackage{xspace}
\usepackage{array}
\usepackage{comment}
\usepackage{wrapfig}
\usepackage{pstricks}
\usepackage{lmodern}
\usepackage{ifthen}
\usepackage[T1]{fontenc}
\usepackage{cite} 
\usepackage{ellipsis}
\usepackage{fixltx2e}

\usepackage[warning, all]{onlyamsmath}
\usepackage{latexsym,fullpage,amsmath,amssymb,amstext,url,verbatim}
\usepackage{epsfig}
\usepackage{graphics,graphicx, hyperref}
\usepackage{amsthm}
\usepackage{float,scrtime}
\usepackage{bbm}
\usepackage[latin1]{inputenc}
\usepackage[ruled,longend,vlined]{algorithm2e}
\floatname{algorithm}{Procedure}

\newcommand{\FPT}{\ensuremath{\text{FPT}}\xspace}

\newcommand{\adj}{\ensuremath{\mathbf{adj}}\xspace}

\newtheorem*{mainI}{Theorem~I}{\bfseries \upshape}{\itshape}
\newtheorem*{mainII}{Theorem~II}{\bfseries \upshape}{\itshape}

\newcommand{\name}[1]{\textsc{#1}}

\newcommand{\pmincmso}{\name{$p$-min-CMSO}}
\newcommand{\pmaxcmso}{\name{$p$-max-CMSO}}
\newcommand{\peqcmso}{\name{$p$-eq-CMSO}}
\newcommand{\YES}{\textsc{Yes}}
\newcommand{\NO}{\textsc{No}}

\newcommand{\mc}{\mathcal}
\newcommand{\mb}{\mathbf}
\newcommand{\mbb}{\mathbb}

\newcommand{\planarF}{\textsc{Planar-$\mc{F}$-Deletion}}
\newcommand{\DplanarF}{\textsc{Disjoint Planar-$\mc{F}$-Deletion}}

\newenvironment{reminder}[1]{\smallskip
\noindent {\bf Reminder of #1  }\em}{}

\newcommand{\Problem}[1]{\textsc{#1}\xspace}
\newcommand{\VC}{\Problem{Vertex Cover}}
\newcommand{\DS}{\Problem{Dominating Set}}
\newcommand{\IM}{\Problem{Induced Matching}}
\newcommand{\FDST}{\Problem{Full-Degree Spanning Tree}}
\newcommand{\CP}{\Problem{Cycle Packing}}
\newcommand{\FVS}{\Problem{Feedback Vertex Set}}
\newcommand{\EDS}{\Problem{Edge Dominating Set}}
\newcommand{\CVD}{\Problem{Chordal Vertex Deletion}}

\newcommand{\ClVD}{\Problem{Cluster Vertex Deletion}}
\newcommand{\IVD}{\Problem{Interval Vertex Deletion}}

\newcommand{\ConnVC}{\Problem{Connected Vertex Cover}}
\newcommand{\ConnDS}{\Problem{Connected Dominating Set}}
\newcommand{\ConnClVD}{\Problem{Connected Cluster Vertex Deletion}}
\newcommand{\ConnCoVD}{\Problem{Connected Cograph Vertex Deletion}}

\newcommand{\NP}{\ensuremath{\text{NP}}\xspace}
\newcommand{\coNP}{\ensuremath{\text{co-NP}}\xspace}
\newcommand{\minor}{\ensuremath{\preceq_{{\mathit{m}}}\!}} 
\newcommand{\tminor}{\ensuremath{\preceq_{\mathit{tm}}\!}} 
\newcommand{\notminor}{\ensuremath{\npreceq_{{\mathit{m}}}\!}}
\newcommand{\nottminor}{\ensuremath{\npreceq_{{\mathit{tm}}}\!}}
\newcommand{\set}[1]{\ensuremath{\left\{#1\right\}}}
\newcommand{\poly}{\mathop\mathit{poly}}
\newcommand{\bound}{\ensuremath{\mathop\mathbf{bd}}}    
\newcommand{\equipi}[3]{\ensuremath{#1 \equiv_{\Pi,#3} #2}} 
\newcommand{\prot}{\ensuremath{\rho}}   
\newcommand{\protd}{\ensuremath{\rho'}} 
\newcommand{\YYYY}{\ensuremath{Y_0 \uplus Y_1 \uplus \cdots \uplus Y_\ell}\xspace} 
\renewcommand{\le}{\leqslant}
\renewcommand{\leq}{\leqslant}

\renewcommand{\geq}{\geqslant}

\newcommand{\widthm}[1]{\ensuremath{\mathop\mathbf{#1}}\xspace}
\newcommand{\tw}{\widthm{tw}}

\newcommand{\cw}{\widthm{cw}}
\newcommand{\rw}{\widthm{rw}}

\newcommand{\fii}{finite integer index\xspace} 

\newtheorem{lemma}{Lemma}
\newtheorem{theorem}{Theorem}
\newtheorem{corollary}{Corollary}
\newtheorem{observation}{Observation}
\newtheorem{proposition}{Proposition}
\newtheorem*{claim}{Claim}

\theoremstyle{definition}
\newtheorem{redrule}{Reduction Rule}
\newtheorem{definition}{Definition}

\newcommand*{\ie}{i.e.\@\xspace}
\newcommand*{\cf}{cf.\@\xspace}
\newcommand*{\wrt}{w.r.t.\@\xspace}
\makeatletter
\newcommand*{\etc}{%
    \@ifnextchar{.}%
        {etc}%
        {etc.\@\xspace}%
}
\makeatother

\newif\ifshort
\newcommand{\short}[1]{%
    \ifshort%
        #1\fi%
}
\newcommand{\journal}[1]{%
    \ifshort%
    \else%
        #1\fi%
}
\newcommand{\omitted}{%
    \ifshort%
        \textup{[$\star$]}
    \else%
    \fi
}

\def\reduce{}
\def\Reduce{}

\usepackage{microtype}
\shortfalse

\title{Linear kernels and single-exponential algorithms\\via protrusion decompositions
\thanks{We would like to point out that this article replaces and extends the results
of \href{http://arxiv.org/abs/1201.2780}{[CoRR, abs/1201.2780, 2012]}. Research
funded by DFG-Project RO 927/12-1 ``Theoretical and Practical Aspects of
Kernelization'', ANR project AGAPE (ANR-09-BLAN-0159), and the
Languedoc-Roussillon Project ``Chercheur d'avenir'' KERNEL.}}

\def\thanksac{Theoretical Computer Science, Department of Computer
    Science, RWTH Aachen University, Germany,
    {\tt \{langer,reidl,rossmani,sikdar\}@cs.rwth-aachen.de}.
    }
\def\thanksacr{\footnotemark[3]}
\def\thanksparis{CNRS, LAMSADE, Paris, France,
                {\tt eunjungkim78@gmail.com}.}
\def\thanksmont{CNRS, LIRMM, Montpellier, France,
                {\tt \{paul,sau\}@lirmm.fr}.}
\def\thansmontr{\footnotemark[4]}

\date{}

\author{Eun Jung Kim\thanks{\thanksparis}
\and Alexander Langer\thanks{\thanksac}
\and Christophe Paul\thanks{\thanksmont}
\and Felix Reidl\thanksacr
\and Peter Rossmanith\thanksacr
\and Ignasi Sau\thansmontr
\and Somnath Sikdar\thanksacr}

\begin{document}
\maketitle
\thispagestyle{empty}


\begin{abstract}
A \emph{$t$-treewidth-modulator} of a graph $G$ is a set $X \subseteq V(G)$
such that the treewidth of $G-X$ is at most~$t-1$. In this paper, we present a
novel algorithm to compute a decomposition scheme for graphs $G$ that come
equipped with a $t$-treewidth-modulator. Similar decompositions have already been
explicitly or implicitly used for obtaining polynomial
kernels~\cite{AFN04,GN07a,BFLPST09,FLMPS11}. Our  decomposition, called
a \emph{protrusion decomposition}, is the cornerstone in obtaining the
following two main results.

Our first result is that any parameterized graph problem (with parameter $k$)
that has \emph{finite integer index} and is \emph{treewidth-bounding} admits a
linear kernel on the class of $H$-topological-minor-free graphs, where $H$ is
some arbitrary but fixed graph. A parameterized graph problem is called
treewidth-bounding if all positive instances have a $t$-treewidth-modulator of
size $O(k)$, for some constant~$t$. This result partially extends previous
meta-theorems on the existence of linear kernels on graphs of bounded
genus~\cite{BFLPST09} and $H$-minor-free graphs~\cite{FLST10}. In particular,
we show that \Problem{Chordal Vertex Deletion}, \Problem{Interval Vertex
Deletion}, \Problem{Treewidth-$t$ Vertex Deletion}, and \Problem{Edge
Dominating Set} have linear kernels on $H$-topological-minor-free graphs.

Our second application concerns the \planarF{} problem. Let $\mathcal{F}$ be a
fixed finite family of graphs containing at least one planar graph. Given an
$n$-vertex graph $G$ and a non-negative integer $k$, \planarF{} asks whether
$G$ has a set $X\subseteq V(G)$ such that $|X|\leqslant k$ and $G-X$ is
$H$-minor-free for every $H\in \mc{F}$. This problem encompasses a number of
well-studied parameterized problems such as \textsc{Vertex Cover},
\textsc{Feedback Vertex Set}, and  \textsc{Treewidth-$t$ Vertex Deletion}. Very
recently, an algorithm for \planarF{} with running time $2^{O(k)}\cdot n \log^2
n$ (such an algorithm is called \emph{single-exponential}) has been presented
in~\cite{FLMS12} under the condition that every graph in $\mc{F}$ is
\emph{connected}. Using our algorithm to construct protrusion decompositions as a
building block, we get rid of this connectivity constraint and present an
algorithm for the general \planarF{} problem running in time $2^{O(k)}\cdot
n^2$. This running time is asymptotically optimal with respect to~$k$, as it is known that unless
the Exponential Time Hypothesis fails, one cannot expect a running time of
$2^{o(k)} \cdot \poly(n)$.

\vspace{.35cm}


\end{abstract}
\newpage

\section{Introduction}\label{sec:Introduction}
\journal{Parameterized complexity deals with algorithms for decision problems
whose instances consist of a pair $(x,k)$, where~$k$ is a secondary measurement
known as the \emph{parameter}. A major goal in parameterized complexity is to
investigate whether a problem with parameter~$k$ admits an algorithm with
running time $f(k) \cdot |x|^{O(1)}$, where~$f$ is a function depending only on
the parameter and $|x|$ represents the input size. Parameterized problems that
admit such algorithms are called \emph{fixed-parameter tractable} and the class
of all such problems is denoted \FPT. For an introduction to the area
see~\cite{DF99,FG06,Nie06}.

A closely related concept is that of \emph{kernelization}. A kernelization
algorithm, or just \emph{kernel}, for a parameterized problem takes an
instance~$(x,k)$ of the problem and, in time polynomial in $|x| + k$, outputs
an equivalent instance~$(x',k')$ such that $|x'|, k' \leqslant g(k)$ for some
function~$g$. The function~$g$ is called the \emph{size} of the kernel and may
be viewed as a measure of the ``compressibility'' of a problem using
polynomial-time preprocessing rules. It is a folklore result in the area that a
decidable problem is in \FPT if and only if it has a kernelization algorithm.
However, the kernel that one obtains in this way is typically of size at least
exponential in the parameter. A natural problem in this context is to find
polynomial or linear kernels for problems that are in \FPT.
}
%
%

\short{
This work contributes to the area of fixed-parameter algorithms and kernels,
see, e.g.,~\cite{DF99,FG06,Nie06} for an introduction.}
\journal{\paragraph{Linear kernels.}}%
During the last decade, a plethora of results emerged on linear
kernels for graph-theoretic problems restricted to {\sl sparse} graph classes.
\journal{A celebrated result in this area is the linear
kernel for \Problem{Dominating Set} on planar graphs by Alber \emph{et
al}.~\cite{AFN04}. This paper prompted an explosion of research papers on
linear kernels on planar graphs, including \DS~\cite{AFN04,CFKX07},
\FVS~\cite{BP08}, \CP~\cite{BPT08}, \IM~\cite{MS09,KPSX11}, \FDST~\cite{GNW10},
and \ConnDS~\cite{LMS11b}. Guo and Niedermeier~\cite{GN07a} designed a general
framework and showed that problems that satisfy a certain ``distance property''
have linear kernels on planar graphs. This result was subsumed by that of
Bodlaender \emph{et al}.~\cite{BFLPST09} who provided a meta-theorem for
problems to have a linear kernel on graphs of bounded genus, a strictly larger
class than planar graphs. Later Fomin \emph{et al}.~\cite{FLST10} extended
these results for bidimensional problems to an even larger graph class, namely,
$H$-minor-free and apex-minor-free graphs. (In all these works, the problems
are parameterized
by the {\sl solution size}.)}%
\short{A celebrated result by Alber \emph{et al}.~\cite{AFN04} prompted an explosion of research
papers on linear kernels on planar graphs, see, e.g.,~\cite{CFKX07,BP08,BPT08,MS09,KPSX11,GNW10,LMS11b}.
Guo and Niedermeier~\cite{GN07a} designed a general framework and showed that problems
that satisfy a certain ``distance property'' have linear kernels on planar
graphs. Bodlaender \emph{et al}.~\cite{BFLPST09} provided a meta-theorem for problems to have a linear
kernel on graphs of bounded genus.
Fomin \emph{et al}.~\cite{FLST10} extended these results for
bidimensional problems on $H$-minor-free and apex-minor-free graphs.
}%
\journal{A common feature of these meta-theorems on sparse graphs is a {\sl
decomposition scheme} of the input graph that, loosely speaking, allows to deal
with each part of the decomposition independently. For instance, the approach
of~\cite{GN07a}, which is much inspired from~\cite{AFN04},  is to consider a
so-called \emph{region decomposition} of the input planar graph. The key point
is that in an appropriately reduced \YES-instance, there are $O(k)$ regions and
each one has constant size, yielding the desired linear kernel. This idea was
generalized in~\cite{BFLPST09} to graphs on surfaces, where the role of regions
is played by \emph{protrusions}, which are graphs with small treewidth and
small boundary (see Section~\ref{sec:Preliminaries} for details). The resulting
decomposition is called \emph{protrusion decomposition}. A crucial point is
that while the reduction rules of~\cite{AFN04} are {\sl problem-dependent},
those of~\cite{BFLPST09} are {\sl automated}, relying on a property called
\emph{finite integer index} (FII), which was introduced by Bodlaender and de
Fluiter~\cite{BvF01}. Loosely speaking (see Section~\ref{sec:Preliminaries}),
having FII essentially guarantees that ``large'' protrusions of a graph can be
replaced by ``small'' gadget graphs preserving equivalence of instances. This
operation is usually called the \emph{protrusion replacement rule}. FII is also of
central importance to the approach of~\cite{FLST10} on $H$-minor-free graphs.

In this article, following the spirit of the aforementioned results, we present
a novel decomposition algorithm to compute protrusion decompositions that
allows us to obtain linear kernels on a larger class of sparse graphs, namely
$H$-topological-minor-free graphs. A \emph{$t$-treewidth-modulator} of a graph
$G$ is a set $X \subseteq V(G)$ such that the treewidth of $G-X$ is at most~$t
- 1$. Our algorithm takes as input a graph $G$ and a $t$-treewidth-modulator $X
\subseteq V(G)$, and outputs a set of vertices $Y_0$ containing $X$ such that
every connected component of $G - Y_0$ is a protrusion (see
Section~\ref{sec:Decomposition} for details). We would like to stress again
that similar decompositions have already been explicitly or implicitly used for
obtaining polynomial kernels~\cite{AFN04,GN07a,BFLPST09,FLMPS11}.

When $G$ is the input graph of a parameterized graph problem $\Pi$ with
parameter $k$, we call a protrusion decomposition of $G$ \emph{linear} if both
$|Y_0|$ and the number of protrusions of $G-Y_0$ are $O(k)$.  If $\Pi$ is such
that \YES-instances have a $t$-treewidth-modulator of size $O(k)$ for some
constant $t$ (such problems are called \emph{treewidth-bounding}, see
Section~\ref{sec:Kernels}), and $G$ excludes some fixed graph $H$ as a
topological minor, we prove that the protrusion decomposition given by our
algorithm is linear. If in addition $\Pi$ has FII, then each protrusion can be
replaced with a gadget of constant size, obtaining an equivalent instance of
size $O(k)$. Our first main result summarizes the above discussion.}%
\short{A common feature of these meta-theorems on sparse graphs is a {\sl
decomposition scheme} of the input graph that, loosely speaking, allows to deal
with each part of the decomposition independently. For instance, the approach
of~\cite{GN07a}, which is much inspired from~\cite{AFN04},  is to consider a
so-called \emph{region decomposition} of the input planar graph. The key point
is that in an appropriately reduced \YES-instance, there are $O(k)$ regions and
each one has constant size, yielding the desired linear kernel. This idea was
generalized in~\cite{BFLPST09} to graphs on surfaces, where the role of regions
is played by \emph{protrusions}, which are graphs with small treewidth and
small boundary (see Section~\ref{sec:Preliminaries} for details). The resulting
decomposition is called \emph{protrusion decomposition}. A crucial point is
that while the reduction rules of~\cite{AFN04} are {\sl problem-dependent},
those of~\cite{BFLPST09} are {\sl automated}, relying on a property called
\emph{finite integer index} (FII), which was introduced by Bodlaender and de
Fluiter~\cite{BvF01}. Loosely speaking, having FII essentially guarantees that
``large'' protrusions of an instance can be replaced by ``small'' equivalent
gadget graphs. This operation is usually called thee\emph{protrusion replacement
rule}. FII is also of central importance to the approach of~\cite{FLST10} on
$H$-minor-free graphs.

In the spirit of the above results, we present
(in Section~\ref{sec:Decomposition}) a novel decomposition algorithm to compute
protrusion decompositions that allows us to obtain (in
Section~\ref{sec:Kernels}) linear kernels on a larger class of sparse graphs. A
\emph{$t$-treewidth-modulator} of a graph $G$ is a set $X \subseteq V(G)$ such
that the treewidth of $G-X$ is at most~$t - 1$. A parameterized problem is
\emph{treewidth-bounding} if \YES-instances have a $t$-treewidth-modulator of
size $O(k)$ for some constant $t$. Our first main result is:
} 

\begin{mainI}
    Fix a graph~$H$. Let~$\Pi$ be a parameterized graph problem on the
    class of $H$-topological-minor-free graphs that is
    treewidth-bounding and has \fii. Then $\Pi$ admits a linear kernel.
\end{mainI}

It turns out that a host of problems including \CVD, \IVD, \EDS,
\Problem{Treewidth-$t$ Vertex Deletion}, to name a few, satisfy the conditions
of our theorem. Since for any fixed graph $H$, the class of
$H$-topological-minor-free graphs strictly contains the class of $H$-minor-free
graphs, our result may be viewed as an extension of the results of Fomin
\emph{et al}.~\cite{FLST10}.

\journal{We also exemplify how our algorithm to obtain a linear protrusion
decomposition can be applied to obtain {\sl explicit} linear kernels, that is,
kernels without using a generic protrusion replacement. This is shown by
exhibiting a simple explicit linear kernel for the \EDS problem on
$H$-topological-minor-free graphs. So far, all known linear kernels for \EDS on
$H$-minor-free graphs~\cite{FLST10} and $H$-topological-minor-free graphs
(given by Theorem~I) relied on generic protrusion replacement. }


\journal{\paragraph{Single-exponential algorithms.} In order to prove
Theorem~I, similarly to~\cite{GN07a,BFLPST09,FLST10} our protrusion
decomposition algorithm is only used to {\sl analyze} the size of the resulting
instance after having applied the protrusion reduction rule. In the second part
of the paper we show that our decomposition scheme can also be used to obtain
{\sl efficient} \FPT algorithms. Before stating our second main result, let us
motivate the problem that we study.}

\short{In order to prove Theorem~I, similarly to~\cite{GN07a,BFLPST09,FLST10}
our protrusion decomposition algorithm is only used to {\sl analyze} the size
of the resulting instance after having applied the protrusion reduction rule.
In the second part of the paper (Section~\ref{sec:PlanarF}) we show that our
decomposition scheme can also be used to obtain {\sl efficient} \FPT
algorithms. Namely, we are interested in \emph{single-exponential} algorithms,
that is, algorithms that solve a parameterized problem with parameter $k$ on an
$n$-vertex graph in time $2^{O(k)} \cdot n^{O(1)}$.\looseness-1}

\journal{During the last decades, parameterized complexity theory has brought
forth several algorithmic meta-theorems that imply that a wide range of
problems are in \FPT (see~\cite{Kre09} for a survey). For instance, Courcelle's
theorem~\cite{Cou90} states that every decision problem expressible in Monadic
Second Order Logic can be solved in linear time when parameterized by the
treewidth of the input graph. At the price of generality, such algorithmic
meta-theorems may suffer from the fact that the function $f(k)$ is
huge~\cite{KT10a,FG04} or non-explicit~\cite{Cou90,RS95c}. Therefore, it has
become a central task in parameterized complexity to provide \FPT algorithms
such that the behavior of the function $f(k)$ is {\sl reasonable}; in other
words, a function $f(k)$ that could lead to a practical algorithm.

Towards this goal, we say that an \FPT parameterized problem is solvable in
{\em single-exponential} time if there exists an algorithm solving it in time
$2^{O(k)}\cdot n^{O(1)}$. For instance, recent results have shown that broad
families of problems admit (deterministic or randomized) single-exponential
algorithms parameterized by treewidth~\cite{CNPPRW11,DFT08,RST11}. On the other
hand, single-exponential algorithms are unlikely to exist for certain
parameterized problems~\cite{LMS11,CNPPRW11}. Parameterizing by the size of the
desired solution, in the case of \VC the existence of a
single-exponential algorithm has been known for a long time, but it took a
while to witness the first (deterministic) single-exponential algorithm for
\FVS, or equivalently \Problem{Treewidth-One Vertex
Deletion}~\cite{GGH06,DFLRS05}.

Both \VC and \FVS can be seen as
graph modification problems in order to attain a hereditary property, that is,
a property closed under taking induced subgraphs. It is well-known that
deciding whether at most $k$ vertices can be deleted from a given graph in
order to attain any non-trivial hereditary property is
\textsc{NP}-complete~\cite{LY80}. The particular case where the property can be
characterized by a finite set of forbidden induced subgraphs can be solved in
single-exponential time when parameterizing by the number of modifications,
even in the more general case where also edge deletions or additions are
allowed~\cite{Cai96}. If the family of forbidden induced subgraphs is infinite,
no meta-theorem is known and not every problem is even \FPT~\cite{Lok08}. A
natural question arises: can we carve out a larger class of hereditary
properties for which the corresponding graph modification problem can be solved
in single-exponential time?

A line of research emerged pursuing this question, which is much inspired by
the \FVS problem. Interestingly, when the {\sl infinite} family of forbidden
{\sl induced subgraphs} can also be captured by a {\sl finite} set $\mc{F}$ of
forbidden {\sl minors}, the \Problem{$\mc{F}$-Deletion} problem (namely, the
problem of removing at most $k$ vertices from an input graph to obtain a graph
which is $H$-minor-free for every $H \in \mc{F}$) is in \FPT by the seminal
meta-theorem of Robertson and Seymour~\cite{RS95c}\footnote{It is worth noting
that, in contrast to the removal of vertices, the problems corresponding to the
operations of removing or contracting edges are not minor-closed (we provide a
proof of this fact in Appendix~\ref{ap:NotMinorClosed}), and therefore the
result of Robertson and Seymour~\cite{RS95c} cannot be applied to these
modification problems.}.}

Let $\mc{F}$ be a finite family of (non-necessarily connected) graphs
containing at least one planar graph. The parameterized problem that we
consider in the second part of the paper is {\sc \planarF{}}: given a graph~$G$
and a non-negative integer parameter~$k$ as input, does $G$ have a set
$X\subseteq V(G)$ such that $|X|\leqslant k$ and $G-X$ is $H$-minor-free for
every $H\in \mc{F}$?
\journal{
\begin{tabbing}
\hspace{1cm} \= {\sc \planarF{}}\\
\> {\bf Input:} \hspace{1cm} \= \parbox[t]{12cm}{A graph $G$ and a non-negative integer $k$.}\\
\> {\bf Parameter:} \> \parbox[t]{12cm}{The integer $k$.}\\
\> {\bf Question:} \> \parbox[t]{12cm}{Does $G$ have a set $X\subseteq V(G)$
such that $|X|\leqslant k$ and $G-X$ is $H$-minor-free for every $H\in
\mc{F}$?}
\end{tabbing}
Note that \VC and \FVS correspond to the special
cases of $\mc{F}=\{K_2\}$ and $\mc{F}=\{K_3\}$, respectively. A recent work by
Joret \emph{et al}.~\cite{JPSST11} handled the case $\mc{F}=\{\theta_c\}$ and
achieved a single-exponential algorithm for \textsc{Planar-$\theta_c$-Deletion}
for any value of $c \geqslant 1$, where $\theta_c$ is the (multi)graph
consisting of two vertices and $c$ parallel edges between them. (Note that the
cases $c=1$ and $c=2$ correspond to \VC and \textsc{Feedback Vertex Set},
respectively.) Kim \emph{et al}.~\cite{KPP12} obtained a single-exponential
algorithm for $\mc{F}=\{K_4\}$, also known as \textsc{Treewidth-Two Vertex
Deletion}. Related works of Philip \emph{et al}.~\cite{PRV10} and Cygan
\emph{et al}.~\cite{CPP10} resolve the case $\mc{F}=\{K_3,T_2\}$, or
equivalently \textsc{Pathwidth-One Vertex Deletion}, in single-exponential
time.}
\short{Note that \VC and {\sc Feedback Vertex Set} correspond to the special
cases of $\mc{F}=\{K_2\}$ and $\mc{F}=\{K_3\}$, respectively. Recent works have
provided, using quite different techniques, single-exponential algorithms for
the particular cases $\mc{F}=\{K_3,T_2\}$~\cite{CPP10,PRV10},
$\mc{F}=\{\theta_c\}$~\cite{JPSST11}, or $\mc{F}=\{K_4\}$~\cite{KPP12}. The
\planarF{} problem was first stated by Fellows and Langston~\cite{FL88}, who
proposed a non-uniform $f(k) \cdot n^2$-time algorithm for some function
$f(k)$, relying on the meta-theorem of Robertson and Seymour~\cite{RS95c}.
Explicit bounds on the function $f(k)$ can be obtained via dynamic programming.
Indeed, as the \textsc{Yes}-instances of \planarF{} have treewidth $O(k)$,
using standard dynamic programming techniques on graphs of bounded treewidth
(see for instance~\cite{Bod88,ADFST11}), it can be seen that \planarF{} can be
solved in time $f(k) \cdot n^2$ with $f(k) = 2^{2^{O(k \log k)}}$. In a recent
unpublished paper~\cite{FLMS11}, Fomin \emph{et al}. proposed a $2^{O(k\log
k)}\cdot n^2$-time algorithm for \planarF{}, which is, up to our knowledge, the
best known result. More recently this year, Fomin \emph{et
al}.~\cite{FLMS12} improved the running time for \planarF{} to
$2^{O(k)}\cdot n \log^2 n$ under the condition that {\sl every} graph in the
family $\mc{F}$ is {\sl connected}. In this paper we get rid of the
connectivity assumption and we prove that the general
\textsc{Planar-$\mc{F}$-Deletion} problem can be solved in single-exponential
time. Namely, our second main result is the following.}

\journal{The \planarF{} problem was first stated by Fellows and
Langston~\cite{FL88}, who proposed a non-uniform\footnote{A non-uniform \FPT
algorithm for a parameterized problem is a collection of algorithms, one for
each value of the parameter $k$.} (and non-constructive) $f(k) \cdot n^2$-time
algorithm for some function $f(k)$, as well as a $f(k) \cdot n^3$-time
algorithm for the general \textsc{$\mc{F}$-Deletion} problem, both relying on
the meta-theorem of Robertson and Seymour~\cite{RS95c}. Explicit bounds on the
function $f(k)$ for \planarF{} can be obtained via dynamic programming. Indeed,
as the \textsc{Yes}-instances of \planarF{} have treewidth $O(k)$, using
standard dynamic programming techniques on graphs of bounded treewidth (see for
instance~\cite{Bod88,ADFST11}), it can be seen that \planarF{} can be solved in
time $f(k) \cdot n^2$ with $f(k) = 2^{2^{O(k \log k)}}$. In a recent
unpublished paper~\cite{FLMS11}, Fomin \emph{et al}. proposed a $2^{O(k\log
k)}\cdot n^2$-time algorithm for \planarF{}, which is, up to our knowledge, the
best known result. More recently this year, Fomin \emph{et
al}.~\cite{FLMS12} improved the running time for \planarF{} to
$2^{O(k)}\cdot n \log^2 n$ under the condition that {\sl every} graph in the
family $\mc{F}$ is {\sl connected}. In this paper, we get rid of the
connectivity assumption, and we prove that the general
\textsc{Planar-$\mc{F}$-Deletion} problem can be solved in single-exponential
time. Namely, our second main result is the following.}


\begin{mainII}
The parameterized \planarF{} problem can be solved in time $2^{O(k)}\cdot n^2$.
\end{mainII}

\short{This result unifies, generalizes, and simplifies a number of results
given in~\cite{GGH06,DFLRS05,CFLLV08,JPSST11,KPP12,FLMS12}. Besides the fact
that removing the connectivity constraint is an important theoretical step
towards the general case where $\mathcal{F}$ may not contain any planar graph,
it turns out that many natural such families $\mathcal{F}$ do contain
disconnected planar graphs~\cite{Din97}. An important feature of our approach,
in comparison with previous work such as~\cite{JPSST11,KPP12,FLMS12}, is that
our algorithm {\sl does not use any reduction rule}. This is because if
$\mathcal{F}$ contains some disconnected graph,
\Problem{Planar-$\mc{F}$-Deletion} has not FII in general, and then the
protrusion replacement rule cannot be applied. A more in-depth discussion can
be found in the Appendix. Finally, it should also be noted that the function
$2^{O(k)}$ in Theorem~II is best possible assuming the Exponential Time
Hypothesis (ETH), as it is known that \VC cannot be solved in time
$2^{o(k)}\cdot \poly(n)$~\cite[Chapter 16]{FG06} unless the ETH fails.}

\journal{This result unifies, generalizes, and simplifies a number of results
given in~\cite{GGH06,DFLRS05,CFLLV08,JPSST11,KPP12,FLMS12}. Let us make a few
considerations about the fact that the family $\mathcal{F}$ may contain
disconnected graphs or not. Besides the fact that removing the connectivity
constraint is an important theoretical step towards the general
\Problem{$\mc{F}$-Deletion} problem, it turns out that many natural such
families $\mathcal{F}$ do contain disconnected graphs. For instance, the
disjoint union of $g$ copies of $K_5$ (or $K_{3,3}$) is a minimal forbidden
minor for the graphs of genus $g-1$~\cite{BHKY62} (see also~\cite{Moh89}). In
particular, the (disconnected) graph made of two copies of $K_5$ is in the
obstruction set of the graphs that can be embedded in the torus. Let us now see
that many natural obstruction sets also contain disconnected {\sl planar}
graphs. Following Dinneen~\cite{Din97}, given an integer $\ell \geqslant 0$ and
a graph invariant function $\lambda$ that maps graphs to integers such that
whenever $H \preceq_m G$ we also have $\lambda(H) \leqslant \lambda(G)$, we say
that the graph class $\mathcal{G}_{\lambda}^{\ell} := \{G: \lambda(G) \leqslant
\ell \}$ is an \emph{$\ell$-parameterized lower ideal}. By Robertson and
Seymour~\cite{RS95c}, we know that for each $\ell$-parameterized lower ideal
$\mathcal{G}_{\lambda}^{\ell}$ there exists a finite graph family $\mathcal{F}$
such that $\mathcal{G}_{\lambda}^{\ell}$ has precisely $\mathcal{F}$ as (minor)
obstruction set. In this setting, the \Problem{$\mc{F}$-Deletion} problem
(parameterized by $k$) asks whether $k$ vertices can be removed from a graph
$G$ so that the resulting graph belongs to the corresponding
$\ell$-parameterized lower ideal $\mathcal{G}_{\lambda}^{\ell}$. For instance,
the parameterized \FVS problem corresponds to the $0$-parameterized lower ideal
with graph invariant $\mathbf{fvs}$, namely $\mathcal{G}_{\mathbf{fvs}}^{0}$,
which is characterized by $\mathcal{F}=\{K_3\}$ and therefore
$\mathcal{G}_{\mathbf{fvs}}^{0}$ is the set of all forests. Interestingly, it
is proved in~\cite{Din97} that for $\ell \geqslant 1$, the obstruction set of
many interesting graph invariants (such as \Problem{$\ell$-Vertex Cover},
\Problem{$\ell$-Feedback Vertex Set}, or \Problem{$\ell$-Face Cover} to name
just a few) contains the disjoint union of obstructions for $\ell -1$. As for
the above-mentioned problems there is a planar obstruction for $\ell = 0$, we
conclude that for $\ell \geqslant 1$ the corresponding family $\mathcal{F}$
contains {\sl disconnected} planar obstructions.

It should also be noted that the function $2^{O(k)}$ in Theorem~II is best
possible, assuming the Exponential Time Hypothesis (ETH). Namely, it is known
that unless the ETH fails, \VC cannot be solved in time $2^{o(k)}\cdot
\poly(n)$~\cite[Chapter 16]{FG06}. It is noteworthy that the class of graphs in
Theorem~II, in some sense, the best achievable one with respect to the
state-of-the-art. When $\mc{F}$ does not contain any planar graph, up to our
knowledge no single case is known to admit a single-exponential algorithm. For
instance, we point out that \Problem{Planar Vertex Deletion}, which amounts to
$\mc{F}=\{K_5,K_{3,3}\}$, is not known to have a single-exponential
parameterized algorithm~\cite{MS12}, while a double-exponential function $f(k)$
is the best known so far~\cite{Kaw09}.

Let us now discuss some important ingredients of our approach to prove
Theorem~II. As mentioned above, when employing protrusion replacement, often
the problem needs to have FII. Many problems enjoy this property, for example
\Problem{Treewidth-$t$ Vertex Deletion} or \Problem{(Connected) Dominating
Set}, among others. Having FII makes the problem amenable to this powerful
reduction rule, and essentially this was the basic ingredient of previous works
such as~\cite{JPSST11,KPP12,FLMS12}. In particular, when every graph in
$\mc{F}$ is connected, the \Problem{Planar-$\mc{F}$-Deletion} problem has
FII~\cite{BFLPST09}, and the single-exponential time algorithm
of~\cite{FLMS12} heavily depends on this feature. However, if one aims at
\Problem{Planar-$\mc{F}$-Deletion} without any connectivity restriction on the
family $\mc{F}$, the requirement for FII seems to be a fundamental hurdle, as
if $\mc{F}$ contains some disconnected graph, then
\Problem{Planar-$\mc{F}$-Deletion} has not FII in general\footnote{As we were
not able to find a reference with a proof of this fact, for completeness we
provide it in Appendix~\ref{ap:notFII}.}. We observe that the unpublished
$2^{O(k\log k)}\cdot n^2$-time algorithm of~\cite{FLMS11} applies to the
general \Problem{Planar-$\mc{F}$-Deletion} problem (that is, $\mathcal{F}$ may
contain some disconnected graph). The reason is that instead of relying on FII,
they rather use tools from \emph{annotated kernelization}~\cite{BFLPST09}.

To circumvent the situation of not having FII, our algorithm {\sl does not use
any reduction rule}, but instead relies on a series of branching steps.
First of all, we apply the iterative compression technique (introduced by Reed
\emph{et al}.~\cite{RSV04}) in order to reduce the
\Problem{Planar-$\mc{F}$-Deletion} problem to its {\sl disjoint} version. In the
\Problem{Disjoint Planar-$\mc{F}$-Deletion} problem, given a graph $G$ and an
initial solution $X$ of size $k$, the task is to decide whether $G$ contains an
alternative solution $\tilde{X}$ disjoint from $X$ of size at most $k-1$.
In our case, the assumption that $\mathcal{F}$ contains some planar graph is
fundamental, as then $G-X$ has bounded treewidth~\cite{RS86a}. Central to our
single-exponential algorithm is our linear-time algorithm to compute a
protrusion decomposition, in this case with the initial solution $X$ as
treewidth-modulator. But for the resulting protrusion decomposition to be
linear, it turns out that we first need to guess the intersection of the
alternative solution with the set $Y_0$. Once we have the desired linear
protrusion decomposition, instead of applying protrusion replacement, we simply
identify a set of $O(k)$ vertices among which the alternative solution has to
live, if it exists. In the whole process described above, there are three
branching steps: the first one is inherent to the iterative compression
paradigm, the second one is required to compute a linear
protrusion-decomposition, and finally the last one enables us to guess the set
of vertices containing the solution. It can be proved that each branching step
is compatible with single-exponential time, which yields the desired result.
Finally, it is worth mentioning that our algorithm is {\sl fully constructive}
(cf.~Section~\ref{sec:colorvec} for details).}


\journal{\paragraph{Organization of the paper.} In
Section~\ref{sec:Preliminaries}, we outline all important definitions that are
relevant to this work. We then exhibit our protrusion decomposition algorithm
in Section~\ref{sec:Decomposition}. As our first application of our
decomposition result, we prove Theorem~I in Section~\ref{sec:Kernels}. In
Section~\ref{sec:PlanarF} we prove Theorem~II. Finally, in
Section~\ref{sec:Conclusion} we conclude with some closing remarks.}

\section{Preliminaries}\label{sec:Preliminaries}
We use standard graph-theoretic notation (see~\cite{Die10}\short{ and the
Appendix} for any undefined terminology). Given a graph $G$, we let $V(G)$
denote its vertex set and $E(G)$ its edge set. For convenience we assume that
$V(G)$ is a totally ordered set.
\journal{%
The neighborhood of a vertex $x \in V(G)$ is the set of all vertices~$y \in
V(G)$ such that~$xy \in E(G)$ and is denoted by~$N^G(x)$. The closed
neighborhood of~$x$ is defined as $N^G[x] := N(x) \cup \{x\}$. The distance
$d_G(x,y)$ of two vertices $x,y \in V(G)$ is the length (number of edges) of a
shortest $x,y$-path in $G$ and $\infty$ if $x,y$ lie in different connected
components of $G$. The \emph{$r$th neighborhood} of a vertex $N^G_r(v) :=
\set{w \in G \mid d_G(v,w) \leq r}$ is the set of vertices within distance at
most $r$ to $v$, in particular we have that $N^G_0(v) = \set{v}$ and $N^G_1(v)
= N^G(v)$. Since we will mainly be concerned with sparse graphs in this paper,
we let $|G|$ denote the number of vertices in the graph $G$. Subscripts and
superscripts are omitted if it is clear which graph is being referred to.
For~$X \subseteq V(G)$, we let~$G[X]$ denote the graph $(X,E_X)$, where
$E_X := \set{xy \mid x,y \in X \ \text{and} \ xy \in E(G)}$, and we define $G-X:=G[V(G)
\setminus X]$.

By the \emph{neighbors of a subgraph} $H \subseteq G$, denoted $N^G(H)$,
we mean the set of vertices in $V(G) \setminus V(H)$
that have at least one neighbor in $H$. We employ the same notation
analogously to denote \emph{neighbors of a subset of vertices} $N^G(S)$ for
$S \subseteq V(G)$. If $X$ is a subset of
vertices disjoint from $S$, then $N_X^G(S)$ is the set $N^G(S)\cap X$. The same
notation naturally extends to a subgraph $H \subseteq G$, that is, $N_X^G(H)$.
(When the graph $G$ is clear from the context, we may drop it from the
notation.) We denote by $\omega(G)$ the size of the largest complete subgraph
of $G$ and by $\#\omega(G)$ the \emph{number} of complete subgraphs.
}%
Given an edge~$e = xy$ of a graph~$G$, we let~$G/e$ denote the graph obtained
from~$G$ by \emph{contracting} the edge~$e$, which amounts to deleting the
endpoints of~$e$, introducing a new vertex~$v_{xy}$, and making it adjacent to
all vertices in $(N^G(x) \cup N^G(y)) \setminus \{x,y\} $. A \emph{minor}
of~$G$ is a graph obtained from a subgraph of~$G$ by contracting zero or more
edges. If~$H$ is a minor of~$G$, we write~$H \minor G$. A graph $G$ is {\em
$H$-minor-free} if $H \notminor G$. A \emph{topological minor} of~$G$ is a
graph obtained from a subgraph of~$G$ by contracting zero or more edges, such
that each edge that is contracted  has at least one endpoint with degree at
most two. We write~$H \tminor G$ to denote that $H$ is a topological minor of
$G$. Note that~$H \tminor G$ implies that $H \minor G$, but not vice-versa. A
graph $G$ is \emph{$H$-topological-minor-free} if $H \nottminor G$.

\subsection{Parameterized problems, kernels and treewidth\short{.}}
A parameterized problem~$\Pi$ is a subset of $\Gamma^{*} \times \mbb{N}_0$,
where $\Gamma$ is some finite alphabet. An instance of a parameterized problem
is a tuple $(x,k)$, where~$k$ is the parameter. %
\short{%
    A \emph{parameterized graph problem}~$\Pi$ is a set
    of tuples $(G,k)$, where $G$ is a graph and $k \in \mathbb{N}_0$,
    such that if $(G,k) \in \Pi$ then for all $G' \cong G$, $(G',k) \in \Pi$.
    If $\mathcal G$ is a graph class, we define $\Pi$ \emph{restricted to}
    $\mathcal G$ as $\Pi_{\mathcal G} = \set{(G,k) \mid (G,k) \in
    \Pi ~\textnormal{and}~ G \in \mathcal G}.$
}%
\journal{%
\begin{definition}[Parameterized graph problem]
    A \emph{parameterized graph problem}~$\Pi$ is a set
    $\{(G,k) \mid G \ \textnormal{is a graph and} \ k \in \mathbb{N}_0\}$
    such that for all graphs $G_1, G_2$ and all $k \in \mathbb{N}_0$,
    if $G_1 \cong G_2$ then $(G_1,k) \in \Pi~\textnormal{iff}~(G_2,k) \in \Pi$.
    If $\mathcal G$ is a graph class, we define $\Pi$ \emph{restricted to}
    $\mathcal G$ as
    $
        \Pi_{\mathcal G} = \set{(G,k) \mid (G,k) \in \Pi ~\textnormal{and}~ G \in \mathcal G}.
    $
\end{definition}
}%
A parameterized problem~$\Pi$ is \emph{fixed-parameter tractable} if there
exists an algorithm that decides instances~$(x,k)$ in time $f(k) \cdot
\poly(|x|)$, where~$f$ is a function of~$k$ alone. \short{ For the formal
definition of \emph{kernelization} see the Appendix.
}\journal{The notion of kernelization is defined as follows.
\begin{definition}[Kernelization]
    A \emph{kernelization algorithm}, or just \emph{kernel}, for a parameterized problem
    $\Pi \subseteq \Gamma^{*} \times \mathbb{N}_0$
    is an algorithm that given $(x,k) \in \Gamma^{*} \times \mathbb{N}_0$ outputs,
    in time polynomial in $|x| + k$, an instance $(x',k') \in \Gamma^{*} \times \mathbb{N}_0$
    such that:
    \begin{enumerate}
        \item $(x,k) \in \Pi$ if and only if $(x',k') \in \Pi$;
        \item $|x'|, k' \leq g(k)$,
    \end{enumerate}
    where~$g$ is some computable function. The function $g$ is called
    the \emph{size} of the kernel. If $g(k) = k^{O(1)}$ or $g(k) = O(k)$,
    we say that $\Pi$ admits a \emph{polynomial kernel} and a \emph{linear kernel}, respectively.
\end{definition}

\begin{definition}[Treewidth]
    Given a graph~$G=(V,E)$, a \emph{tree-decomposition of $G$} is an ordered pair
    $(T, \{ W_x \mid x \in V(T) \})$, where~$T$ is a tree and $\{W_x \mid x \in V(T)\}$ is a collection
    of vertex sets of~$G$, with one set for each node of the tree~$T$ such that the following hold:
    \begin{enumerate}
        \item $\bigcup_{x \in V(T)} W_x = V(G)$;
        \item for every edge~$e = uv$ in~$G$, there exists~$x \in V(T)$ such that~$u,v \in W_x$;
        \item for each vertex~$u \in V(G)$, the set of nodes~$\{x \in V(T) \mid u \in W_x\}$
        induces a subtree.
    \end{enumerate}
    The vertices of the tree~$T$ are usually referred to as \emph{nodes} and the sets~$W_x$
    are called \emph{bags}. The \emph{width} of a tree-decomposition is the size of a largest
    bag minus one. The \emph{treewidth} of~$G$, denoted~$\tw(G)$, is the smallest width of a
    tree-decomposition of~$G$.
\end{definition}
}

\short{ Given a graph~$G=(V,E)$, we denote a \emph{tree-decomposition of $G$}
by $(T, \{ W_x \mid x \in V(T) \})$, where~$T$ is a tree and $\{W_x \mid x \in
V(T)\}$ are the bags of the decomposition.} Given a bag $B$ of a
tree-decomposition with tree $T$, we denote by $T_B$ the subtree rooted
at the node corresponding to bag $B$, and by $G_B := G[\bigcup_{x\in T_B}W_x]$
the subgraph of $G$ induced by the vertices appearing in the bags corresponding
to the nodes of $T_B$. A \emph{join bag} $B$ of a rooted tree-decomposition is
a bag such that the root of $T_B$ has degree at least two. If a graph $G$ is
disconnected, a \emph{forest-decomposition} of $G$ is the union of
tree-decompositions of its connected components. We refer the reader to
Diestel's book~\cite{Die10} for an introduction to the theory of treewidth. For
the definition of {\em nice tree-decomposition}, we refer the readers
to~\cite{Klo94}.

\short{
In this paper, we will refer to problems that are definable in (Counting)
Monadic Second Order Logic (CMSO). We refer
to~\cite{FG06,CE12} for a more detailed presentation on (C)MSO logic.
In a \pmincmso{} graph problem (respectively, \pmaxcmso{} or \peqcmso{}) $\Pi$,
one has to decide the existence of a set $S$ of at most $k$ vertices/edges
(respectively, at least $k$ or exactly $k$) in an input graph $G$ such that the
CMSO expressible predicate $P_{\Pi}(G,S)$ is satisfied.

}

\journal{
\subsection{(Counting) Monadic Second Order Logic\short{.}}

\def\taugraph{\tau_{\it Graph}}

Monadic Second Order Logic (MSO) is an extension of First Order Logic that
allows quantification over sets of objects.
We identify graphs with relational structures over a vocabulary
$\taugraph$, consisting of the unary relation symbols \textsc{Vert} and \textsc{Edge}
and the binary relation symbol \textsc{Inc}.
A graph $G=(V,E)$ is then represented by a $\taugraph$-structure $\mc{G}$ with
universe $U(\mc{G}) =V\cup E$ such that:
\begin{itemize}\setlength{\itemsep}{-.0pt}
\item $\textsc{Vert}^{\mc{G}}=V$ and  $\textsc{Edge}^{\mc{G}}=E$ represent
  the vertex set and the edge  set, respectively, and
\item  $\textsc{Inc}^{\mc{G}}=\{\,(v,e)\mid v\in V,
e\in E \mbox{ and } v \mbox{ is incident to } e\,\}$ represents the incidence
relation.
\end{itemize}
A \emph{Monadic Second Order} formula contains two types of variables:
\emph{individual variables} to be used for elements of the universe, usually
denoted by lowercase letters $x,y,z,\ldots$ and \emph{set variables}
to be used for subsets of the universe, usually denoted by uppercase letters
$X,Y,Z,\ldots$.
Atomic formulas on $\taugraph$ are:
$x=y$, $x \in X$,  $x \in \textsc{Vert}$, $x \in \textsc{Edge}$, and
$\textsc{Inc}(x, y)$ for all individual variables $x,y$ and set variables $X$.
MSO formulas on $\taugraph$ are
built from the atomic formulas using Boolean connectives $\neg,\wedge,\vee$,
and quantification $\exists x, \forall x, \forall X, \forall Y$ for
individual variables~$x$ and set variables~$X$.
MSO formulas are interpreted in $\taugraph$-structures in the natural way,
e.g., $\textsc{Inc}(x,y)$ being true iff in~$G$ the vertex $v$ represented by $x$
is incident to the edge $e$ represented by $y$.\journal{\\}

In a \emph{Counting} Monadic Second Order (CMSO) formula, we have additional
atomic formulas $\mb{card}_{n,p}(X)$ on set variables~$X$, which are true if the
set $U$ represented by the variable $X$ has size $n\pmod p$. We refer
to~\cite{FG06,CE12} for a more detailed presentation on (C)MSO logic.
In a \pmincmso{} graph problem (respectively, \pmaxcmso{} or \peqcmso{}) $\Pi$,
one has to decide the existence of a set $S$ of at most $k$ vertices/edges
(respectively, at least $k$ or exactly $k$) in an input graph $G$ such that the
CMSO expressible predicate $P_{\Pi}(G,S)$ is satisfied.
}

\subsection{Protrusions, $t$-boundaried graphs, and finite integer index\short{.}}

We restate the main definitions of the
protrusion machinery developed in~\cite{BFLPST09,FLST10}.
Given a graph~$G=(V,E)$ and a set~$W \subseteq V$, we define~$\partial_G(W)$
as the set of vertices in~$W$ that have a neighbor in~$V \setminus W$. For a
set~$W \subseteq V$ the neighborhood of~$W$ is
$N^G(W) = \partial_G(V \setminus W)$.
Subscripts are omitted when it is clear which graph is being referred to.

\begin{definition}[$t$-protrusion~\cite{BFLPST09}]
    Given a graph~$G$, a set~$W \subseteq V(G)$ is a \emph{$t$-protrusion}
    of~$G$ if $|\partial_G (W)| \leq t$ and $\tw(G[W]) \le t-1$.\footnote{
    In~\cite{BFLPST09}, $\tw(G[W]) \le t$, but we want the size of the bags to be at most~$t$.}
    If~$W$ is a $t$-protrusion, the vertex set~$W' = W \setminus \partial_G (W)$ is the
    \emph{restricted protrusion of~$W$}.
    We call~$\partial_G (W)$ the \emph{boundary} and $|W|$ the \emph{size}
    of the $t$-protrusion $W$ of $G$. Given a restricted $t$-protrusion $W'$, we
    denote its \emph{extended protrusion} by $W'^+ = W' \cup N(W') = W$.
\end{definition}
\journal{
\noindent A rough outline of a protrusion is depicted in Figure~\ref{fig:Anatomy}.

\begin{figure}[t]
    \center\includegraphics[width=.6\textwidth]{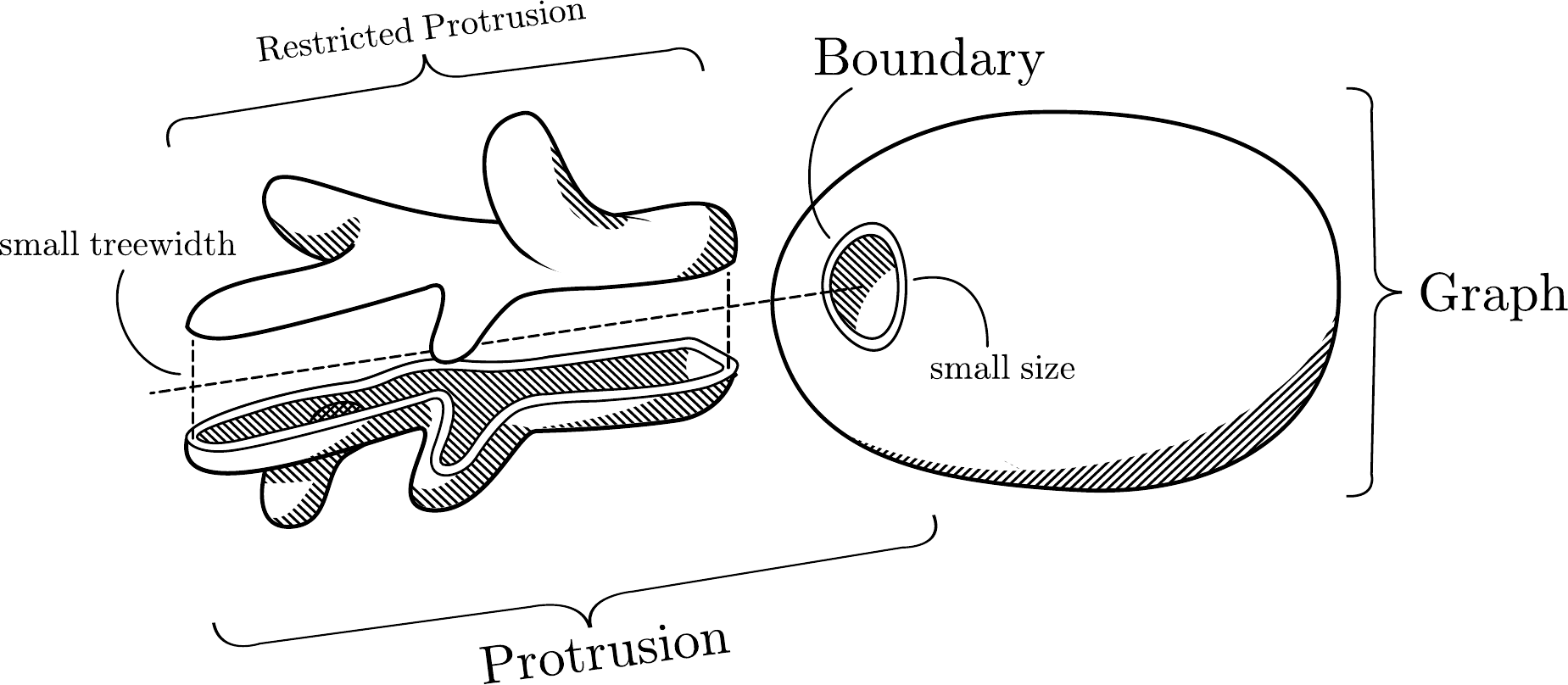}
    \caption{\label{fig:Anatomy} Basic anatomy of a protrusion.}
\end{figure}
}

A \emph{$t$-boundaried graph} is a graph~$G=(V,E)$ with a set~$\bound(G)$
(called the \emph{boundary}\footnote{Usually denoted by $\partial(G)$, but this
collides with our usage of $\partial$.} or the \emph{terminals} of~$G$)
of~$t$ distinguished vertices labeled~$1$ through~$t$.
Let $\mathcal G_{t}$ denote the class of $t$-boundaried graphs, with graphs from~$\mathcal G$.
If~$W \subseteq V$ is an $r$-protrusion in $G$, then we let~$G_W$ be the $r$-boundaried
graph $G[W]$ with boundary~$\partial_G (W)$, where the vertices of~$\partial_G (W)$ are
assigned labels~$1$ through~$r$ according to their order in~$G$.

\short{
    \emph{Gluing} two $t$-boundaried graphs $G_1$ and $G_2$
    creates the graph~$G_1 \oplus G_2$ obtained by taking
    the disjoint union of~$G_1$ and~$G_2$ and identifying
    vertices in~$\bound(G_1)$ and~$\bound(G_2)$ with the same labels.
    If~$G_1$ is a subgraph of~$G$ with a $t$-boundary $\bound(G_1)$,
    \emph{ungluing  $G_1$ from~$G$} creates the
    $t$-boundaried graph $G \ominus G_1 = G - (V(G_1) \setminus \bound(G_1))$
    with boundary $\bound(G\ominus G_1) =  \bound(G_1)$,
    the vertices of which are assigned labels according to their order in the graph~$G$.
    Let~$W$ be a $t$-protrusion in~$G$, let~$G_W$ denote
    the graph~$G[W]$ with boundary $\bound(G_W) = \partial_G(W)$, and
    let~$G_1$ be a $t$-boundaried graph.
    Then \emph{replacing~$G_W$ by~$G_1$} corresponds to the operation
    $
        (G \ominus G_W) \oplus G_1
    $.

}

\journal{
\begin{definition}[Gluing and ungluing]
    For $t$-boundaried graphs $G_1$ and $G_2$, we let~$G_1 \oplus G_2$ denote the
    graph obtained by taking the disjoint union of~$G_1$ and~$G_2$ and identifying
    each vertex in~$\bound(G_1)$ with the vertex in~$\bound(G_2)$ with the same label.
    This operation is called \emph{gluing}.

    Let~$G_1$ be a subgraph of a graph~$G$ and suppose that~$G_1$ has a
        boundary~$\bound(G_1)$ of size~$t$.
        The operation of \emph{ungluing} $G_1$ from~$G$ creates a $t$-boundaried graph,
        denoted by $G \ominus G_1$, and defined as follows:
    \begin{eqnarray*}
        G \ominus G_1        & = & G - (V(G_1) \setminus \bound(G_1)), \\
        \bound(G\ominus G_1) & = & \bound(G_1).
    \end{eqnarray*}
    The vertices of $\bound(G \ominus G_1)$ are assigned labels~$1$ through~$t$
    according to their order in the graph~$G$.
\end{definition}

\begin{definition}[Replacement]
    Let~$G=(V,E)$ be a graph with a $t$-protrusion $W$; let~$G_W$ denote
        the graph~$G[W]$ with boundary $\bound(G_W) = \partial_G(W)$; and
    finally, let~$G_1$ be a $t$-boundaried graph.
    Then \emph{replacing}~$G_W$ by~$G_1$ corresponds to the operation
    $
        (G \ominus G_W) \oplus G_1
    $.
\end{definition}

\begin{definition}[Protrusion decomposition]
    An $(\alpha,t)${\em -protrusion decomposition} of a graph $G$ is a partition
    ${\cal P}=Y_{0}\uplus Y_{1}\uplus \cdots\uplus Y_{\ell}$ of $V(G)$ such
    that:
    \begin{enumerate}
    \item for every $1\leqslant i\leqslant \ell$, $N(Y_{i})\subseteq Y_{0}$;
    \item $\max\{\ell, |Y_{0}|\}\leqslant \alpha$;
    \item for every $1\leqslant i\leqslant \ell$, $Y_i\cup N_{Y_0}(Y_i)$  is a $t$-protrusion of $G$.
    \end{enumerate}
    The set $Y_0$ is called the \emph{separating part} of $\mc{P}$.
\end{definition}
}

\short{
\begin{definition}[Protrusion decomposition]
    An $(\alpha,t)${\em -protrusion decomposition} of a graph $G$ is a partition
    ${\cal P}=Y_{0}\uplus Y_{1}\uplus \cdots\uplus Y_{\ell}$ of $V(G)$ such
    that: $(1)$~for every $1\leqslant i\leqslant \ell$, $N(Y_{i})\subseteq Y_{0}$;
    $(2)$~$\max\{\ell, |Y_{0}|\}\leqslant \alpha$;
    $(3)$~for every $1\leqslant i\leqslant \ell$, $Y_i\cup N_{Y_0}(Y_i)$  is a $t$-protrusion of $G$.
    The set $Y_0$ is called the \emph{separating part} of $\mc{P}$.
\end{definition}

}

Hereafter, the value of $t$ will be fixed to some constant. When $G$ is the
input of a parameterized graph problem with parameter $k$, we say that an
$(\alpha,t)$-protrusion decomposition of $G$ is \emph{linear} (resp.
\emph{quadratic}) whenever $\alpha =O(k)$ (resp. $\alpha = O(k^2)$).

\short{
We now restate the definition of one of the most important notions used in this paper.
\begin{definition}[Finite integer index (FII)~\cite{BvF01}] \label{def:finiteii}
    Let~$\Pi_{\mathcal G}$ be a parameterized graph problem restricted to a class~$\mathcal G$ and let
    $G_1, G_2$ be two $t$-boundaried graphs in~$\mathcal{G}_t$. We say that~$\equipi{G_1}{G_2}{t}$
    if there exists a constant~$\Delta_{\Pi,t}(G_1,G_2)$  (that depends on $\Pi$, $t$, and
        the ordered pair
        $(G_1, G_2)$) such that for all $t$-boundaried graphs $G_3$ and for all~$k$:
        $(1)$~$G_1 \oplus G_3 \in \mc G$ iff $G_2 \oplus G_3 \in \mc G$;
        $(2)$~$(G_1 \oplus G_3, k) \in \Pi$ iff $(G_2 \oplus G_3, k + \Delta_{\Pi,t}(G_1,G_2)) \in \Pi$.
        We say that the problem~$\Pi_{\mathcal G}$ has \emph{finite integer index in the class~$\mathcal G$}
    iff for every integer~$t$, the equivalence relation~$\equipi{}{}{t}$ has finite index.
        In the case that $(G_1 \oplus G, k) \not \in \Pi$ or
    $G_1 \oplus G \not \in \mc G$ for \emph{all} $G \in \mc G_t$, we set
    $\Delta_{\Pi,t}(G_1,G_2) = 0$. Note that $\Delta_{\Pi,t}(G_1,G_2) = -\Delta_{\Pi,t}(G_2,G_1)$.
\end{definition}
}

\journal{
We now restate the definition of one of the most important notions used in this paper.
\begin{definition}[Finite integer index (FII)~\cite{BvF01}] \label{def:finiteii}
    Let~$\Pi_{\mathcal G}$ be a parameterized graph problem restricted to a class~$\mathcal G$ and let
    $G_1, G_2$ be two $t$-boundaried graphs in~$\mathcal{G}_t$. We say that~$\equipi{G_1}{G_2}{t}$
    if there exists a constant~$\Delta_{\Pi,t}(G_1,G_2)$  (that depends on $\Pi$, $t$, and
        the ordered pair
        $(G_1, G_2)$) such that for all $t$-boundaried graphs $G_3$ and for all~$k$:
    \begin{enumerate}
        \item $G_1 \oplus G_3 \in \mc G$ iff $G_2 \oplus G_3 \in \mc G$;
        \item $(G_1 \oplus G_3, k) \in \Pi$ iff $(G_2 \oplus G_3, k + \Delta_{\Pi,t}(G_1,G_2)) \in \Pi$.
    \end{enumerate}
        We say that the problem~$\Pi_{\mathcal G}$ has \emph{finite integer index in the class~$\mathcal G$}
    iff for every integer~$t$, the equivalence relation~$\equipi{}{}{t}$ has finite index.
        In the case that $(G_1 \oplus G, k) \not \in \Pi$ or
    $G_1 \oplus G \not \in \mc G$ for \emph{all} $G \in \mc G_t$, we set
    $\Delta_{\Pi,t}(G_1,G_2) = 0$. Note that $\Delta_{\Pi,t}(G_1,G_2) = -\Delta_{\Pi,t}(G_2,G_1)$.
\end{definition}

If a parameterized problem has finite integer index then its instances can be
reduced by ``replacing protrusions''. The technique of replacing protrusions
hinges on the fact that each protrusion of ``large'' size can be replaced by a
``small'' gadget from the same equivalence class as the protrusion, which
consequently behaves similarly \wrt the problem at hand. If~$G_1$ is replaced
by a gadget~$G_2$, then the parameter~$k$ in the problem changes
by~$\Delta_{\Pi,t}(G_1,G_2)$. What is not immediately clear is that given that
a problem~$\Pi$ has finite integer index, how does one show that there always
exists a set of representatives for which the parameter is guaranteed not to
\emph{increase}. The next lemma shows that this is indeed the case. } \short{
If a parameterized problem has FII then it can be reduced by ``replacing
protrusions'', hinging on the fact that each ``large'' protrusion can be
replaced by a ``small'' gadget from the same equivalence class that behaves
similar \wrt to the problem at hand. Exchanging~$G_1$ by a gadget~$G_2$ changes
the parameter~$k$ by~$\Delta_{\Pi,t}(G_1,G_2)$.
Lemma~\ref{lemma:representatives} guarantees the existence of a set of
representatives such that the replacement operation does \emph{not} increase
the parameter. In the Appendix we show how to find protrusions in polynomial
time. We also show how to identify by which representative to replace a protrusion,
assuming that we are \emph{given} the set of representatives, an assumption
we  make from now on. This makes our algorithms in
Section~\ref{sec:Kernels} \emph{non-uniform}, as those in previous
work~\cite{BFLPST09,FLST10,FLMPS11,FLMS12}.
\begin{lemma}\label{lemma:representatives}\omitted\short{\footnote{
    The proofs of the results marked with `$[\star]$' are deferred to the full version in the Appendix.}}
    Let $\Pi$ be a parameterized graph problem that has \fii
    in a graph class $\mathcal G$. Then for every fixed~$t$,
    there exists a finite set $\mathcal{R}_t$ of $t$-boundaried graphs such that
    for each $t$-boundaried graph $G \in \mathcal{G}_t$ there exists a $t$-boundaried
    graph $G' \in \mathcal{R}_t$ such that $G \equipi{}{}{t} G'$ and $\Delta_{\Pi,t}(G,G') \geq 0$.
\end{lemma}

\begin{definition}[Protrusion limit]
    For a parameterized graph problem $\Pi$ that has \fii in the
    class $\mathcal{G}$, let~$\mathcal{R}_t$ denote the set of
        representatives of the equivalence classes of $\equiv_{\Pi,t}$
        such that using them to replace protrusions does not increase the parameter.
        The \emph{protrusion limit} of~$\Pi_{\mathcal{G}}$ is a function
        $\prot_{\Pi_{\mathcal{G}}} \colon \mbb{N} \rightarrow \mathbb{N}$
    defined as $\prot_{\Pi_{\mathcal{G}}}(t) = \max_{G \in \mathcal{R}_t} |V(G)|$.
    We drop the subscript when it is clear which graph problem is being referred to.
    We also define $\protd(t) := \prot(2t)$.
\end{definition}

\begin{lemma}[\hspace*{-1ex}{\cite{BFLPST09}}]\label{lemma:constantProtrusions}\omitted
    Let~$\Pi$ be a parameterized graph problem with \fii in~$\mathcal{G}$
    and let $t \in \mathbb{N}$ be a constant.
    For a graph $G \in \mathcal{G}$, if one is given a $t$-protrusion $X \subseteq V(G)$ such that
    $\protd_{\Pi_{\mathcal{G}}}(t) < |X|$, then one can, in time $O(|X|)$, find a $2t$-protrusion $W$
    such that $\protd_{\Pi_{\mathcal{G}}}(t) < |W| \leq 2 \cdot \protd_{\Pi_{\mathcal{G}}}(t)$.
\end{lemma}

}

\journal{
\begin{lemma}\label{lemma:representatives}
    Let $\Pi$ be a parameterized graph problem that has \fii
    in a graph class $\mathcal G$. Then for every fixed~$t$,
    there exists a finite set $\mathcal{R}_t$ of $t$-boundaried graphs such that
    for each $t$-boundaried graph $G \in \mathcal{G}_t$ there exists a $t$-boundaried
    graph $G' \in \mathcal{R}_t$ such that $G \equipi{}{}{t} G'$ and $\Delta_{\Pi,t}(G,G') \geq 0$.
\end{lemma}
\begin{proof}
    The set $\mathcal{R}_t$ consists of one element from each equivalence class
    of $\equiv_{\Pi,t}$. Since~$\Pi$ has finite integer index, the set~$\mathcal{R}_t$
    is finite. Therefore we only have to show that there exist representatives that
        satisfy the requirement in the statement of the lemma.

    To this end, fix any equivalence class
    $\mathcal G_t' \in \mathcal G_t/\!\!\equiv_{\Pi,t}$.
    First consider the case where there exists $G_1 \in \mathcal G_t'$ such that
        for all $G \in \mathcal G_t$,
    either $G_1 \oplus G \not \in \mc G$ or for all $k \in \mathbb{N}_0$,
    $(G_1 \oplus G, k) \not \in \Pi$.
    Since~$\mathcal G_t'$ is an equivalence class, this means that at least one
    of these two conditions holds for \emph{every} graph $G \in \mathcal G_t'$.
    Thus $\Delta_{\Pi,t}(G_1,G_2) = 0$ for all $t$-boundaried graphs $G_1,G_2 \in \mathcal G_t'$
    and we can simply take a graph of smallest size from $\mathcal G_t'$ as representative.

    We can now assume that for the chosen $\mc G_t'$ it holds that there
    exists a $t$-boundaried graph $G \in \mathcal G_t$ such that for
    \emph{all} $G_1 \in \mathcal G_t'$ we have that $G_1 \oplus G \in \mc G$
    and, for some $k \in \mathbb{N}$, $(G_1 \oplus G, k) \in \Pi_{\mathcal{G}}$.
    Consider the following binary relation $\preceq$ over $\mathcal G_t'$:
    for all $G_1, G_2 \in \mathcal G_t'$,
    $$
        G_1 \preceq G_2 \Leftrightarrow \Delta_{\Pi,t}(G_1,G_2) \geq 0.
    $$
    As $\Delta_{\Pi,t}(G,G) = 0$ for all $G \in \mathcal G_t$, it
    immediately follows that the relation $\preceq$ is reflexive. Furthermore, the
    relation is \emph{total} as every graph is comparable to every other graph
    from the same equivalence class.

    We next show that the relation $\preceq$ is also transitive, making
    it a total quasi-order. Let $G_1,G_2,G_3 \in \mathcal G_t'$
    be such that $G_1 \preceq G_2$ and $G_2 \preceq G_3$.
    This is equivalent to saying that $c_{12} = \Delta_{\Pi,t}(G_1,G_2) \geq 0$
    and $c_{23} = \Delta_{\Pi,t}(G_2,G_3) \geq 0$.
    For every $G \in \mathcal{G}_t$ such that $G_1 \oplus G \in \mathcal G$
    and $(G_1 \oplus G, k) \in \Pi$ for some $k \in \mathbb{N}$, we have
        \begin{eqnarray*}
        (G_1 \oplus G, k) \in \Pi & \Leftrightarrow & (G_2 \oplus G, k+c_{12}) \in \Pi \\
                      & \Leftrightarrow & (G_3 \oplus G, k+c_{12}+c_{23}) \in \Pi.
    \end{eqnarray*}
    By definition, $\Delta_{\Pi,t}(G_1,G_3) = c_{12} + c_{23} \geq 0$ and hence~$G_1 \preceq G_3$.
    We conclude that $\preceq$ is transitive and therefore a total quasi-order.

    We now show that the class~$\mathcal{G}_t'$ can be partitioned into layers that can
    be linearly ordered. We will pick our representative for the class $\mathcal{G}_t'$
    from the first layer in this ordering. To do this, we define the following
    equivalence relation over $\mathcal G_t'$. For all $G_1, G_2 \in \mathcal G_t'$,
    define
    \begin{eqnarray*}
        G_1 \equiv G_2 & \Leftrightarrow & G_1 \preceq G_2 ~\textnormal{and}~ G_2 \preceq G_1 \\
                   & \Leftrightarrow & \Delta_{\Pi,t}(G_1,G_2) = 0.
    \end{eqnarray*}
    Now, the equivalence classes $\mathcal G_t'/ \! \! \equiv$
    can be linearly ordered as follows. Fix a graph $G \in \mc G_t$ such that for
    any $G_1 \in \mathcal G_t'$ we have that $G_1 \oplus G \in \mc G$ and
    $(G_1 \oplus G, k) \in \Pi$ for some~$k \in \mathbb{N}$, this graph
    must exist since we handled equivalence classes of $\mc G_t/\!\!\equiv_{\Pi,t}$
    which do not have such a graph in the first part of the proof. Consider the
        function $\Phi_G \colon \mathcal G_t' / \! \! \equiv \rightarrow \mathbb{N}_0$
    defined via
    $$
        \Phi_G([G']) = \min \set{k \in \mbb N \mid (G' \oplus G,k) \in \Pi}.
    $$
    Observe that $\Phi_G([G_2]) = \Phi_G([G_1]) + \Delta_{\Pi,t}(G_1,G_2)$ for all
    $G_1,G_2 \in \mathcal G_t'$ and, in particular, that
    $$
        \Phi_G([G_1]) = \Phi_G([G_2]) \Leftrightarrow G_1 \equiv G_2.
    $$
    Thus $\Phi_G$ induces a linear order on $\mathcal G_t' / \! \! \equiv$. Moreover, since
    $\Phi_G(\cdot) \geq 0$, there exists a class $[G^*]$ in $\mathcal{G}_t'/ \! \! \equiv$
    that is a minimum element in the order induced by~$\Phi_G$.
    For any $t$-boundaried graph $G \in [G^*]$, it then follows that for all $G_1 \in \mathcal{G}_t'$,
    $\Delta_{\Pi,t}(G,G_1) \geq 0$. The representative of $\mathcal{G}_t'$ in $\mathcal{R}_t$
    is an arbitrary $t$-boundaried graph $G' \in [G^*]$ of smallest size. This proves
    the lemma.
\end{proof}

We now show that the protrusion reduction rule is safe.
\begin{lemma}[Safety]
        Let~$\mathcal G$ be a graph class and let~$\Pi_{\mathcal{G}}$ be a
        parameterized graph problem with finite integer index \wrt~$\mathcal{G}$.
        If~$(G',k')$ is the instance obtained from one application of the protrusion
        reduction rule to the instance~$(G,k)$ of~$\Pi_\mathcal{G}$, then
    \begin{enumerate}
        \item $G' \in \mathcal{G}$;
                \item $(G',k')$ is a \textsc{Yes}-instance iff
            $(G,k)$ is a \textsc{Yes}-instance; and
        \item $k' \leq k$.
    \end{enumerate}
\end{lemma}
\begin{proof}
    Suppose that~$(G',k')$ is obtained from~$(G,k)$ by replacing a $2t$-boundaried
    subgraph~$G_W$ (induced by a $2t$-protrusion~$W$) by a
        representative~$G_1 \in \mathcal{R}_{2t}$. Let~$\tilde{G}$ be the $2t$-boundaried
        graph $G - W'$, where~$W'$ is the restricted protrusion of~$W$ and $B(\tilde{G}) = \partial_G(W)$.
        Since~$G_W \equipi{}{}{2t} G_1$, we have by
        Definition~\ref{def:finiteii},
    \begin{enumerate}
        \item $G = \tilde{G} \oplus G_W \in \mathcal{G}$ iff
                      $\tilde{G} \oplus G_1 \in \mathcal{G}$.
        \item $(\tilde{G} \oplus G_W, k) \in \Pi_{\mathcal G}$
                      iff $(\tilde{G} \oplus G_1, k - \Delta_{\Pi,2t}(G_1,G_W)) \in \Pi_{\mathcal{G}}$.
    \end{enumerate}
    Hence $G' = \tilde{G} \oplus G_1 \in \mathcal{G}$.
        Lemma~\ref{lemma:representatives} ensures that $\Delta_{\Pi,2t}(G_1,G_W) \geq 0$,
        and hence~$k' = k - \Delta_{\Pi,2t}(G_1,G_W)) \leq k$.
\end{proof}

In what follows, unless otherwise stated, when applying protrusion replacement
rules we will assume that for each $t \in \mathbb{N}$, we are given the set
$\mathcal{R}_t$ of representatives of the equivalence classes of
$\equiv_{\Pi_{\mathcal G},t}$. The representatives are chosen in accordance
with the condition stated in Lemma~\ref{lemma:representatives} so that for all
$G \in \mathcal{R}_t$ and all $G' \equiv_{\Pi_{\mathcal G},t} G$, we have that
$\Delta_{\Pi_{\mathcal G},t} (G,G') \geq 0$. Note that this makes our
algorithms of Section~\ref{sec:Kernels} \emph{non-uniform}. However
non-uniformity is implicitly assumed in previous work that used the protrusion
machinery for designing kernelization
algorithms~\cite{BFLPST09,FLST10,FLMPS11,FLMS12}.
\begin{definition}[Protrusion limit]
    For a parameterized graph problem $\Pi$ that has \fii in the
    class $\mathcal{G}$, let~$\mathcal{R}_t$ denote the set of
        representatives of the equivalence classes of $\equiv_{\Pi,t}$
        as in Lemma~\ref{lemma:representatives}. The
        \emph{protrusion limit} of~$\Pi_{\mathcal{G}}$ is a function
        $\prot_{\Pi_{\mathcal{G}}} \colon \mbb{N} \rightarrow \mbb{N}$
    defined as $\prot_{\Pi_{\mathcal{G}}}(t) = \max_{G \in \mathcal{R}_t} |V(G)|$.
    We drop the subscript when it is clear which graph problem is being referred to.
    We also define $\protd(t) := \prot(2t)$.
\end{definition}

The next two lemmas deal with finding protrusions in graphs. The first of these
guarantees that whenever there exists a ``large enough'' protrusion there
exists a protrusion that is large but bounded by a \emph{constant} (that
depends on the problem and the boundary size). As we shall see later, the fact
that we deal with protrusions of constant size enables us to efficiently test
which representative to replace them by, assuming that we have the set of
representatives. For completeness, we provide the proof of the following lemma.

\begin{lemma}[{\rm \!\cite{BFLPST09}}]\label{lemma:constantProtrusions}\omitted
    Let~$\Pi$ be a parameterized graph problem with \fii in~$\mathcal{G}$
    and let $t \in \mbb{N}$ be a constant.
    For a graph $G \in \mathcal{G}$, if one is given a $t$-protrusion $X \subseteq V(G)$ such that
    $\protd_{\Pi_{\mathcal{G}}}(t) < |X|$, then one can, in time $O(|X|)$, find a $2t$-protrusion $W$
    such that $\protd_{\Pi_{\mathcal{G}}}(t) < |W| \leq 2 \cdot \protd_{\Pi_{\mathcal{G}}}(t)$.
\end{lemma}
\begin{proof}
    Let $(T, \mathcal{X})$ be a nice tree-decomposition for $G[X]$ of width~$t-1$. Root~$T$ at
    an arbitrary node. Let~$u$ be the \emph{lowest} node of~$T$ such
    that if~$W$ is the set of vertices in the bags associated
    with the nodes in the subtree~$T_u$ rooted at~$u$, then~$|W| > \protd_{\Pi_{\mathcal{G}}}(t)$.
    Clearly $W$ is a $2t$-protrusion with boundary $X_u \cup \partial_G (X)$, where~$X_u \subseteq V(G)$
    is the bag associated with the node~$u$ of~$T$.
     By the choice of~$u$, it is clear that~$u$ cannot be a forget node. If~$u$ is an introduce
    node with child~$v$, then the number of vertices in the bags associated
    with the nodes of~$T_v$ must be exactly~$\protd_{\Pi_{\mathcal G}}(t)$.
    Since~$u$ introduces an additional vertex of~$G$, we have~$|W| = \protd_{\Pi_{\mathcal G}}(t) +1$.
    Finally consider the case when~$u$ is a join node with children~$y,z$. Then the bags associated
    with these nodes~$X_u, X_y, X_z$ are identical and since
    \smallskip
    $$
        \big| \, \bigcup_{\mathclap{j \in V(T_y)}} X_j \, \big|  <  \protd_{\Pi_{\mathcal{G}}}(t) \qquad \text{and} \qquad
        \big| \, \bigcup_{\mathclap{j \in V(T_z)}} X_j \, \big|  <  \protd_{\Pi_{\mathcal{G}}}(t),
    $$
    \smallskip
    we have that
    $W = \bigcup_{j \in V(T_y)} X_j \cup \bigcup_{j \in V(T_z)} X_j$ has size at most
    $ 2 \cdot \protd_{\Pi_{\mathcal{G}}}(t)$.

    Computing a nice tree-decomposition $(T, \mathcal{X})$ of~$G[X]$ takes time
    $O(2^{O(t^3)} \cdot |X|)$~\cite{Bod96} and the time required to
    compute a $2t$-protrusion from~$T$ is~$O(|X|)$. Since~$t$ is a constant,
    the total time taken is~$O(|X|)$.
\end{proof}

For a fixed~$t$, the protrusion~$W$ is of \emph{constant} size but, in the reduction
rule to be described, would be replaced by a representative of size at most~$\prot_{\Pi_{\mathcal G}}(2t)$.
This means that each time the reduction rule is applied, the size of the graph
strictly decreases and, by Lemma~\ref{lemma:representatives}, the parameter does not increase.
The reduction rule can therefore be applied at most $n$ times, where $n$ is the number of vertices
in the input graph. As we shall see later, each application of the reduction rule takes
time polynomial in~$n$, assuming that we are given the set of representatives.
Therefore, in polynomial time, we would obtain an instance in which every $t$-protrusion
has size at most $\prot_{\Pi_{\mathcal G}}(2t)$. This trick is described
in~\cite{BFLPST09} but is stated here
for the sake of completeness.

The next lemma describes how to find a $t$-protrusion of maximum size.
\begin{lemma}[Finding maximum sized protrusions]\omitted
    Let~$t$ be a constant. Given an $n$-vertex graph~$G$, a $t$-protrusion of~$G$ with the
        maximum number of vertices can be found in time $O(n^{t + 1})$.
\end{lemma}
\begin{proof}
    For a vertex set $B \subseteq V(G)$ of size at most~$t$, let~$C_{B,1}, \ldots, C_{B,p}$
    be the connected components of~$G - B$ such that, for $1 \le i \le p$,
        $\tw(G[V(C_{B,i}) \cup B]) \le t$. The connected components of $G - B$ can be determined
        in $O(n)$ time and one can test whether the graph induced by~$V(C_{B,i}) \cup B$ has
        treewidth at most~$t-1$ in time $O(2^{O(t^3)} \cdot n)$~\cite{Bod96}. Since we have
        assumed that~$t$ is a fixed constant, deciding whether the treewidth is within~$t-1$
        can be done in linear time. By definition, $\bigcup_{i=1}^{p}V(C_{B,i}) \cup B$ is
        a $t$-protrusion with boundary~$B$. Conversely every $t$-protrusion~$W$ consists
        of a boundary~$\partial (W)$ of size at most~$t$ such that the restricted
        protrusion~$W' = W \setminus \partial (W)$ is a collection of connected components~$C$
        of~$G - \partial(W)$
        satisfying the condition~$\tw(G[V(C) \cup \partial (W)]) \le t-1$.
    Therefore to find a $t$-protrusion of maximum size, one simply runs through
    all vertex sets~$B$ of size at most~$t$ and for each set determines the maximum
    $t$-protrusion with boundary~$B$. The largest $t$-protrusion over all choices
    of the boundary~$B$ is a largest $t$-protrusion in the graph. All of this takes
        time $O(n^{t + 1})$.
\end{proof}

Finally, given a $2t$-protrusion $W$ with the desired size constraints,
we show how to determine which representative of our equivalence class is equivalent
to $G[W]$.
\begin{lemma}\omitted
Let $\Pi$ be a parameterized graph problem that has \fii on
$\mathcal G$. For $t \in \mbb{N}$, a constant, suppose that the set
$\mathcal{R}_{t}$ of representatives of the \journal{equivalence relation}
$\equiv_{\Pi,t}$ is given. If $W$ is a $t$-protrusion of size at most $c$, a
fixed constant, then one can decide in constant time which $G' \in
\mathcal{R}_t$ satisfies $G' \equiv_{\Pi,t} G[W]$.
\end{lemma}
\begin{proof}
    Fix $G' \in \mathcal{R}_t$. We wish to test whether $G' \equiv_{\Pi,t} G[W]$.
    For each $\tilde{G} \in \mathcal{R}_t$, solve the problem $\Pi$ on the
    constant-sized instances $G[W] \oplus \tilde{G}$ and $G' \oplus \tilde{G}$ and
    let $s(G[W],\tilde{G})$ and $s(G',\tilde{G})$ denote the size of the optimal
    solution. Then by the definition of finite integer index, we have $G'
    \equiv_{\Pi,t} G[W]$ if and only if $s(G[W],\tilde{G}) - s(G',\tilde{G})$ is
    the \emph{same} for all $\tilde{G} \in \mathcal{R}_t$. To find out which graph
    in $\mathcal{R}_t$ is the correct representative of $G[W]$, we run this test
    for each graph in $\mathcal{R}_t$, of which there are a constant number. The
    total time taken is, therefore, a constant.
\end{proof}
}

\section{Constructing protrusion decompositions}\label{sec:Decomposition}
\label{sec:decomposition}


\journal{In this section we present our algorithm to compute protrusion
decompositions. Our approach is based on an algorithm which marks the bags of a
tree-decomposition of an input graph $G$ that comes equipped with a subset $X
\subseteq V(G)$ such that the graph $G-X$ has bounded treewidth. Let henceforth
$t$ be an integer such that $\tw(G-X) \leqslant t-1$ and let $r$ be an integer
that is also given to the algorithm. This parameter $r$ will depend on the
particular graph class to which $G$ belongs and the precise problem one might
want to solve (see Sections~\ref{sec:Kernels} and~\ref{sec:PlanarF} for more
details). More precisely, given optimal tree-decompositions of the connected
components of $G-X$ with at least $r$ neighbors in $X$, the bag marking
algorithm greedily identifies a set of bags $\mathcal{M}$ in a bottom-up
manner. The set $V(\mc{M})$ of vertices contained in marked bags together with
$X$ will form the separating part $Y_0$ of the protrusion decomposition.
Intuitively, the marked bags will be mapped bijectively into a collection of
pairwise vertex-disjoint connected subgraphs of $G-X$, each of which has a
large neighborhood in $X$ (namely, of size greater than $r$), implying in
several particular cases a limited number of marked bags (see
Sections~\ref{sec:Kernels} and~\ref{sec:PlanarF}). In order to guarantee that
the connected components of $G-(X\cup V(\mc{M}))$ form protrusions with small
boundary, the set $\mc{M}$ is closed under taking LCA's (least common
ancestors; see Lemma~\ref{lem:LCA}). The precise description of the procedure
can be found in Algorithm~\ref{alg:marking} below and a sketch of the
decomposition is depicted in Figure~\ref{fig:AlgOutput}.}

\short{In this section we present our algorithm to compute protrusion
decompositions. Our approach is based on an algorithm which marks the bags of a
tree-decomposition of an input graph $G$ that comes equipped with a $t$-treewidth-modulator
$X \subseteq V(G)$. Our algorithm also takes an additional integer parameter
$r$, which depends on the graph class to which $G$ belongs and the precise
problem one might want to solve (see Sections~\ref{sec:Kernels}
and~\ref{sec:PlanarF} for more details). More precisely, given optimal
tree-decompositions of the connected components of $G-X$ with at least $r$
neighbors in $X$, the bag marking algorithm greedily identifies a set of bags
$\mathcal{M}$ in a bottom-up manner. The set $V(\mc{M})$ of vertices contained
in marked bags together with $X$ will form the separating part $Y_0$ of the
protrusion decomposition. Intuitively, the marked bags will be mapped
bijectively into a collection of pairwise vertex-disjoint connected subgraphs
of $G-X$, each of which has a neighborhood in $X$ of size greater than $r$,
implying in several particular cases a limited number of marked bags (see
Sections~\ref{sec:Kernels} and~\ref{sec:PlanarF}). In order to guarantee that
the connected components of $G-(X\cup V(\mc{M}))$ form protrusions with small
boundary, the set $\mc{M}$ is closed under taking LCA's (least common
ancestors). The precise description of the procedure can be found in
Algorithm~\ref{alg:marking} below.}

\begin{algorithm}[h]
\short{\small} \KwIn{A graph $G$, a subset $X \subseteq V(G)$ such that
$\tw(G-X)\leqslant t-1$, and an integer $r >0$.}
\short{\def\BlankLine{\vskip .73ex}} 

\BlankLine

Set $\mc{M}\leftarrow\emptyset$ as the set of marked bags\; Compute an optimal
rooted tree-decomposition $\mc{T}_C=(T_C,\mc{B}_C)$ of every connected
component $C$ of $G-X$ such that $|N_{X}(C)|\geqslant r$\; \BlankLine

Repeat the following loop for every rooted tree-decomposition $\mc{T}_C$\;
\While{$\mc{T}_C$ contains an unprocessed bag}{

\BlankLine Let $B$ be an unprocessed bag at the farthest distance from the root
of $\mc{T}_C$\; \BlankLine
\textbf{[LCA marking step]}\\
 \If{$B$ is the LCA of two marked bags of $\mc{M}$}{
    $\mc{M}\leftarrow\mc{M}\cup\{B\}$ and remove the vertices of $B$ from every bag of $\mc{T}_C$\;
    }

\BlankLine
\textbf{[Large-subgraph marking step]}\\
\ElseIf{$G_B$ contains a connected component $C_B$ such that
$|N_{X}(C_B)|\geqslant r$}{
    $\mc{M}\leftarrow\mc{M}\cup\{B\}$ and remove the vertices
    of $B$ from every bag of $\mc{T}_C$\;
    }

\BlankLine Bag $B$ is now processed\; }

\BlankLine \Return{$Y_0=X\cup V(\mc{M})$}\;

\caption{{\sc Bag marking algorithm}} \label{alg:marking}
\end{algorithm}

\short{Note that an optimal tree-decomposition of every connected component $C$
of $G-X$ such that $|N_{X}(C)|\geqslant r$ can be computed in time linear in $n
= |V(G)|$ using the algorithm of Bodlaender~\cite{Bod96}.
In the Appendix we sketch how the Large-subgraph
marking step can be implemented using standard dynamic programming techniques.
It is quite easy to see that
Algorithm~\ref{alg:marking} runs in linear time.}

\journal{Before we discuss properties of the set $\mc{M}$ of marked bags and
the set $Y_0 = X \cup V(\mc{M})$, let us establish the time complexity of the
bag marking algorithm and describe how the dynamic programming is done in the
Large-subgraph marking step. Since the dynamic programming procedure is quite
standard, we just sketch the main ideas.

\paragraph{Implementation and time complexity of Algorithm~\ref{alg:marking}.}
First, an optimal tree-decomposition of every connected component $C$ of $G-X$
such that $|N_{X}(C)|\geqslant r$
can be computed in time linear in $n = |V(G)|$ using the algorithm of
Bodlaender for graphs of bounded treewidth~\cite{Bod96}. We root such
tree-decomposition at an arbitrary bag. For the sake of simplicity of the
analysis, we can assume that the tree-decompositions are nice, but it is not
necessary for the algorithm.

Note that the LCA marking step can clearly be performed in linear time. Let us
now briefly discuss how we can detect, in the Large-subgraph marking step, if a
graph $G_B$ contains a connected component $C_B$ such that
$|N_{X}(C_B)|\geqslant r$ using dynamic programming. For each bag $B$ of the
tree-decomposition, we have to keep track of which vertices of $B$ belong to
the same connected component of $G_B$.

Note that we only need to remember the
connected components of the graph $G_B$ which intersect $B$, as the other ones
will never be connected to the rest of the graph. For each such connected
component $C_B$ intersecting $B$, we also store $N_X(C_B)$, and note that by
definition of the algorithm, it follows that for non-marked bags $B$,
$|N_{X}(C_B)| < r$. At a ``join'' bag $J$ with children $B_1$ and $B_2$, we
merge the connected components of $G_{B_1}$ and $G_{B_2}$ sharing at least one
vertex (which is necessarily in $J$), and update their neighborhood in $X$
accordingly. If for some of these newly created connected components $C_J$ of
$G_J$, it holds that $|N_{X}(C_J)|\geqslant r$, then the bag $J$ needs to be
marked. At a ``forget'' bag $F$ corresponding to a forgotten vertex $v$, we
only have to forget the connected component $C$ of $G_F$ containing $v$ if
$V(C) \cap F = \emptyset$. Finally, at an ``introduce'' bag $I$ corresponding
to a new vertex $v$, we have to merge the connected components of $G_I$ after
the addition of vertex $v$, and update the neighbors in $X$ according to the
neighbors of $v$ in $X$.

\begin{figure}[t]
    \center\includegraphics[width=.6\textwidth]{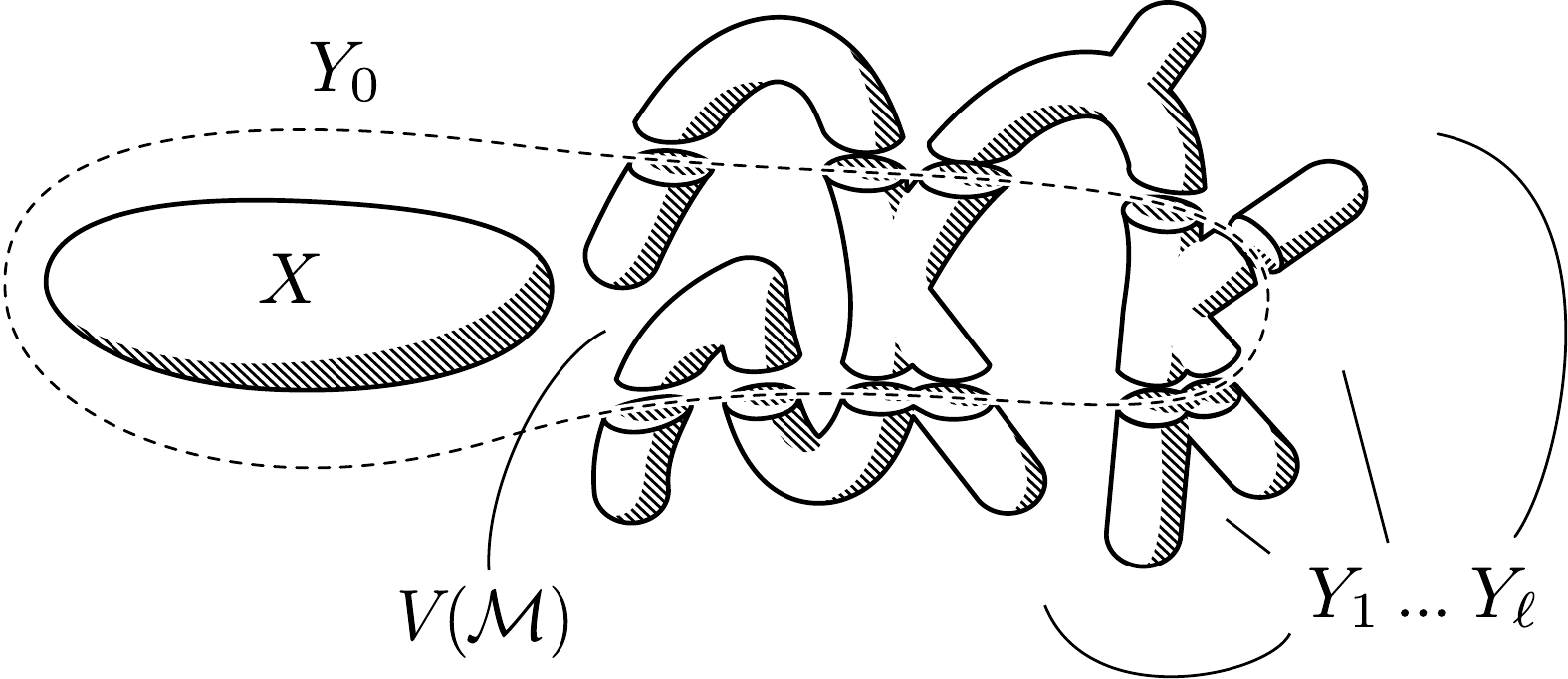}
    \caption{\label{fig:AlgOutput}A sketch of how the marking algorithm obtains
    a protrusion decomposition. $X$ denotes a treewidth-modulator. Edges among
    the individual vertex sets are not depicted.}
\end{figure}

Note that for each bag $B$, the time needed to update the information about the
connected components of $G_B$ depends polynomially on $t$ and $r$. In order for
the whole algorithm to run in linear time, we can deal with the removal of
marked vertices in the following way. Instead of removing them from every bag
of the tree-decomposition, we can just label them as ``marked'' when marking a
bag $B$, and just not take them into account when processing further bags.

The next lemma follows from the above discussion.}



\begin{lemma}\omitted\label{lem:algo-running-time}
Algorithm~\ref{alg:marking} can be implemented to run in $O(n)$ time, where the
hidden constant depends only on $t$ and $r$.
\end{lemma}

\journal{\paragraph{Basic properties of Algorithm~\ref{alg:marking}.} Denote by
$\mc{T}$ the union of the set of optimal tree-decompositions $\mc{T}_C$ of
every connected component $C$ of $G-X$ with at least $r$ neighbors in $X$.}

\short{We now state some properties of Algorithm~\ref{alg:marking}.
The following lemma states that every connected component of
$G-Y_0$ has a small neighborhood in $X$ and thus forms a restricted protrusion.
Its proof uses the fact that the marked bags are closed under
taking LCA's, and therefore every maximal non-marked subtree of a
tree-decomposition is adjacent to at most two marked bags.}

\journal{\begin{lemma}\label{lem:LCA} If $T$ is a maximal connected subtree of
$\mc{T}$ not containing any marked bag of $\mc{M}$, then $T$ is adjacent to at
most two marked bags of $\mathcal{T}$.
\end{lemma}
\begin{proof} As that every tree-decomposition in $\mc{T}$ is rooted,
so is any maximal subtree $T$ of $\mc{T}$ not containing any marked bag of
$\mc{M}$. Assume that $\mc{M}$ contains two distinct marked bags, say $B_1$ and
$B_2$, each adjacent to a leaf of $T$. As $T$ is connected, observe that the
LCA $B$ of $B_1$ and $B_2$ belongs to $T$. Since $\mc{M}$ is closed under
taking LCA, $T$ contains a marked bag $B$, a contradiction. It follows that $T$
is adjacent to at most two marked bags: a unique one adjacent to a leaf, and
possibly another one adjacent to its root.
\end{proof}

As a consequence of the previous lemma we can now argue that every connected
component of $G-Y_0$ has a small neighborhood in $X$ and thus forms a
restricted protrusion.}

\begin{lemma}\omitted\label{lem:protrusion}
Let $Y_0$ be the set of vertices computed by Algorithm~\ref{alg:marking}. Every
connected component $C$ of $G-Y_0$ satisfies $|N_X(C)| < r$ and $|N_{Y_0}(C)| <
r+2t$.
\end{lemma}
\journal{\begin{proof} Let $C$ be a connected component of $G-Y_0$. Observe
that $C$ is contained in a connected component $C_X$ of $G-X$ such that either
$|N_X(C_X)|< r$ or $|N_X(C_X)|\geqslant r$. In the former case, as
Algorithm~\ref{alg:marking} does not mark any vertex of $C_X$, $C=C_X$ and so
$|N_{Y_0}(C)| < r+2t$ trivially holds. So assume that $|N_X(C_X)|\geqslant r$.
Then $C_X$ has been chopped by Algorithm~\ref{alg:marking} and clearly
$C\subseteq C_X\setminus V(\mc{M})$. More precisely, if $\mc{T}_{C_X}$ is the
rooted tree-decomposition of $C_X$, there exists a maximal connected subtree
$T$ of $\mc{T}_{C_X}$ not containing any marked bag such that $C\subseteq
V(T)\setminus V(\mc{M})$. By construction of $\mc{M}$, every connected
component of the subgraph induced by $V(T)\setminus V(\mc{M})$ has strictly
less than $r$ neighbors in $X$ (otherwise the root of $T$ or one of its
descendants would have been marked at the Large-subgraph marking step). It
follows that $|N_X(C)| < r$. To conclude, observe that Lemma~\ref{lem:LCA}
implies that the neighbors of $C$ in $V(\mc{M})$ are contained in at most two
marked bags of $\mathcal{T}$. It follows that $|N_{Y_0}(C)| < r+2t$.
\end{proof}}

Given a graph $G$ and a subset $S \subseteq V(G)$, we define a \emph{cluster of
$G-S$} as a maximal collection of connected components of $G-S$ with the same
neighborhood in $S$. Note that the set of all clusters of $G-S$ induces a
partition of the set of connected components of $G-S$, which can be easily
found in linear time if $G$ and $S$ are given.

By Lemma~\ref{lem:protrusion} and using the fact that $\tw(G-X) \leqslant t-1$,
the following proposition follows.

\begin{proposition}
    \label{prop:decompositionSparse} Let $r,t$ be two positive integers, let $G$ be
    a graph and $X \subseteq V(G)$ such that $\tw(G-X) \leqslant t-1$, let $Y_0
    \subseteq V(G)$ be the output of Algorithm~\ref{alg:marking} with input
    $(G,X,r)$,
    and let $Y_1, \ldots, Y_{\ell}$ be the set of all clusters of $G-Y_0$.
    Then $\mc{P}:=Y_0\uplus Y_1\uplus\cdots\uplus Y_{\ell}$ is a
    $(\max\{\ell,|Y_0|\},2t+r)$-protrusion decomposition of $G$.
\end{proposition}

In other words, each cluster of $G-Y_0$ is a restricted $(2t+r)$-protrusion.
Note that Proposition~\ref{prop:decompositionSparse} neither
bounds $\ell$ or $|Y_0|$. In the
sequel, we will use Algorithm~\ref{alg:marking} and
Proposition~\ref{prop:decompositionSparse} to give explicit bounds on $\ell$
and $|Y_0|$, in order to achieve two different results. In
Section~\ref{sec:Kernels} we use Algorithm~\ref{alg:marking} and
Proposition~\ref{prop:decompositionSparse}
to obtain \textsl{linear kernels} for a large class of problems on \textsl{sparse}
graphs. In Section~\ref{sec:PlanarF} we use
Algorithm~\ref{alg:marking} and Proposition~\ref{prop:decompositionSparse}
to obtain a {\sl single-exponential algorithm} for the parameterized \planarF{} problem.

\section{Linear kernels on graphs excluding a topological minor}\label{sec:Kernels}
\journal{In this section we prove Theorem~I. We then state a number of concrete
problems that satisfy the structural constraints imposed by this theorem
(Subsection~\ref{subsec:concreteProblems}), discuss these constraints in the
context of previous work in this area
(Subsection~\ref{subsec:comparisonStructural}), and trace graph classes to
which our approach can be lifted (Subsection~\ref{subsec:limits}). Finally, in
Subsection~\ref{subsec:exampleKernel} we discuss how to use the machinery
developed in proving Theorem~I to obtain a concrete kernel for the \EDS
problem.} \short{In this section we prove our first main result (Theorem~I).
We then state a number of
concrete problems that satisfy the structural constraints imposed by this
theorem and discuss these constraints in the context of previous work in this
area. Graph classes to which our approach can still be lifted and an example of
how to use the machinery developed in proving Theorem~I to obtain a concrete
kernel for the \EDS problem can be found in the Appendix.}

With the protrusion machinery outlined in Section~\ref{sec:Preliminaries} at
hand, we can now describe the protrusion reduction rule. \journal{Informally,
we find a sufficiently large $t$-protrusion (for some yet to be fixed constant
$t$), replace it with a small representative, and change the parameter
accordingly. }In the following, we will drop the
subscript from the protrusion limit functions $\prot_{\Pi}$ and $\protd_{\Pi}$.
\short{
\begin{redrule}[Protrusion reduction rule]\label{rrule:prot}
    Let $\Pi_{\mc G}$ denote a parameterized graph problem
    restricted to some graph class $\mc G$, let $(G,k) \in \Pi_{\mc G}$
    be a \YES-instance of $\Pi_{\mc G}$, and let~$t \in \mbb N$ be a constant.
        Suppose that $W' \subseteq V(G)$ is a $t$-protrusion of~$G$ such that $|W'| > \protd(t)$. Let $W \subseteq V(G)$ be a $2t$-protrusion
of~$G$
        such that $\protd(t) < |W| \leq 2 \cdot \protd(t)$,
    obtained as described in Lemma~\ref{lemma:constantProtrusions}.
    We let $G_W$ denote the $2t$-boundaried graph $G[W]$ with boundary
    $\bound(G_W) = \partial_G(W)$. Let further~$G_1 \in \mathcal{R}_{2t}$ be the
    representative of~$G_W$ for the equivalence
    relation~$\equipi{}{}{|\partial(W)|}$ as defined in Lemma~\ref{lemma:representatives}.
        The protrusion reduction rule (for boundary size $t$) is the following:
    \emph{Reduce $(G,k)$ to~$(G',k') = (G[V \setminus W'] \oplus G_1, k-\Delta_{\Pi,2t}(G_1,G_W))$}.
\end{redrule}

}
\journal{
\begin{redrule}[Protrusion reduction rule]\label{rrule:prot}
    Let $\Pi_{\mc G}$ denote a parameterized graph problem
    restricted to some graph class $\mc G$, let $(G,k) \in \Pi_{\mc G}$
    be a \YES-instance of $\Pi_{\mc G}$, and let~$t \in \mbb N$ be a constant.
        Suppose that $W' \subseteq V(G)$ is a $t$-protrusion of~$G$
        such that $|W'| > \protd(t)$. Let $W \subseteq V(G)$ be a $2t$-protrusion of~$G$
        such that $\protd(t) < |W| \leq 2 \cdot \protd(t)$,
    obtained as described in Lemma~\ref{lemma:constantProtrusions}.
    We let $G_W$ denote the $2t$-boundaried graph $G[W]$ with boundary
    $\bound(G_W) = \partial_G(W)$. Let further~$G_1 \in \mathcal{R}_{2t}$ be the
    representative of~$G_W$ for the equivalence
    relation~$\equipi{}{}{|\partial(W)|}$ as defined in Lemma~\ref{lemma:representatives}.

    \noindent The protrusion reduction rule (for boundary size $t$) is the following:
    \begin{quote}
        \emph{Reduce $(G,k)$ to~$(G',k') = (G[V \setminus W'] \oplus G_1, k-\Delta_{\Pi,2t}(G_1,G_W))$}.
    \end{quote}
\end{redrule}
} By Lemma~\ref{lemma:representatives}, the parameter in the new instance does not increase.
\short{The safety of the above reduction rule is shown in the Appendix.}
\journal{We now show that the protrusion reduction rule is safe.
\begin{lemma}[Safety]
        Let~$\mathcal G$ be a graph class and let~$\Pi_{\mathcal{G}}$ be a
        parameterized graph problem with finite integer index \wrt~$\mathcal{G}$.
        If~$(G',k')$ is the instance obtained from one application of the protrusion
        reduction rule to the instance~$(G,k)$ of~$\Pi_\mathcal{G}$, then
    \begin{enumerate}
        \item $G' \in \mathcal{G}$;
                \item $(G',k')$ is a \textsc{Yes}-instance iff
            $(G,k)$ is a \textsc{Yes}-instance; and
        \item $k' \leq k$.
    \end{enumerate}
\end{lemma}
\begin{proof}
    Suppose that~$(G',k')$ is obtained from~$(G,k)$ by replacing a $2t$-boundaried
    subgraph~$G_W$ (induced by a $2t$-protrusion~$W$) by a
        representative~$G_1 \in \mathcal{R}_{2t}$. Let~$\tilde{G}$ be the $2t$-boundaried
        graph $G - W'$, where~$W'$ is the restricted protrusion of~$W$ and $B(\tilde{G}) = \partial_G(W)$.
        Since~$G_W \equipi{}{}{2t} G_1$, we have by
        Definition~\ref{def:finiteii},
    \begin{enumerate}
        \item $G = \tilde{G} \oplus G_W \in \mathcal{G}$ iff
                      $\tilde{G} \oplus G_1 \in \mathcal{G}$.
        \item $(\tilde{G} \oplus G_W, k) \in \Pi_{\mathcal G}$
                      iff $(\tilde{G} \oplus G_1, k - \Delta_{\Pi,2t}(G_1,G_W)) \in \Pi_{\mathcal{G}}$.
    \end{enumerate}
    Hence $G' = \tilde{G} \oplus G_1 \in \mathcal{G}$.
        Lemma~\ref{lemma:representatives} ensures that $\Delta_{\Pi,2t}(G_1,G_W) \geq 0$,
        and hence~$k' = k - \Delta_{\Pi,2t}(G_1,G_W)) \leq k$.
\end{proof}
}

\begin{observation}
    If $(G,k)$ is reduced \wrt the protrusion reduction rule with boundary size $\beta$,
    then for all $t \leq \beta$, every $t$-protrusion $W$ of $G$ has size at most $\protd(t)$.
\end{observation}

In order to obtain linear kernels, we require the problem instances to have
more structure. In particular, we adapt the notion of \emph{quasi-compactness}
introduced in~\cite{BFLPST09} to define what we call \emph{treewidth-bounding}.\short{\looseness-1}
\short{
\begin{definition}[Treewidth-bounding]\label{def:treewidthBounding}
    A parameterized graph problem $\Pi$ is called $(s,t)$-\emph{treewidth-bounding} for
    a function~$s\colon \mbb N \rightarrow \mbb N$
    and a constant~$t$ if for all $(G,k) \in \Pi$ there exists
    $X\! \subseteq \! V(G)$ (the \emph{treewidth-modulator})
        such that $|X| \le s(k)$ and $\tw(G-X) \le t-1$.
    We call $\Pi$ \emph{treewidth-bounding
    on a graph class $\mc G$} if this condition holds under the restriction that
    $G \in \mc G$.
    We call $s$ the \emph{treewidth-modulator size} and $t$ the \emph{treewidth bound}
    of the problem $\Pi$.
\end{definition}
}
\journal{
\begin{definition}[Treewidth-bounding]\label{def:treewidthBounding}
    A parameterized graph problem $\Pi$ is called $(s,t)$-\emph{treewidth-bounding} if
    there exists a function~$s\colon \mbb N \rightarrow \mbb N$
    and a constant~$t$ such that
    for every $(G,k) \in \Pi$ there exists~\journal{$X \subseteq V(G)$}\short{$X\! \subseteq \! V(G)$}
        such that:%
    \journal{\begin{enumerate}
        \item $|X| \le s(k)$; and
        \item $\tw(G-X) \le t - 1$.
    \end{enumerate}}%
    We call a problem \emph{treewidth-bounding
    on a graph class $\mc G$} if the above property holds under the restriction that
    $G \in \mc G$. We call $X$ a \emph{$t$-treewidth-modulator} of $G$,
    $s$ the \emph{treewidth-modulator size} and $t$ the \emph{treewidth bound}
    of the problem $\Pi$.
\end{definition}
} \noindent We assume in the following that the problem $\Pi$ at hand is
treewidth-bounding with bound $t$ and modulator size $s(\cdot)$, that is, a
\YES-instance $(G,k) \in \Pi_{\mc G}$ has a modulator set $X \subseteq V(G)$
with $|X| \leq s(k)$ and $\tw(G-X) \leq t-1$. Note that in general $s,t$ depend
on $\Pi$ and $\mc G$. \journal{For many problems that are treewidth-bounding,
such as \VC, \FVS, \Problem{Treewidth-$t$ Vertex Deletion}, the set~$X$ is
actually the solution set. However, in general, $X$ could be \emph{any} vertex
set and does not have to be given nor efficiently computable to obtain a
kernel. The fact that it exists is all we need for our proof to go through.
}%

The rough idea of the proof of Theorem~I is as follows. We assume that the
given instance $(G,k)$ is reduced \wrt the protrusion reduction rule for some
yet to be fixed constant boundary size~$\beta$. Consequently, every
$\beta$-protrusion of $G$ has size at most $\protd(\beta)$. For a protrusion
decomposition \YYYY obtained from Algorithm~\ref{alg:marking} with a carefully
chosen threshold, we can then show that $|Y_0| = O(k)$ using properties of
$H$-topological-minor-free graphs. The bound on the total size of the clusters
of $G-Y_0$ then follows from these properties and from the protrusion reduction
rule.\journal{\\} \journal{We first prove a result
(Theorem~\ref{thm:KernelByConstriction}) that is slightly more general than
Theorem~I and identifies all the key ingredients needed for our result. To do
this, we use a sequence of lemmas (\ref{lemma:MarkAndReduce},
\ref{lemma:BoundY0}, \ref{lemma:BoundEll}) which
bounds the total size of the clusters of the protrusion decomposition.}%
\short{We first prove a slight generalization of Theorem~I which highlights all
the key ingredients required.} To this end, we define the \emph{constriction}
operation, which essentially shrinks paths into edges.
%
%
\short{
\begin{definition}[Constriction]
    Let $G$ be a graph and let $\mc P$ be a set of paths in $G$ such
    that for each $P \in \mc P$ we have $(1)$~the endpoints of $P$ are
        not connected by an edge in $G$; and $(2)$~for all $P' \in \mc P$, with
        $P' \ne P$, $V(P) \cap V(P')$ has at most one vertex, which must
        also be an endpoint of both paths.
    We define the \emph{constriction} of $G$ under $\mc P$, denoted by $G|_{\mc P}$,
    as the graph $H$ obtained by connecting the endpoints of each $P \in \mc P$
        by an edge and then removing all inner vertices of $P$.
\end{definition}
}
\journal{
\begin{definition}[Constriction]
    Let $G$ be a graph and let $\mc P$ be a set of paths in $G$ such
    that for each $P \in \mc P$ it holds that:
    \begin{enumerate}
        \item the endpoints of $P$ are not connected by an edge in $G$; and
        \item for all $P' \in \mc P$, with $P' \ne P$, $P$ and $P'$ share at most
              a single vertex which must also be an endpoint of both
    \end{enumerate}
    We define the \emph{constriction} of $G$ under $\mc P$, written $G|_{\mc P}$,
    as the graph $H$ obtained by connecting the endpoints of each $P \in \mc P$
        by an edge and then removing all inner vertices of $P$.
\end{definition}
}

We say that $H$ is a \emph{$d$-constriction} of $G$ if there exists
$G' \subseteq G$ and a set of paths $\mc P$ in $G'$ such that $d = \max_{P \in
\mc P} |P|$ and $H = G'|_{\mc P}$. Given graph classes $\mc G, \mc H$ and some
integer $d \geq 2$, we say that \emph{$\mc G$~$d$-constricts into $\mc H$} if
for every $G \in \mc G$, every possible $d$-constriction $H$ of $G$ is
contained in the class $\mc H$. For the case that $\mc G = \mc H$ we say that
$\mc G$ is \emph{closed under $d$-constrictions}. We will call
$\mc H$ the \emph{witness} class, as the proof of
Theorem~\ref{thm:KernelByConstriction} works by taking an input graph $G$ and
constricting it into some witness graph $H$ whose properties will yield the
desired bound on $|G|$. We let $\omega(G)$ denote the size of a largest clique
in $G$ and $\# \omega(G)$ the total number of cliques in $G$
(not necessarily maximal ones).\looseness-1

\begin{theorem}\label{thm:KernelByConstriction}\omitted
    Let $\mc G,\mc H$ be graph classes closed under taking subgraphs such
    that $\mc G$ $d$-constricts into $\mc H$ for a fixed constant $d \in \mbb N$.
    Assume that $\mc H$ has the property that there exists functions
        $f_E,f_{\#\omega} \colon \mbb N \rightarrow \mbb N$ and a constant
        $\omega_{\mc H}$ (depending only on $\mc H$)
    such that for each graph $H \in \mc H$ the following conditions hold:
    $$
        |E(H)| \leq f_E(|H|),
        \enskip \#\omega(H) \leq f_{\#\omega}(|H|),
        \enskip \text{and} \enskip \omega(H) < \omega_{\mc H}.
    $$
    Let $\Pi$ be a parameterized graph problem
    that has \fii and is
    $(s, t)$-treewidth-bounding, both on the graph class $\mc G$.
    Define $x_k := s(k) + 2t \cdot f_E(s(k))$. Then any reduced
    instance $(G,k) \in \Pi$ has a protrusion decomposition
    $V(G) = \YYYY$ such that:
    \journal{\begin{enumerate}
        \item $|Y_0| \le x_k$;
        \item $|Y_i| \le \protd(2t+ \omega_{\mc H})$ for $1 \le i \le \ell$;
        and
        \item $\ell \le f_{\#\omega}(x_k)+x_k+1$.
    \end{enumerate}}%
    \short{
        (1) $|Y_0| \le x_k$; (2) $|Y_i| \le \protd(2t+ \omega_{\mc H})$ for $1 \le i \le \ell$;
        and (3) $\ell \le f_{\#\omega}(x_k)+x_k+1$.
    }%
    Hence $\Pi$ restricted to $\mc G$ admits kernels of size at most%
    \journal{$$
        x_k + (f_{\#\omega}(x_k)+x_k+1) \protd(2t + \omega_{\mc H}).
    $$}
    \short{$x_k + (f_{\#\omega}(x_k)+x_k+1) \protd(2t + \omega_{\mc H}).$}
\end{theorem}
\short{\smallskip}%
\journal{%
\noindent We split the proof of Theorem~\ref{thm:KernelByConstriction} into several lemmas. First,
let us fix the way in which the decomposition \YYYY is obtained: given a reduced \YES-instance
$(G,k) \in \mc G$, let $X \subseteq V(G)$ be a treewidth-modulator of size at most $|X| \le s(k)$
such that $\tw(G-X) \leq t-1$. We run Algorithm~\ref{alg:marking} on the input
$(G,X,\omega_{\mc H})$.

\begin{lemma}\label{lemma:MarkAndReduce}
    The protrusion decomposition \YYYY obtained by running Algorithm~\ref{alg:marking}
    on $(G,X,\omega_{\mc H})$ has the following properties:
    \begin{enumerate}
        \item For each $1 \le i \le \ell$, we have $|Y_i| \le \protd(2t+\omega_{\mc H})$;
        \item For each connected subgraph $C_B$ witnessed by Algorithm~\ref{alg:marking}
                      in the ``Large-subgraph marking step'',
                      $|C_B| \le \protd(2t+\omega_{\mc H}) + t$.
    \end{enumerate}
\end{lemma}
\begin{proof}
    The first claim follows directly from Lemma~\ref{lem:protrusion}: for each $1 \le i \le \ell$,
    we have $|N_{Y_0}(Y_i)| \leq 2t + \omega_{\mc H}$. As $Y_i \subseteq G-X$, it follows
    that $\tw(G[Y_i]) \leq t-1$ and therefore $Y_i$ forms a restricted $(2t+r)$-protrusion
    in $G$. Since our instance is reduced, we have $|Y_i| \le \protd(2t+\omega_{\mc H})$.

    Note that during a run of the algorithm, if a bag $B$ currently being considered is
    not marked, then each connected component $C_B$ of $G_B$ satisfies $|N_X(C_B)| < r$.
    Hence $C_B$ along with its neighbors in~$X$ is a $t$-protrusion and since the
    instance is reduced we have $|C_B| \le \protd(2t+\omega_{\mc H})$. Moreover the algorithm ensures that
    $|N_R(C_B)| \le 2t$, where $R = {V(G)\setminus (X \cup V(\mathcal M) \cup \set{B})}$,
    and thus a component with a neighborhood larger than $2t+r$ must have at least
    $r$ neighbors in $X$. Now as every step of the algorithm adds at most $t$ more vertices to the
    components of $G_B$, it follows that once a component with at least $r$ neighbors
    in $X$ \emph{is} witnessed, it can contain at most $\protd(2t+\omega_{\mc H})+t$ vertices.
\end{proof}

Now, let us prove the claimed bound on $|Y_0|$ by making use of the assumed bounds $\omega_{\mc H}$
and $f_E(\cdot)$ imposed on graphs of the witness class $\mc H$.
\begin{lemma}\label{lemma:BoundY0}
    The number of bags marked by Algorithm~\ref{alg:marking} to obtain \YYYY is
        at most $2f_E(s(k))$, and therefore $|Y_0| \leq x_k = s(k)+2f_E(s(k))\cdot t$.
\end{lemma}
\begin{proof}
    \def\Rsize{\protd(2t+\omega_{\mc H})+t}
    For each bag marked in the ``Large-subgraph marking step'' of the algorithm,
    a connected subgraph $C$ of $G-X$ with $|N_X(C)| \geq \omega_{\mc H}$ is witnessed.
    Suppose that the algorithm witnesses $p$ such connected subgraphs $C_1,\dots,C_p$.
    Then the number of marked bags is at most $2p$, since the LCA
    marking step can at most double the number of marked bags.

    By the design of Algorithm~\ref{alg:marking}, the connected subgraphs $C_i$ are
    pairwise vertex-disjoint and $|C_i| \leq \Rsize$, for all $1 \leq i \leq p$,
    \cf Lemma~\ref{lemma:MarkAndReduce}. Define $\mc P$ to be a largest collection of
    paths such that the following conditions hold. For each path $P \in \mc P$:
    \begin{itemize}
        \item the endpoints of $P$ are both in $X$;
        \item the inner vertices of $P$ are all in a single subgraph
                      $C_i$, for some $1 \leq i \leq p$; and
        \item for all $P' \in \mc P$ with $P' \neq P$, the endpoints of $P$ and $P'$
                      are not identical and their inner vertices
                      are in different subgraphs $C_i$ and $C_j$.
    \end{itemize}
    First, we show that any largest collection $\mathcal{P}$ of paths satisfying
    the above conditions is such that $|\mathcal{P}| = p$, that is, such a collection has
    one path per subgraph in $\set{C_1, \ldots, C_p}$. Assume that $\mathcal{P}$ is a
    largest collection of paths satisfying
    the conditions stated above and consider the graph $H = G|_{\mc P}[X]$ induced by the vertex set $X$
    in the graph $G|_{\mathcal{P}}$ obtained by constricting the paths in $\mathcal{P}$.
    By assumption, $H \in \mc H$ as $\mc G$ $d$-constricts into $\mc H$ and $\mc H$ is
   closed under taking subgraphs. The constant~$d$ is given by
    $$d = \max_{P \in \mc P}|P| \leq \max_{1 \leq i \leq p}|C_i| \leq \Rsize.$$

    Suppose that $|\mc P| < p$, \ie, there exists some $C_i$ for $1 \le i \le p$
    such that no path of $\mc P$ uses vertices of $C_i$. Consider the neighborhood
    $Z = N^G_X(C_i)$ of~$C_i$ in~$X$. As we chose the threshold of the
    marking algorithm to ensure that $|Z| \geq \omega_{\mc H}$, it follows that $Z$ cannot induce a
    clique in $H$. But then there exist vertices $u,v \in Z$ with $uv \not \in E(H)$ and we could
    add a $uv$-path whose inner vertices are in $C_i$ to $\mc P$ without conflicting
    with any of the above constraints (including the bound on $d$), which
    contradicts our assumption that $\mc P$ is of largest size.
    We therefore conclude that $|\mc P| = p$.

    Since there is a bijection from the collection of subgraphs $\set{C_1, \ldots, C_p}$
    and the paths of $\mc P$, we may bound $p$ by the number of edges in $H$,
    which is at most $f_E(|H|)$. But $|H| = |X| = s(k)$ and we thus obtain the
    bound $p \le f_E(s(k))$ on the \emph{number} of large-degree subgraphs witnessed by
    Algorithm~\ref{alg:marking}. Therefore the number of marked bags is
    $|\mc M| \le 2f_E(s(k))$. As every marked bag
    adds at most $t$ vertices to $Y_0$, we obtain the claimed bound
    $$
        |Y_0| = |X| + \big| \, \bigcup_{\mathclap{1\leq i\leq p}}C_i \, \big|
                        \leq s(k) + 2t \cdot f_E(s(k)) = x_k.
    $$\end{proof}

We will now use this bound on the size of $Y_0$ to bound the sizes of the
clusters $Y_1 \uplus \cdots \uplus Y_\ell$ of $G-Y_0$. The important properties
used are that the instance $(G,k)$ is reduced, that each $Y_i$ has a small
neighborhood in $Y_0$ and hence has small size, and that the witness graph
obtained from $G$ via constrictions has a bounded number of cliques given by
the function $f_{\#\omega}(\cdot)$.

\begin{lemma}\label{lemma:BoundEll}
    The number of vertices in $\bigcup_{1\le i\le \ell}Y_i$ is bounded
    by $(f_{\#\omega}(|Y_0|)+|Y_0|+1) \cdot \protd(2t+\omega_{\mc H})$.
\end{lemma}
\begin{proof}
    The clusters $Y_1, \dots, Y_\ell$ contain connected components of $G-Y_0$
    and have the property that for each $1 \leq i \leq \ell$, $N_{Y_0}(Y_i) \leq 2t + \omega_{\mc H}$.
    We proceed analogously to the proof of Lemma~\ref{lemma:BoundY0}.
    Let $\mc P$ be a maximum collection of paths~$P$ such that the endvertices of~$P$
    are in $Y_0$ and all its inner vertices are in some cluster $Y_i$. Moreover
    for all paths $P_1, P_2 \in \mc P$, with $P_1 \neq P_2$, it follows that
    each path has a distinct set of endvertices and a distinct component for
    their inner vertices.
    Consider the graph $H = G|_{\mc P}[Y_0]$ induced by $Y_0$ in the graph obtained from $G$
    by constricting the paths in $\mc P$. Note that each neighborhood
    $Z_i = N^G_{Y_0}(Y_i)$, for $1\leq i\leq \ell$, induces a clique in $H$ as otherwise
     we could augment $\mc P$ by another path.
    As the total number of cliques of graphs in $\mc H$ is bounded by $f_{\#\omega}$, we know that
    $\set{Z_1,\dots,Z_\ell}$ contains at most $f_{\#\omega}(|H|)+|H|+1$
    distinct sets (including the empty and singleton sets).
    Thus
    $$
        \ell \le f_{\#\omega}(|H|)+|H|+1 = f_{\#\omega}(|Y_0|)+|Y_0|+1,
    $$
    where we used the fact that $|H| = |Y_0|$ by construction.
    Since $Y_1,\dots,Y_\ell$ are clusters \wrt $Y_0$, we obtain $\ell$ restricted
        $(2t+\omega_{\mc H})$-protrusions in $G$ (adding the respective neighborhood
    in $Y_0$ to each cluster yields the corresponding $(2t\!+\!\omega_{\mc H})$-protrusion).
    Thus the sets $Y_1,\dots,Y_\ell$ contain in total at most
    $$
        \big| \, \bigcup_{\mathclap{1 \leq i \leq \ell}} Y_i \, \big|
                \leq (f_{\#\omega}(|Y_0|)+|Y_0|+1) \cdot \protd(2t+\omega_{\mc H})
    $$
    vertices.
\end{proof}

We now can easily prove Theorem~\ref{thm:KernelByConstriction}.
\begin{proof}[Proof of Theorem~\ref{thm:KernelByConstriction}.]
    By Lemma~\ref{lemma:BoundY0} we know that $|Y_0| = x_k$.
    Together with Lemma~\ref{lemma:BoundEll} we can bound the
    total number of vertices in a reduced instance by
    \journal{\begin{eqnarray*}
        |V(G)| &=& |\YYYY|  \\
               &\le& x_k + (f_{\#\omega}(x_k)+x_k+1) \protd(2t + \omega_{\mc
               H}),
    \end{eqnarray*}}
    \short{$$
        |V(G)| = |\YYYY| \le x_k + (f_{\#\omega}(x_k)+x_k+1) \protd(2t + \omega_{\mc
        H}),
        $$
    }
    again using the shorthand $x_k = s(k) + 2f_E(s(k))\cdot t$.
\end{proof}
} We now show how to apply Theorem~\ref{thm:KernelByConstriction} to
obtain kernels. Let $\mc G_H$ be the class of graphs that exclude some
fixed graph $H$ as a topological minor. Observe that $\mc G_H$ is closed under
taking topological minors, and is therefore closed under taking
$d$-constrictions for any~$d \geq 2$.

In order to obtain $f_E,f_{\#\omega}$, and $\omega_{\mc G_H}$ we use the fact
that $H$-topological minor free graphs are \emph{$\varepsilon$-degenerate}.
That is, there exists a constant $\varepsilon$ (that depends only on $H$) such
that every subgraph of $G \in \mc G_H$ contains a vertex of degree at most
$\varepsilon$. The following are well-known properties of degenerate
graphs.\looseness-1
\begin{proposition}[Bollob{\'a}s and Thomason\cite{BT98}, Koml{\'o}s and Szemer{\'e}di \cite{KS96a}]\label{prop:HFreeDegree}
    There is a constant $\beta \leq 10$ such that, for $r>2$, every graph with no $K_r$-topological-minor has
    average degree at most $\beta r^2$.\looseness-1
    \end{proposition}
\journal{As an immediate consequence, any graph with average degree larger than~$\beta r^2$
contains \emph{every} $r$-vertex graph as a topological minor.
If a graph~$G$ excludes~$H$ as a topological minor, then~$G$ clearly
excludes~$K_r$ as a topological minor. What is also true is that
the \emph{total} number of cliques (not necessarily maximal) in~$G$ is $O(|V(G)|)$.}
\begin{proposition}[Fomin, Oum, and Thilikos \cite{FOT10}]\label{prop:HFreeCliques}
    There is a constant $\tau < 4.51$ such that, for $r > 2$, every $n$-vertex graph with no
    $K_r$-topological-minor has at most $2^{\tau r \log r} n$ cliques.
\end{proposition}
In the following, let $r := |H|$ denote the size of the forbidden topological
minor. The following is a slightly generalized version of our first main
theorem.

\begin{theorem}
    Fix a graph~$H$ and let $\mc G_H$ be the class of $H$-topological-minor-free graphs.
    Let~$\Pi$ be a parameterized graph-theoretic problem
    that has \fii and is  $(s_{\Pi,\mc G_H},t_{\Pi,\mc G_H})$-treewidth-bounding
    on the class $\mc G_H$. Then $\Pi$
    admits a kernel of size $O(s_{\Pi,\mc G_H}(k))$.\Reduce\Reduce
\end{theorem}
\begin{proof}
    We use Theorem~\ref{thm:KernelByConstriction} with the functions
    $f_E(n) = \frac{1}{2}\beta r^2 n,f_{\#\omega}(n) = 2^{\tau r\log r} n$
    obtained from Propositions~\ref{prop:HFreeDegree} and~\ref{prop:HFreeCliques}.
    Observe that \journal{an $H$-topological-minor-free graph}\short{graphs in $\mc G_H$} cannot contain a clique
    of size $r$, thus $\omega_{\mc G_H} \leq r$. The kernel size is
    then bounded by
    \def\sk{s_{\Pi,\mc G_H}(k)}%
    \def\maxcl{r}%
    \def\clqs{2^{\tau r\log r}}%
    \def\edgs{\beta r^2}%
    \def\t{t}%
    \short{$
        \sk \cdot ( 1 + \edgs\t + (\clqs( 1 + \edgs\t ) +
                          \edgs\t ) \cdot \protd(2\t+\maxcl) ) +
                          \protd(2\t+\maxcl),
    $}%
    \journal{$$
        \sk \cdot ( 1 + \edgs\t + (\clqs( 1 + \edgs\t ) +
                          \edgs\t ) \cdot \protd(2\t+\maxcl) ) +
                          \protd(2\t+\maxcl),
    $$}%
    where we omitted the subscript of $t_{\Pi,\mc G_H}$ for the sake of readability.\Reduce
\end{proof}
\noindent Theorem~I is now just a consequence of the special case for which the
treewidth-bound is linear. Note that the class of graphs with bounded degree is
a subset of those that exclude a fixed topological minor, thus the above result
translates directly to this class.

\subsection{Problems affected by our result\short{.}}\label{subsec:concreteProblems}

We present concrete problems that satisfy the prerequisites of Theorem~I. All
of the following problems are treewidth-bounding with linear
treewidth-modulators.

\begin{corollary}\label{cor:concreteProblems}%
    \journal{Fix a graph~$H$. The following problems are linearly treewidth-bounding and
    have \fii on the class of $H$-topological-minor-free graphs
    and hence possess a linear kernel on this graph class:}%
    \short{The following problems are linearly treewidth-bounding and
    have \fii on $\mc G_H$ and hence admit linear kernels on $\mc G_H$:}
    \VC\footnote{Listed for completeness; these problems have
    a kernel with a linear number of vertices on general graphs.};
    \ClVD\footnotemark[\value{footnote}];
    \FVS;
    \CVD;
    \journal{\Problem{Interval} and \Problem{Proper Interval Vertex Deletion}; }%
    \short{\Problem{Interval}\,and \Problem{Proper\,Interval\,Vertex\,Deletion}; }%
    \Problem{Cograph Vertex Deletion};
    \EDS.
\end{corollary}

In particular, Corollary~\ref{cor:concreteProblems} also implies that \CVD and
\IVD can be decided on \journal{$H$-topological-minor-free graphs}\short{$\mc
G_H$} in time~$O(c^k \cdot \poly(n))$ for some constant~$c$. (This follows
because one can first obtain linear kernel and then use brute-force to solve
the kernelized instance.) On general graphs only an $O(f(k) \cdot \poly(n))$
algorithm is known, where~$f(k)$ is not even specified~\cite{Mar10}.

\begin{corollary}%
    \journal{\CVD and \IVD are solvable in single-exponential time on $H$-topological-minor-free graphs.}%
    \short{\CVD and \IVD are solvable in single-exponential time on $\mc G_H$.}%
\end{corollary}

A natural extension of the (vertex deletion) problems in
Corollary~\ref{cor:concreteProblems} is to seek a solution that induces a
\emph{connected} graph. \short{For certain connected problems (which are
usually more difficult both in terms of proving FPT and polynomial kernels), we
obtain linear kernels.}\journal{The connected versions of problems are
typically more difficult both in terms of proving fixed-parameter tractability
and establishing polynomial kernels. For instance, \textsc{Vertex Cover} admits
a $2k$-vertex kernel but \ConnVC has no polynomial kernel unless~$\NP \subseteq
\coNP/\!\poly$~\cite{DLS09}. However on $H$-topological-minor-free graphs,
\ConnVC (and a couple of others) admit a linear kernel.}
\begin{corollary}%
    \journal{\ConnVC, \ConnCoVD, and \ConnClVD have linear kernels in graphs excluding a fixed topological minor.}%
    \short{\ConnVC, \ConnCoVD and \ConnClVD have linear kernels on graphs of $\mc G_H$.}%
\end{corollary}

\short{ Another property of \journal{$H$-topological-minor-free
graphs}\short{$\mc G_H$} is that the well-known graph width measures treewidth
(\tw), rankwidth (\rw), and cliquewidth (\cw) are all within a constant
multiplicative factor from each other~\cite{FOT10}. The \textsc{Width-$b$
Vertex Deletion} problem~\cite{KPP12} is defined as follows: given a graph~$G$
and an integer~$k$, do there exist at most~$k$ vertices whose deletion results
in a graph with width at most~$b$? For treewidth as the measure, this problem
is--per Defintion~\ref{def:treewidthBounding}--treewidth-bounding. In the case
of $H$-topological-minor-free graphs this also holds for rankwidth and
cliquewidth. The fact that this problem has \fii follows from the sufficiency
condition known as \emph{strong monotonicity} in~\cite{BFLPST09}. Since
branchwidth differs only by a constant factor from treewidth in general
graphs~\cite{RS91}, this gives us the following.
\begin{corollary}
    The \textsc{Width-$b$ Vertex Deletion} problem has a linear kernel on
        $H$-topological-minor-free graphs, where the width measure
    is either treewidth, cliquewidth, branchwidth, or rankwidth.
\end{corollary}

}

\journal{ Another property of $H$-topological-minor-free graphs is that the
well-known graph width measures treewidth (\tw), rankwidth (\rw), and
cliquewidth (\cw), are all within a constant multiplicative factor of one
another.
\begin{proposition}[Fomin, Oum, and Thilikos \cite{FOT10}]\label{prop:HFreeWidth}
    There is a constant
    $\tau$ such that for every $r > 2$, if $G$ excludes $K_r$ as a topological minor, then%
    \journal{$$\begin{aligned}
        \rw (G) \leq \cw (G) & <  2 \cdot 2^{\tau r \log r} \rw (G) \\
        \rw (G) \leq \tw (G) + 1 & < \frac{3}{4}(r^2 + 4r - 5)  2^{\tau r \log r} \rw (G). \\
    \end{aligned}$$}
\end{proposition}

An interesting vertex-deletion problem related to graph width measures is
\textsc{Width-$b$ Vertex Deletion}~\cite{KPP12}: given a graph~$G$ and an
integer~$k$, do there exist at most~$k$ vertices whose deletion results in a
graph with width at most~$b$? From Definition~\ref{def:treewidthBounding} (see
Section~\ref{sec:Preliminaries}), it follows that if the width measure is
treewidth, then this problem is treewidth-bounding. By
Proposition~\ref{prop:HFreeWidth}, this also holds if the width measure is
either rankwidth or cliquewidth. The fact that this problem has \fii follows
from the sufficiency condition known as \emph{strong monotonicity}
in~\cite{BFLPST09}. Since branchwidth differs only by a constant factor from
treewidth in general graphs~\cite{RS91}, this gives us the following.
\begin{corollary}%
    \journal{The \textsc{Width-$b$ Vertex Deletion} problem has a linear kernel on
    $H$-topological-minor-free graphs, where the width measure
    is either treewidth, cliquewidth, branchwidth, or rankwidth.}%
    \short{The \textsc{Width-$b$ Vertex Deletion} problem has a linear kernel on
    $\mc G_H$, where the width measure is either treewidth, cliquewidth, branchwidth, or rankwidth.}%
\end{corollary}
}
\subsection{A comparison with earlier results\short{.}}\label{subsec:comparisonStructural}
\short{Theorem~I requires problems to be treewidth-bounding, at first glance, a
quite strong restriction. However, the property of being treewidth-bounding
appears implicitly or explicitly in previous work on linear kernels on sparse
graphs~\cite{BFLPST09,FLST10}. The linear kernels on bounded genus graphs
in~\cite{BFLPST09} rely on a condition called \emph{quasi-compactness}, which
is similar to that of treewidth-bounding. The latter is a stronger condition in
the sense that if a problem is treewidth-bounding and the graphs are embeddable
on a surface of genus $g$, then the problem is also quasi-compact, but not
vice-versa. The fact that we use a stronger structural condition is expected,
since our result holds on a larger graph class. Interestingly, the
conditions imposed for linear kernels on $H$-minor-free graphs~\cite{FLST10},
namely \emph{bidimensionality} and \emph{separability}, imply
treewidth-boundedness. This lends credence to our belief that being
treewidth-bounding is the crucial ingredient in the proofs of linear kernels
for problems on sparse graphs. }
\journal{
We briefly compare the structural constraints imposed in Theorem~I
with those imposed in the results on linear kernels on graphs
of bounded genus~\cite{BFLPST09} and $H$-minor-free graphs~\cite{FLST10}.
In particular, we discuss how restrictive is the condition
of being treewidth-bounding. A graphical summary of the various notions of sparseness and
the associated structural constraints used to obtain results
on linear kernels is depicted in Figure~\ref{fig:Hierarchy}.

\begin{figure}[t]
    \center\includegraphics[width=.7\textwidth]{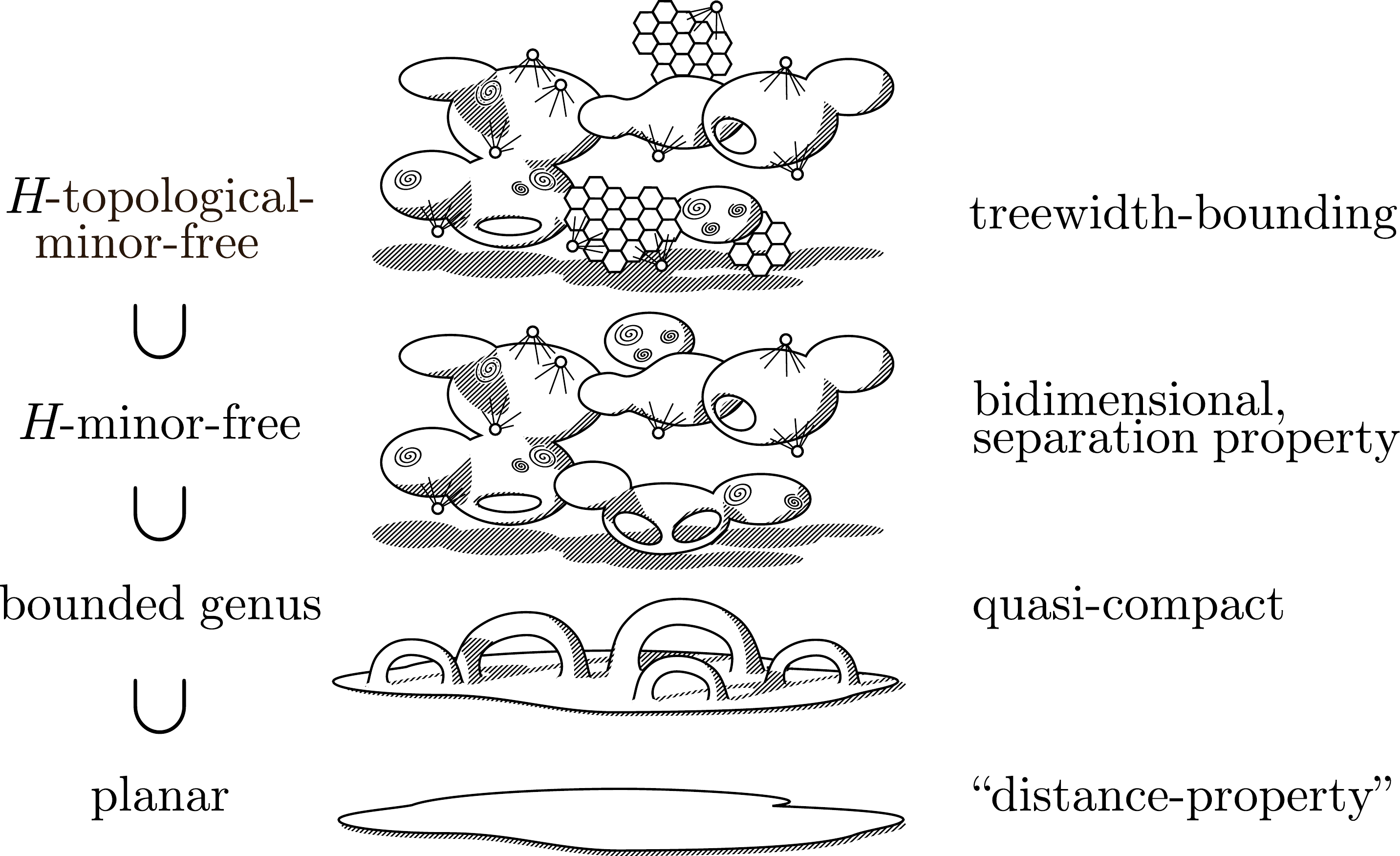}
    \caption{\label{fig:Hierarchy}Kernelization results for problems with finite integer index on
            sparse graph classes with their corresponding additional condition.}
\end{figure}

The theorem that guarantees linear kernels on graphs of bounded genus
in~\cite{BFLPST09} imposes a condition called \emph{quasi-compactness}. The
notion of quasi-compactness is similar to that of treewidth-bounding:
\YES-instances $(G,k)$ satisfy the condition that there exists a vertex set $X
\subseteq V(G)$ of ``small'' size whose deletion yields a graph of bounded
treewidth. Formally, a problem $\Pi$ is called \emph{quasi-compact} if there
exists an integer~$r$ such that for every $(G,k) \in \Pi$, there is an
embedding of~$G$ onto a surface of Euler-genus at most $g$ and a set $X
\subseteq V(G)$ such that $|X| \leq r \cdot k$ and $\tw(G - R^r_G(X)) \leq r$.
Here $R^r_G(X)$ denotes the set of vertices of~$G$ at radial distance at
most~$r$ from~$X$. It is easy to see that the property of being
treewidth-bounding is stronger than quasi-compactness in the sense that if a
problem is treewidth-bounding and the graphs are embeddable on a surface of
genus $g$, then the problem is also quasi-compact, but not the other way
around. The fact that we use a stronger structural condition is expected, since
our result proves a linear kernel on a much larger graph class.

More interesting are the conditions imposed for linear kernels on
$H$-minor-free graphs~\cite{FLST10}. The problems here are required to be
\emph{bidimensional} and satisfy a so-called \emph{separation property}.
Roughly speaking, a problem is bidimensional if the solution size on a $k
\times k$-grid is~$\Omega(k^2)$ and the solution size does not decrease by
deleting/contracting edges. The notion of the separation property is
essentially the following. A problem has the separation property, if for any
graph $G$ and any vertex subset $X \subseteq V(G)$, the optimum solution of~$G$
projected on any subgraph $G'$ of $G-X$ differs from the optimum for $G'$ by at
most~$|X|$ (\cf~\cite{FLST10} for details.) At first glance, these conditions
seem to have nothing to do with the property of being treewidth-bounding.
However in the same paper~\cite[Lemma 3.2]{FLST10}, the authors show that if a
problem on $H$-minor-graphs is bidimensional and has the separation property
then it is also $(ck,t)$-treewidth-bounding for some constants $c,t$ that
depend on the graph $H$ excluded as a minor. Using this fact, the main result
of~\cite{FLST10} (namely, that bidimensional problems with FII and the
separation property have linear kernels on $H$-minor-free graphs) can be
reproved as an easy corollary of Theorem~\ref{thm:KernelByConstriction}.

This discussion shows that in the results on linear kernels
on sparse graph classes that we know so far, the treewidth-bounding
condition has appeared in some form or the other.
In the light of this we feel that this is the key
condition for proving linear kernels on sparse graph classes.
}

\journal{
\subsection{The limits of our approach}\label{subsec:limits}

It is interesting to know for which notions of sparseness
(beyond $H$-topological-minor-free graphs) we can use our technique
to obtain polynomial kernels. We show that our technique \emph{fails}
for the following notion of sparseness: graph classes that locally
exclude a minor~\cite{DGK07}. The notion of locally excluding a minor
was introduced by Dawar \emph{et al.}~\cite{DGK07} and graphs
that locally exclude a minor include bounded-genus graphs
but are incomparable with $H$-minor-free graphs~\cite{NOdM11}.
However we also show that there exists (restricted) graph classes
that locally exclude a minor where it is still possible to
obtain a polynomial kernel using our technique.

\begin{definition}[Locally excluding a minor~\cite{DGK07}]
    A class $\mc G$ of graphs \emph{locally excludes a minor} if for
    every $r \in \mbb N$ there is a graph $H_r$ such that the $r$-neighborhood
    of a vertex of any graph of $\mc G$ excludes $H_r$ as a minor.
\end{definition}

Therefore if $\mc G$ locally excludes a minor then the $1$-neighborhood of a
vertex in any graph of $\mc G$ does not contain $H_1$ as a minor, and hence as
a subgraph. In particular, the neighborhood of no vertex contains a clique on
$h_1 := |H_1|$ vertices as a subgraph, meaning that the clique number of such
graphs is bounded above by $h_1$. The total number of cliques in any graph of
$\mc G$ is then  bounded by $h_1 n^{h_1}$, and the number of edges can be
trivially bounded by $n^2$.  We now have almost all the prerequisites for
applying Theorem~\ref{thm:KernelByConstriction}. However the class $\mc G$ is
not closed under taking $d$-constrictions. Taking a $d$-constriction in a graph
$G \in \mc G$ can increase the clique number of the constricted graph. This
seems to be a bottleneck in applying Theorem~\ref{thm:KernelByConstriction}.
However if we assume that the size of the locally forbidden minors
$\set{H_r}_{r \in \mbb N}$ grows very slowly, then we can still obtain a
polynomial kernel.
\begin{definition}\label{def:AccordingToF}
    Given $g \colon \mbb N \rightarrow \mbb N$, we say that
    a graph class $\mc G$ \emph{locally excludes minors according to $g$}
    if there exists a constant $n_0 \in \mbb N$, such that for all
        $r \geq n_0$, the $g(r)$-neighborhood of a vertex in any graph
    of $\mc G$ does not contain $K_r$ as a minor.
\end{definition}

\begin{lemma}\label{lemma:SlowGrowConstr}
    Let $\mc G$ be a graph class that locally excludes a minor according
    to $g\colon \mbb N \rightarrow \mbb N$ and let $n_0$ be the constant as in the
    above definition. Then for any $r \geq n_0$, the class $\mc G$
        $g(r)$-constricts into a graph class $\mc H$ that excludes $K_r$ as a subgraph.
\end{lemma}
\begin{proof}
    Assume the contrary. Let $G \in \mc G$ and suppose that for some $r \geq 2$ the
    graph $H$ obtained by a $g(r)$-constriction of $G$ contains~$K_r$ as
    a subgraph. Pick any vertex $v$ in this subgraph of $H$. The $g(r)$-neighborhood
    of $v$ in $G$ must contain $K_r$ as a minor, a contradiction.
\end{proof}

Note that in the following, we assume that the problem is treewidth-bounding
on general graphs.

\begin{corollary}\label{cor:SlowGrowKernel}
    Let $\Pi$ be a parameterized graph problem with \fii
    that is $(s(k),t_{\Pi})$-treewidth-bounding. Let $\mc G$
    be a graph class locally excluding a minor according to a function
    $g \colon \mbb N \rightarrow \mbb N$ such that for all $r \geq n_0$,
    $g(r) \geq \protd(2t_\Pi + r) + 1$. Then there exists a constant $r_0$ such that
    $\Pi$ admits kernels of size $O(s(k)^{r_0})$ on $\mc G$.
\end{corollary}
\begin{proof}
    By Lemma~\ref{lemma:SlowGrowConstr}, taking a $g(r)$-constriction
    results in a graph class $\mc H$ that excludes $K_r$ as a subgraph,
    for large enough $r$.
    Fixing $r = n_0$, where $n_0$ is the constant in Definition~\ref{def:AccordingToF},
    we apply Theorem~\ref{thm:KernelByConstriction} with the
    trivial functions $f_E(n) = n^2$, $f_{\#\omega}(n) = r \cdot n^r$ and
        $\omega_{\mc H} = r$. By Lemma~\ref{lemma:SlowGrowConstr}, we have that
    $\omega_{\mc G_H} \leq r$. The kernel size is then bounded by
    $$
        s(k) + s(k)^2 2t + \left(( s(k) + 2t \cdot s(k)^2)^r+ s(k)
                + 2t \cdot s(k)^2 + 1 \right ) \protd(2t + r)
        \in O(s(k)^{2r}),
    $$
     where we omitted the subscript of $t_{\Pi}$ for the sake of readability. With
    $r_0 = 2r  = 2n_0 $, the bound in the statement of the corollary follows.
\end{proof}

We do not know how quickly the function $\protd(\cdot)$
grows but intuition from automata theory seems to suggest that
this has at least superexponential growth. As such, the graph class for
which the polynomial kernel result holds (Corollary~\ref{cor:SlowGrowKernel})
is pretty restricted. However this does suggest a limit to which our approach
can be pushed as well as some intuition as to why our result is not easily extendable
to graph classes locally excluding a minor. We note that graph classes of bounded expansion
present the same problem.

\subsection{An illustrative example: Edge Dominating Set} \label{subsec:exampleKernel}

In this section we show how Theorem~I can actually be used to
obtain a simple explicit kernel for the \EDS problem
on $H$-topological-minor-free graphs. This is made possible by the fact that
we can find in polynomial time a small enough treewidth-modulator \emph{and}
replace the generic protrusion reduction rule by a handcrafted specific
reduction rule.

Let us first recall the problem at hand.
We say that an edge $e$ is dominated by a set of edges $D$ if either
$e\in D$ or $e$ is incident with at least one edge in $D$. The problem
\EDS asks, given a graph $G$ and an integer $k$,
whether there is an edge dominating set $D\subseteq E(G)$ of size at most $k$,
i.e., an edge set which dominates every edge of $G$. The canonical parameterization
of this problem is by the integer $k$, \ie the size of solution set.

There is a simple 2-approximation algorithm for \textsc{Edge Dominating
Set}~\cite{YG80}. Given an instance $(G,k)$, where G is
$H$-topological-minor-free, let $D$ be an edge dominating set of $G$, given by
the 2-approximation. We can assume that $|D|\leqslant 2k$ since
otherwise we can correctly declare $(G,k)$ as a \NO-instance. Take
$X:=\{v\in V(G) \mid v \text{ is incident to some edge in }D\}$ as the
treewidth-modulator: note that $|X|\leq 4k$ and that $G-X$ is of treewidth at most $0$,
i.e., an independent set. One can easily verify that
that the bag marking Algorithm~\ref{alg:marking} of
Section~\ref{sec:Decomposition} would mark exactly those vertices of
$G-X$ whose neighborhood in $X$ has size at least $r := |H|$. By applying
the edge-bound of Proposition~\ref{prop:HFreeDegree} to
Lemma~\ref{lemma:BoundY0} we get that $|V(\mc{M})| \leqslant \beta r^2 \cdot 8k$.

Take $Y_0:=X\cup V(\mc{M})$ and let $\mc{P}:=\YYYY$
be a partition of $V(G)$, where again $Y_i$, $1\leqslant i\leqslant q$,
 is now a cluster \wrt $Y_0$, \ie the vertices in a single $Y_i$
share the same neighborhood in $X$ and the $Y_i$ are of maximal size under this
condition. We have one reduction rule, which can be construed as an
concrete instantiation of generic the protrusion replacement rule. We would
like to stress that this reduction rule relies on the fact that we already have
a protrusion decomposition of $G$, given by Algorithm~\ref{alg:marking}.

\begin{quote}
    \short{\vspace{-3mm}}
    \textbf{Twin elimination rule}: If $|Y_i|>|N_{Y_0}(Y_i)|$ for some $i\neq 0$,
    let $G'$ be the instance obtained by keeping $|N_{Y_0}(Y_i)|$ many vertices of
    $Y_i$ and removing the rest of $Y_i$. Take $k':=k$.
    \short{\vspace{-3mm}}
\end{quote}

\begin{lemma}\label{lem:twinEliminationRule}
    The twin elimination rule is safe.
\end{lemma}
\begin{proof}
    Let $G_i$ be the graph induced by the vertex set $N_{Y_0}(Y_i)\cup Y_i$ and let
    $E_i$ be its edge set (as $N_{Y_0}(Y_i)=N(Y_i)$, we shall omit the subscript
    $Y_0$). For a vertex $v\in V(G)$, we define the set $E(v)$ as the set of edges
    incident with $v$. The notations $G'_i$, $E'_i$, $Y'_i$, and $E'(v)$ are
    defined analogously for the graph $G'$ obtained after the application of
    twin elimination rule. We say that a vertex $v\in V(G)$ is
    \emph{covered} by an edge set $D$ if $v$ is incident with an edge of $D$.

    To see the forward direction, suppose that $(G,k)$ is a \textsc{Yes}-instance
    and let $D$ be an edge dominating set of size at most $k$. Without loss of
    generality, we can assume that $|D\cap E_i|\leqslant |N(Y_i)|$. Indeed, it can
    be easily checked that the edge set $(D\setminus E_i)\cup E(u)$, for an
    arbitrarily chosen $u\in Y_i$, is an edge dominating set. Hence at most
    $|N(Y_i)|$ vertices out of $Y_i$ are covered by $D$, and thus we can apply
    twin elimination rule so as to delete only those vertices which are
    not incident with $D$. It just remains to observe that $D$ is an edge
    dominating set of $G'$.

    For the opposite direction, let $D'$ be an edge dominating set for $G'$ of size
    at most $k$. We first argue that $N(Y'_i)$ is covered by $D'$ without loss of
    generality. Indeed, suppose $v\in N(Y'_i)$ is not covered by $D'$. In order for
    an edge $e=uv\in E'(v)\cap E'_i$ to be dominated by $D'$, at least one edge in
    $E'(u)$ should be contained in $D$. Since the sets $\{E'(u): u\in Y'_i\}$ are
    mutually disjoint, it follows that $|D'\cap E'_i|\geqslant |Y'_i|$. Now take an
    alternative edge set $D'':=(D'\setminus E'_i)\cup E'(u)$ for an arbitrary
    vertex $u\in Y'_i$. It is not difficult to see that $D''$ is an edge dominating
    set for $G'$. Moreover, we have $|D''|\leqslant|D'|\leqslant k$ as $|D'\cap
    E'_i|\geqslant |Y'_i|=|E'(u)|=|N(Y'_i)|$. Hence $D''$ is also an edge
    dominating set of size at most $k$. Assuming that $N(Y'_i)$ is covered by $D'$,
    it is easy to see that $D'$ dominates $E_i$ and thus $D'$ is an edge dominating
    set of $G$. This complete the proof.
\end{proof}

\def\YS{(2\beta r^2  + 1)4k}
Back to the partition $\mc{P}$, we can apply twin elimination rule in
time $O(n)$ and ensure that $|Z_i|\leqslant r-1$ for $1\leqslant i \leqslant
q$. The bound on $q$ is proved in Lemma~\ref{lemma:BoundEll} and taken together
with the edge- and clique-bounds from Proposition~\ref{prop:HFreeDegree} and~\ref{prop:HFreeCliques},
respectively, we obtain
$$
    q \leq 2^{\tau r\log r}(\YS) + \YS + 1
$$
and thus we get the overall bound
\def\cl{2^{\tau r\log r}}
\def\edg{2\beta r^2}
\def\pot{20.8}
\journal{
\begin{eqnarray*}
    |G| &\leq& |Y_0|+|Y_1|+\dots+|Y_q| \\
        &=& (\edg+1)4k + (\cl ((\edg+1)4k) + ((\edg+1)4k) + 1)(r-1) \\
        &=& 4k \left( 1 + 2\beta r^2 + (2^{\tau r\log r + 1}\beta r^2+2^{\tau r\log r}+2\beta r^2)(r-1) \right) + r - 1 \\
        &<& 4k \left( 1 + 20r^2 + ( \pot^{r\log r+1} 20r^2+\pot^{r\log r}+20r^2)(r-1) \right) + r
\end{eqnarray*}
}
\short{
\begin{eqnarray*}
    |G| &\leq& (\edg+1)4k + (\cl ((\edg+1)4k) + ((\edg+1)4k) + 1)(r-1) \\
        &<& 4k \left( 1 + 20r^2 + ( \pot^{r\log r+1} 20r^2+\pot^{r\log r}+20r^2)(r-1) \right) + r
\end{eqnarray*}
}
on the size of $G$.
We remark that this upper bound can be easily made explicit
once $H$ is fixed.  Again, we can get better constants on $H$-minor-free
graphs, just by replacing constants $\beta r^2$ and $2^{\tau r\log r}$ with $\alpha(r\sqrt{\log r})$ and
$2^{\mu r \log\log r}$, respectively. Finally, note that the whole procedure can be carried
out in linear time.
}

\section{Single-exponential algorithm for Planar-$\mathcal F$-Deletion}\label{sec:PlanarF}
This section is devoted to the single-exponential algorithm
for the \planarF{} problem.  Let henceforth $H_p$ be some fixed (connected or
disconnected) arbitrary planar graph in the family $\mathcal{F}$, and let
\journal{$r := |H_p|$}\short{$r := |V(H_p)|$}. First of all, using iterative
compression, we reduce the problem to obtaining a single-exponential algorithm
for the \DplanarF{} problem, which is defined as follows: \short{given a graph
$G$ and a subset of vertices $X \subseteq V(G)$ such that $G-X$ is
$H$-minor-free for every $H\in \mc{F}$, compute a set $\tilde{X}\subseteq V(G)$
disjoint from $X$ such that $|\tilde{X}|<|X|$ and $G-\tilde{X}$ is
$H$-minor-free for every $H\in \mc{F}$, if such a set exists. The parameter is
$k = |X|$. }

\journal{
\begin{tabbing}
\hspace{1cm} \= {\sc \DplanarF{}}\\
\> {\bf Input:} \hspace{1cm} \= \parbox[t]{12cm}{ A graph $G$ and a subset of vertices $X\subseteq V(G)$ such that $G-X$ is $H$-minor-free for every $H\in \mc{F}$.}\\
\> {\bf Parameter:} \> \parbox[t]{12cm}{The integer $k=|X|$.}\\
\> {\bf Objective:} \> \parbox[t]{12cm}{Compute a set $\tilde{X}\subseteq V(G)$
disjoint from $X$ such that $|\tilde{X}|<|X|$ and $G-\tilde{X}$ is
$H$-minor-free for every $H\in \mc{F}$, if such a set exists.}
\end{tabbing}

The input set $X$ is called the {\em initial solution} and the set
$\tilde{X}$ the {\em alternative solution}. Let $t_{\mathcal{F}}$ be a constant
(depending on the family $\mathcal{F}$) such that $\tw(G-X)\leqslant
t_{\mathcal{F}}-1$ (note that such a constant exists by Robertson and
Seymour~\cite{RS86}).

The following lemma relies on the fact that being $\mc{F}$-minor-free is a
hereditary property with respect to induced subgraphs. We omit the proof as it
is now a classical statement (the interested reader can refer, for example,
to~\cite{CFLLV08,LSS09,JPSST11,KPP12}).}

\short{The input set $X$ is called the {\em initial solution} and the set
$\tilde{X}$ the {\em alternative solution}.  Let $t_{\mathcal{F}}$ be a
constant (depending on the family $\mathcal{F}$) such that $\tw(G-X)\leqslant
t_{\mathcal{F}}-1$ (note that such a constant exists by Robertson and
Seymour~\cite{RS86}). The following lemma relies on the fact that being
$\mc{F}$-minor-free is a hereditary property with respect to induced subgraphs.
We omit the proof as it is now a classical statement (see for
instance~\cite{CFLLV08,LSS09,JPSST11,KPP12}).}

\begin{lemma} \label{lem:ic}
If the parameterized \DplanarF~problem can be solved in time $c^k \cdot p(n)$,
where $c$ is a constant and $p(n)$ is a polynomial in $n$, then the
parameterized \planarF~problem can be solved in time $(c+1)^k \cdot p(n) \cdot
n$.
\end{lemma}


\short{To solve \DplanarF{}, our approach starts by computing a protrusion
decomposition using Algorithm~\ref{alg:marking} with input $(G,X,r)$. But it
turns out that the set $Y_0$ output by Algorithm~\ref{alg:marking} does not
define a {\sl linear} protrusion decomposition of $G$, which is crucial for our
purposes (in fact, it can be only proved that $Y_0$ defines a {\sl quadratic}
protrusion decomposition of $G$). To circumvent this problem, our strategy is
to guess the intersection $I$ of the alternative solution $\tilde{X}$ with the
set $Y_0$. As a result, we obtain Proposition~\ref{prop:protrusion
decomposition}, which is fundamental in order to prove Theorem~II.

\begin{proposition}[\textbf{Linear protrusion decomposition}]
\label{prop:protrusion decomposition} Let $(G,X,k)$ be a \textsc{Yes}-instance
of the parameterized \DplanarF{} problem. There exists a $2^{O(k)} \cdot
n$-time algorithm that identifies a set $I \subseteq V(G)$ of size at most $k$
and a $(O(k),2t_{\mathcal{F}}+r)$-protrusion decomposition $\mc{P}=Y_0\uplus
Y_1\uplus\cdots\uplus Y_{\ell}$ of $G-I$ such that: $(1)$~$X\subseteq Y_0$; and
$(2)$~there exists a set $X' \subseteq V(G) \setminus Y_0$ of size at most
$k-|I|$ such that $G-\tilde{X}$, with $\tilde{X}=X'\cup I$, is $H$-minor-free
for every graph $H\in\mc{F}$.
\end{proposition}

Towards the proof of Proposition~\ref{prop:protrusion decomposition}, we need
the following ingredients.

\begin{proposition}[Thomason~\cite{Tho01}]
\label{prop:numberEdges} There is a constant $\alpha < 0.320$ such that every
$n$-vertex graph with no $K_r$-minor has at most $(\alpha r \sqrt{\log r})
\cdot n$ edges.
\end{proposition}

\begin{proposition}[Fomin, Oum, and Thilikos~\cite{FOT10}]
\label{prop:numberCliques} There is a constant $\mu < 11.355$ such that, for
$r> 2$, every $n$-vertex graph with no $K_r$-minor has at most $2^{\mu r \log
\log r} \cdot n$ cliques.
\end{proposition}

For the sake of simplicity, let henceforth $\alpha_r := \alpha r \sqrt{\log r}$
and $\mu_r := 2^{\mu r \log \log r}$. For each guessed set $I\subseteq Y_0$, we
denote $G_I:=G-I$. The proofs of the following lemmas use arguments similar to
those of the proof of Theorem~\ref{thm:KernelByConstriction}.

\begin{lemma}\omitted \label{lem:numberOfMarkedNodes}
If $(G,X,k)$ is a \textsc{Yes}-instance of the \DplanarF{} problem, then the
set $Y_0=V(\mc{M})\cup X$ of vertices returned by Algorithm~\ref{alg:marking}
has size at most $k + 2 t_{\mathcal{F}} \cdot (1 + \alpha_r) \cdot k$.
\end{lemma}

\begin{lemma}\omitted \label{lem:smallComponents}
If $(G_I,Y_0\setminus I,k-|I|)$ is a \textsc{Yes}-instance of the \DplanarF{}
problem, then the number of clusters of $G_I-Y_0$ is at most
$(5t_{\mathcal{F}}\alpha_r \mu_r) \cdot k$.
\end{lemma}\Reduce

We are now ready to prove Proposition~\ref{prop:protrusion decomposition}.\Reduce\reduce

\begin{proof}[Proof of Proposition~\ref{prop:protrusion
decomposition}] By Lemmas~\ref{lem:algo-running-time}
and~\ref{lem:numberOfMarkedNodes}, we can compute in linear time a set $Y_0$ of
$O(k)$ vertices containing $X$ such that every cluster of $G-Y_0$ is a
restricted $(2t_{\mathcal{F}} + r)$-protrusion. If $(G,X,k)$ is a
\textsc{Yes}-instance of the \DplanarF{} problem, then there exists a set
$\tilde{X}$ of size at most $|X|$ and disjoint from $X$ such that $G-\tilde{X}$
does not contain any graph $H\in\mc{F}$ as a minor. Branching on every possible
subset of $Y_0\setminus X$, one can guess the intersection $I$ of $\tilde{X}$
with $Y_0\setminus X$. By Lemma~\ref{lem:numberOfMarkedNodes}, the branching
degree is $2^{O(k)}$. As $(G,X,k)$ is a \textsc{Yes}-instance, for at least one
of the guessed subsets $I$, the instance $(G_I,Y_0\setminus I,k-|I|)$ is a
\textsc{Yes}-instance of the  \DplanarF{} problem. Now, by
Lemma~\ref{lem:smallComponents}, the partition $\mc{P}=(Y_0\setminus I)\uplus
Y_1\uplus\cdots \uplus Y_{\ell}$, where $\{Y_1,\dots,Y_{\ell}\}$ is the set of
clusters of $G_I - Y_0$, is an $(O(k),r+2t_{\mathcal{F}})$-protrusion
decomposition of $G_I$.\looseness-1
\end{proof}

After having proved Proposition~\ref{prop:protrusion decomposition}, we can now
focus on solving \DplanarF{} in single-exponential time when a linear
protrusion decomposition is given. The key observation is that for every
restricted protrusion $Y_i$, there is a {\sl finite} number of representatives
such that any partial solution lying on $Y_i$ can be replaced with one of them
while preserving the feasibility of the solution. This follows from the
\emph{finite index} of MSO-definable properties (see, e.g.,~\cite{BvF01}).
Then, to solve the problem in single-exponential time we can just use
brute-force in the union of these representatives, which has overall size
$O(k)$. Full details about this procedure can be found in the Appendix.


\begin{proposition}\omitted \label{prop:sol-decomp-single}
Let $(G,Y_0,k)$ be an instance of \DplanarF{} and let $\mc{P}=Y_0\uplus
Y_1\uplus \cdots \uplus Y_{\ell}$ be an $(\alpha,\beta)$-protrusion
decomposition of $G$, for some constant $\beta$. There exists an
$O(2^{\ell}\cdot n)$-time algorithm which computes a solution $\tilde{X}
\subseteq V(G)\setminus Y_0$ of size at most $k$ if it exists, or correctly
decides that there is no such solution.
\end{proposition}\Reduce

We finally have all the ingredients to piece everything together and prove
Theorem~II.\Reduce\reduce

\begin{proof}[Proof of Theorem~II] Lemma~\ref{lem:ic} states that
\planarF{} can be reduced to \DplanarF{} so that
the former is single-exponential time solvable provided that the latter is,
and the degree of the polynomial function in $n$ increases by one. We now proceed
to solve \DplanarF{} in time $2^{O(k)}\cdot n$. Given an instance $(G,X,k)$ of
\DplanarF{}, we apply Proposition~\ref{prop:protrusion decomposition} to either
correctly decide that $(G,X,k)$ is a \textsc{No}-instance, or identify in time
$2^{O(k)} \cdot n$ a set $I \subseteq V(G)$ of size at most $k$ and a
$(O(k),2t_{\mathcal{F}}+r)$-protrusion decomposition $\mc{P}=Y_0\uplus
Y_1\uplus\cdots\uplus Y_{\ell}$ of $G-I$, with $X\subseteq Y_0$, such that
there exists a set $X' \subseteq V(G) \setminus Y_0$ of size at most $k-|I|$
such that $G-\tilde{X}$, with $\tilde{X}=X'\cup I$, is $H$-minor-free for every
graph $H\in\mc{F}$. Finally, using Proposition~\ref{prop:sol-decomp-single} we
can solve the instance $(G_I,Y_0\setminus I, k-|I|)$ in time $2^{O(k)}\cdot n$.
\end{proof}}


\journal{Let us provide a brief sketch of our algorithm to solve \DplanarF{}.
We start by computing a protrusion decomposition using
Algorithm~\ref{alg:marking} with input $(G,X,r)$. But it turns out that the set
$Y_0$ output by Algorithm~\ref{alg:marking} does not define a {\sl linear}
protrusion decomposition of $G$, which is crucial for our purposes (in fact, it
can be only proved that $Y_0$ defines a {\sl quadratic} protrusion
decomposition of $G$). To circumvent this problem, our strategy is to first use
Algorithm~\ref{alg:marking} to identify a set $Y_0$ of $O(k)$ vertices of $G$,
and then guess the intersection $I$ of the alternative solution $\tilde{X}$
with the set $Y_0$. We prove that if the input is a \textsc{Yes}-instance of
\DplanarF{}, then $V(\mc{M})$ contains a subset $I$ such that the connected
components of $G-V(\mc{M})$ can be clustered together with respect to their
neighborhood in $Y_0 \setminus I$ to form an
$(O(k-|I|),2t_{\mathcal{F}}+r)$-protrusion decomposition $\mc{P}$ of the graph
$G-I$. As a result, we obtain Proposition~\ref{prop:protrusion decomposition},
which is fundamental in order to prove Theorem~II.

\begin{proposition}[\textbf{Linear protrusion decomposition}]
\label{prop:protrusion decomposition} Let $(G,X,k)$ be a \textsc{Yes}-instance
of the parameterized \DplanarF{} problem. There exists a $2^{O(k)} \cdot
n$-time algorithm that identifies a set $I \subseteq V(G)$ of size at most $k$
and a $(O(k),2t_{\mathcal{F}}+r)$-protrusion decomposition $\mc{P}=Y_0\uplus
Y_1\uplus\cdots\uplus Y_{\ell}$ of $G-I$ such that: \vspace{-1mm}
\begin{enumerate}\setlength{\itemsep}{-.0pt}
\item $X\subseteq Y_0$;
\item there exists a set $X' \subseteq V(G) \setminus Y_0$ of size at most $k-|I|$ such that $G-\tilde{X}$, with $\tilde{X}=X'\cup I$, is
$H$-minor-free for every graph $H\in\mc{F}$.
\end{enumerate}
\end{proposition}

%
%
%

At this stage of the algorithm, we can assume that a subset $I$ of the
alternative solution $\tilde{X}$ has been identified, and it remains to solve
the instance $(G-I,X,k-|I|)$ of the \DplanarF{} problem, which comes equipped
with
a linear protrusion decomposition $\mc{P}=Y_0\uplus Y_1\uplus\cdots\uplus
Y_{\ell}$.
In order to solve this problem, we prove the following proposition:

\begin{proposition}\label{prop:sol-decomp-single}
Let $(G,Y_0,k)$ be an instance of \DplanarF{} and let $\mc{P}=Y_0\uplus
Y_1\uplus \cdots \uplus Y_{\ell}$ be an $(\alpha,\beta)$-protrusion
decomposition of $G$, for some constant $\beta$. There exists an
$2^{O(\ell)}\cdot n$-time algorithm which computes a solution $\tilde{X}
\subseteq V(G)\setminus Y_0$ of size at most $k$ if it exists, or correctly
decides that there is no such solution.
\end{proposition}

The key observation in the proof of Proposition~\ref{prop:sol-decomp-single} is
that for every restricted protrusion $Y_i$, there is a {\sl finite} number of
representatives such that any partial solution lying on $Y_i$ can be replaced
with one of them while preserving the feasibility of the solution. This follows
from the \emph{finite index} of MSO-definable properties (see,
e.g.,~\cite{BvF01}). Then, to solve the problem in single-exponential time we
can just use brute-force in the union of these representatives, which has
overall size $O(k)$.


\paragraph{Organization of the section.}
In Subsection~\ref{sec:algoAnalysis} we analyze Algorithm~\ref{alg:marking}
when the input graph is a \textsc{Yes}-instance of \DplanarF{}. The branching
step guessing the intersection of the alternative solution $\tilde{X}$ with
$V(\mc{M})$ is described in Subsection~\ref{sec:BranchAndProtrusion},
concluding the proof of Proposition~\ref{prop:protrusion decomposition}.
Subsection~\ref{sec:colorvec} gives a proof of
Proposition~\ref{prop:sol-decomp-single}, and finally
Subsection~\ref{sec:proofBigTheorem} proves Theorem~II.

\subsection{Analysis of the bag marking algorithm} \label{sec:algoAnalysis}


We first need two results concerning graphs with excluding clique minor. The
following lemma states that graphs excluding a fixed graph as a minor have
linear number of edges.

\begin{proposition}[Thomason~\cite{Tho01}]
\label{prop:numberEdges} There is a constant $\alpha<0.320$ such that every
$n$-vertex graph with no $K_r$-minor has at most $(\alpha r \sqrt{\log r})
\cdot n$ edges.
\end{proposition}

Recall that a \emph{clique} in a graph is a set of pairwise adjacent vertices.
For simplicity, we assume that a single vertex and the empty graph are also
cliques.

\begin{proposition}[Fomin, Oum, and Thilikos~\cite{FOT10}]
\label{prop:numberCliques} There is a constant $\mu < 11.355$ such that, for
$r> 2$, every $n$-vertex graph with no $K_r$-minor has at most $2^{\mu r \log
\log r} \cdot n$ cliques.
\end{proposition}

For the sake of simplicity, let henceforth in this section $\alpha_r := \alpha
r \sqrt{\log r}$ and $\mu_r := 2^{\mu r \log \log r}$.


Let us now analyze some properties of Algorithm~\ref{alg:marking} when the
input graph is a \textsc{Yes}-instance of the \DplanarF{} problem. In this
case, the bound on the treewidth of $G-X$ is $t_{\mathcal{F}} - 1$. The
following two lemmas show that the number of bags identified at the
``Large-subgraph marking step'' is linearly bounded by $k$. Their proofs use
arguments similar to those used in the proof of
Theorem~\ref{thm:KernelByConstriction}, but we provide the full proofs here for
completeness.

\begin{lemma}\label{lem:connected-subgraphs}
Let $(G,X,k)$ be a \textsc{Yes}-instance of the \DplanarF{} problem. If
$C_1,\ldots, C_{\ell}$ is a collection of connected pairwise disjoint subsets
of $V(G)\setminus X$ such that for all $1\leqslant i\leqslant \ell$,
$|N_X(C_i)|\geqslant r$, then $\ell \leqslant (1 + \alpha_r) \cdot k$.
\end{lemma}
\begin{proof}
Let $X' \subseteq V(G) \setminus X$ be a solution for $(G,X,k)$, and observe
that $\ell'\leqslant k$ of the sets $C_1,\ldots, C_{\ell}$ contain vertices of
$X'$. Consider the sets $C_{\ell'+1},\dots, C_{\ell}$ which are disjoint with
$X'$, and observe that $G[X\cup(\bigcup_{\ell'<j \leqslant \ell} C_j)]$ is an
$H$-minor-free graph. We proceed to construct a family of graphs
$\{G_i\}_{\ell' \leqslant i\leqslant \ell}$, with $V(G_i) = X$ for all $\ell'
\leqslant i\leqslant \ell$, and such that $G_{i}$ is a minor of
$G[X\cup(\bigcup_{\ell'<j\leqslant i} C_j)]$, in the following way. We start
with $E(G_{\ell'}) = E[G(X)]$, and suppose inductively that the graph $G_{i-1}$
has been successfully constructed. Since by assumption $G_{i-1}$ is a minor of
$G[X\cup(\bigcup_{\ell'<j\leqslant i-1} C_j)]$, which in turn is a minor of
$G[X\cup(\bigcup_{\ell'<j\leqslant \ell} C_j)]$, it follows that $G_{i-1}$ is
$H$-minor-free, and therefore it cannot contain a clique on $r$ vertices. In
order to construct $G_i$ from $G_{i-1}$, let $x_i,y_i$ be two vertices in $X$
such that both $x_i$ and $y_i$ are neighbors in $G$ of some vertex in $C_i$,
and such that $x_i$ and $y_i$ are non-adjacent in $G_{i-1}$. Note that such two
vertices exist, since we can assume that $r \geqslant 2$ and $G_{i-1}$ is
$H$-minor-free. Then $G_i$ is constructed from $G_{i-1}$ by adding an edge
between $x_i$ and $y_i$. Since $C_i$ is connected by hypothesis, we have that
$G_{i}$ is indeed a minor of $G[X\cup(\bigcup_{\ell'<j\leqslant i} C_j)]$.
Since $G_{\ell}$ is $H$-minor-free, it follows by
Proposition~\ref{prop:numberEdges} that $|E(G_{\ell})| \leqslant \alpha_r \cdot
|X|$ edges. Since by construction we have that $\ell - \ell' \leqslant
|E(G_{\ell})|$, we conclude that $\ell = \ell' + (\ell - \ell') \leqslant k  +
\alpha_r \cdot k = (1 + \alpha_r) \cdot k$, as we wanted to prove.\end{proof}

\begin{lemma} \label{lem:numberOfMarkedNodes}
If $(G,X,k)$ is a \textsc{Yes}-instance of the \DplanarF{} problem, then the
set $Y_0=V(\mc{M})\cup X$ of vertices returned by Algorithm~\ref{alg:marking}
has size at most $k + 2 t_{\mathcal{F}} \cdot (1 + \alpha_r) \cdot k$.
\end{lemma}
\begin{proof}
As $|X|=k$ and as the algorithm marks bags of an optimal forest-decomposition
of $G-X$, which is a graph of treewidth at most $t_{\mathcal{F}}$, in order to
prove the lemma it is enough to prove that the number of marked bags is at most
$2 \cdot (1 + \alpha_r) \cdot k$. It is an easy observation to see that the set
of connected components $C_B$ identified at the Large-subgraph marking step
contains pairwise vertex disjoint subset of vertices, each inducing a connected
subgraph of $G-X$ with at least $r$ neighbors in $X$. It follows by
Lemma~\ref{lem:connected-subgraphs}, that the number of bags marked at the
Large-subgraph marking step is at most $(1 + \alpha_r) \cdot k$. To conclude it
suffices to observe that the number of bags identified at the LCA marking step
cannot exceed the number of bags marked at the Large-subgraph marking step.
\end{proof}

\subsection{Branching step and linear protrusion decomposition}
\label{sec:BranchAndProtrusion}

At this stage of the algorithm, we have identified a set $Y_0=X\cup V(\mc{M})$
of $O(k)$ vertices such that, by Proposition~\ref{prop:decompositionSparse},
every connected component of $G-Y_0$ is a restricted
$(2t_{\mathcal{F}}+r)$-protrusion. We would like to note that it can be proved,
using ideas similar to the proof of Lemma~\ref{lem:smallComponents} below, that
$Y_0$ together with the clusters of $G-Y_0$ form a {\sl quadratic} protrusion
decomposition of the input graph $G$. But as announced earlier, for time
complexity issues we seek a {\sl linear} protrusion decomposition. To this end,
the second step of the algorithm consists in a branching to guess the
intersection $I$ of the alternative solution $\tilde{X}$ with the set of marked
vertices $V(\mathcal{M})$. By Lemma~\ref{lem:numberOfMarkedNodes}, this step
yields $2^{O(k)}$ branchings, which is compatible with the desired
single-exponential time.

For each guessed set $I\subseteq Y_0$, we denote $G_I:=G-I$. Recall that a
cluster of $G_I-Y_0$ as a maximal collection of connected components of
$G_I-Y_0$ with the same neighborhood in $Y_0 \setminus I$. We use
Observation~\ref{obs:cluster}, a direct consequence of
Lemma~\ref{lem:protrusion},  to bound the number of clusters under the
condition that $G_I$ contains a vertex subset $X'$ disjoint from $Y_0$ of size
at most $k-|I|$ such that $G_I-X'$ does not contain any graph $H\in\mc{F}$ as a
minor (and so the graph $G-\tilde{X}$, with $\tilde{X}=X'\cup I$, does not
contain either any graph $H\in\mc{F}$ as a minor).

\begin{observation} \label{obs:cluster}
For every cluster $\mc{C}$ of $G_I-Y_0$, $|N_{Y_0}(\mc{C})| <
r+2t_{\mathcal{F}}$.
\end{observation}

The proof of the following lemma has a similar flavor to those of
Theorem~\ref{thm:KernelByConstriction} and Lemma~\ref{lem:connected-subgraphs}.

\begin{lemma} \label{lem:smallComponents}
If $(G_I,Y_0\setminus I,k-|I|)$ is a \textsc{Yes}-instance of the \DplanarF{}
problem, then the number of clusters of $G_I-Y_0$ is at most
$(5t_{\mathcal{F}}\alpha_r \mu_r) \cdot k$.
\end{lemma}
\begin{proof}
Let $\mathcal{C}$ be the collection of all clusters of $G_I - Y_0$. Let $X'$ be
a subset of vertices disjoint from $Y_0$ such that $|X'|\leqslant k-|I|=k_I$
and $G_I-X'$ is $H$-minor-free for every graph $H\in\mc{F}$. Observe that at
most $k_I$ clusters in $\mc{C}$ contain vertices from $X'$. Let
$C_1,\ldots,C_{\ell}$ be the clusters in $\mathcal{C}$ that do not contain
vertices from $X'$. So we have that $|\mathcal{C}| \leqslant k_I + \ell
\leqslant k + \ell$. Let $G_{\mathcal{C}}$ be the subgraph of $G$ induced by
$(Y_0 \setminus I) \cup \bigcup_{i=1}^{\ell}C_i$. Observe that as
$(G_I,Y_0\setminus I,k-|I|)$ is a \textsc{Yes}-instance of the \DplanarF{}
problem, $G_{\mathcal{C}}$ is $H$-minor-free for every graph $H\in\mc{F}$.

We greedily construct from $G_{\mathcal{C}}$ a graph $G_{\mathcal{C}}'$, with
$V(G_{\mathcal{C}}')=Y_0 \setminus I$, as follows. We start with
$G_{\mathcal{C}}' = G[Y_0 \setminus I]$. As long as there is a non-used cluster
$C\in \mathcal{C}$ with two non-adjacent neighbors $u,v$ in $Y_0 \setminus I$,
we add to $G_{\mathcal{C}}'$ an edge between $u$ and $v$ and mark $C$ as used.
The number of clusters in $\mathcal{C}$ used so far in the construction of
$G_{\mathcal{C}}'$ is bounded above by the number of edges of
$G_{\mathcal{C}}'$. Observe that by construction $G_{\mathcal{C}}'$ is clearly
a minor of $G_{\mathcal{C}}$. Thereby $G'_{\mathcal{C}}$ is an $H$-minor-free
graph (for every $H\in\mc{F}$) on at most $k + 2 t_{\mathcal{F}} \cdot (1 +
\alpha_r) \cdot k$ vertices  (by Lemma~\ref{lem:numberOfMarkedNodes}). By
Proposition~\ref{prop:numberEdges}, it follows that
$|E(G_{\mathcal{C}}')|\leqslant \alpha_r \cdot ( k + 2 t_{\mathcal{F}} \cdot (1
+ \alpha_r) \cdot k)$ and so there are the same number of used clusters.

Let us now count the number of non-used clusters. Observe that the neighborhood
in $Y_0 \setminus I$ of each non-used cluster induces a (possibly empty) clique
in $G_{\mathcal{C}}'$ (as otherwise some further edge could have been added to
$G_{\mathcal{C}}'$). As by definition distinct clusters have distinct
neighborhoods in $Y_0 \setminus I$, and as $G_{\mathcal{C}}'$ is an
$H$-minor-free graph (for every $H\in\mc{F}$) on at most $k + 2 t_{\mathcal{F}}
\cdot (1 + \alpha_r) \cdot k$ vertices, Proposition~\ref{prop:numberCliques}
implies that the number of non-used clusters is at most $\mu_r \cdot ( k + 2
t_{\mathcal{F}} \cdot (1 + \alpha_r) \cdot k)$. Summarizing, we have that
$|\mathcal{C}| \leqslant k + (\alpha_r + \mu_r) \cdot ( k + 2 t_{\mathcal{F}}
\cdot (1 + \alpha_r) \cdot k) \leqslant (5t_{\mathcal{F}}\alpha_r \mu_r) \cdot
k$, where in the last inequality we have used that $\mu_r \geqslant \alpha_r$
and we have assumed that $\alpha_r \geqslant 4$.
\end{proof}

Piecing all lemmas together, we can now provide a proof of
Proposition~\ref{prop:protrusion decomposition}.

\begin{proof}[Proof of Proposition~\ref{prop:protrusion
decomposition}] By Lemmas~\ref{lem:algo-running-time}
and~\ref{lem:numberOfMarkedNodes} and Observation~\ref{obs:cluster}, we can
compute in linear time a set $Y_0$ of $O(k)$ vertices containing $X$ such that
every cluster of $G-Y_0$ is a restricted $(2t_{\mathcal{F}} + r)$-protrusion.
If $(G,X,k)$ is a \textsc{Yes}-instance of the \DplanarF{} problem, then there
exists a set $\tilde{X}$ of size at most $|X|$ and disjoint from $X$ such that
$G-\tilde{X}$ does not contain any graph $H\in\mc{F}$ as a minor.
Branching on every possible subset of $Y_0\setminus X$, one can guess the
intersection $I$ of $\tilde{X}$ with $Y_0\setminus X$. By
Lemma~\ref{lem:numberOfMarkedNodes}, the branching degree is $2^{O(k)}$. As
$(G,X,k)$ is a \textsc{Yes}-instance, for at least one of the guessed subsets
$I$, the instance $(G_I,Y_0\setminus I,k-|I|)$ is a \textsc{Yes}-instance of
the \DplanarF{} problem. By Lemma~\ref{lem:smallComponents}, the partition
$\mc{P}=(Y_0\setminus I)\uplus Y_1\uplus\cdots \uplus Y_{\ell}$, where
$\{Y_1,\dots,Y_{\ell}\}$ is the set of clusters of $G_I-Y_0$, is an
$(O(k),r+2t_{\mathcal{F}})$-protrusion decomposition of $G_I$.
\end{proof}

\subsection{Solving Planar-F-Deletion with a linear protrusion decomposition}
\label{sec:colorvec}

After having proved Proposition~\ref{prop:protrusion decomposition}, we can now
focus in this subsection on solving \DplanarF{} in single-exponential time when
a linear protrusion decomposition is given. Let $P_{\Pi}(G,S)$ denote the MSO
formula which holds if and only if $G-S$ is $\mc{F}$-minor-free.

Consider an instance $(G,Y_0,k)$ of \DplanarF{} equipped with a linear
protrusion decomposition $\mc{P}$ of $G$. Let $\mc{P}=Y_0\uplus Y_1\uplus
\cdots \uplus Y_{\ell}$ be an $(\alpha,\beta)$-protrusion decomposition of $G$
for some constant $\beta$. The key observation is that for every restricted
protrusion $Y_i$, there is a finite number of representatives such that any
partial solution lying on $Y_i$ can be replaced with one of them while
preserving the feasibility of the solution.



We fix a constant $t$. Let $\mc{U}_t$ be the universe of $t$-boundaried graphs,
and let $\mc{U}_t^{\mbox{{\tiny small}}}$ denote the universe of $t$-boundaried
graphs with treewidth at most $t - 1$. Throughout this subsection
we will assume that all the restricted protrusions belonging to a given
protrusion decomposition have the same boundary size, equal to the maximum
boundary size over all protrusions. This assumption is licit as if some
protrusion has smaller boundary size, we can add dummy independent vertices to
it without interfering with the structure of the solutions and without
increasing the treewidth.

\begin{definition}\label{def:equivset}
Let $\mc{P}=Y_0\uplus Y_1\uplus \cdots \uplus Y_{\ell}$ be an $(\alpha,
t)$-protrusion decomposition of $G$. For each $1 \leqslant i \leqslant \ell$,
we define the following equivalence relation $\sim_{\mc{F},i}$ on subsets of
$Y_i$: for $Q_1,Q_2 \subseteq Y_i$, we define $Q_1 \sim_{\mc{F},i} Q_2$ if for
every $H\in \mc{U}_t$, $G[Y_i^+ \setminus Q_1]\oplus H$ is $\mc{F}$-minor-free
if and only if $G[Y_i^+ \setminus Q_2]\oplus H$ is $\mc{F}$-minor-free.
\end{definition}

Note that $G$ equipped with $\mc{P}$ can be viewed as a gluing of two
$\beta$-boundaried graphs $G[Y_i^+]$ and $G\ominus G[Y_i^+]$, for any
$1\leqslant i\leqslant \ell$, where $Y_i^+=N_{G_I}[Y_i]$. Let us consider the
equivalence relation $\sim_{\mc{F},i}$ applied on $Y_i$ when $G$ is viewed as
such gluing. Extending the notation suggested in Section
\ref{sec:Preliminaries}, we say that $\mc{S}$ is a \emph{set of (minimum-sized) representatives} of
the equivalence relation $\approx$ if $\mc{S}$ contains exactly one element (of minimum cardinality) from every equivalence class under $\approx$. Let $\mc{R}(Y_i):=\{Q^i_1,\ldots
, Q^i_{q_i}\}$ be a set of minimum-sized representatives of equivalence classes
under $\sim_{\mc{F},i}$ for every $1\leqslant i \leqslant \ell$. We say that a
set $\tilde{X}\subseteq V(G)\setminus Y_0$ is {\em decomposable} if
$\tilde{X}=Q^1\cup \cdots \cup Q^{\ell}$ for some $Q^i\in \mc{R}(Y_i)$ for
$1\leqslant i \leqslant \ell$.

\begin{lemma}[Solution decomposability]\label{lem:sol-decomp}
Let $(G,Y_0,k)$ be an instance of \DplanarF{} and let $\mc{P}=Y_0\uplus
Y_1\uplus \cdots \uplus Y_{\ell}$ be an $(\alpha,\beta)$-protrusion
decomposition of $G$. Then, there exists a solution $\tilde{X} \subseteq
V(G)\setminus Y_0$ of size at most $k$ if and only if there exists a
decomposable solution $\tilde{X}^*\subseteq V(G)\setminus Y_0$ of size at most
$k$.
\end{lemma}
\begin{proof}
Let $\tilde{X}$ be a subset of $V(G)\setminus Y_0$. Let $S_i:=\tilde{X}\cap Y_i$ for every $1
\leqslant i\leqslant \ell$, $\bar{S}_i:=\tilde{X}\cap (V(G)\setminus
Y_i)$ and let $H:=G\ominus G[Y_i^+]-\bar{S}_i$ be the associated $t$-boundaried graph
with $\bound(H):=\bound(G\ominus G[Y_i^+])$. Fix $i$ and choose the (unique) representative $Q^i\in \mc{R}(Y_i)$ such
that $Q^i\sim_{\mc{F},i} S_i$. Note that $S_i\cap \bound(H) = \bar{S}_i \cap \bound(H) = \emptyset$.

We claim that $G-\tilde{X}$ is $\mc{F}$-minor-free if and only
if $G-(Q^i\cup \bar{S}_i)$ is $\mc{F}$-minor-free. Indeed, $G-\tilde{X}=G-(S_i\cup
\bar{S}_i)$, which can be written as $G[Y_i^+ \setminus S_i]\oplus H$. From the choice of $Q^i\in \mc{R}(Y_i)$ such that $Q^i\sim_{\mc{F},i}
S_i$, it follows that $G[Y_i^+ \setminus S_i]\oplus H$ is $\mc{F}$-minor-free if and only if $G[Y_i^+ \setminus Q^i]\oplus H$ is so. Noting that $G[Y_i^+ \setminus Q^i]\oplus H=G-(Q^i\cup
\bar{S}_i)$ proves our claim.



By replacing each $S_i$ with its representative $Q^i\in \mc{R}(Y_i)$, we
eventually obtain $\tilde{X}^*$ of the form $\tilde{X}^* = \bigcup_{1\leqslant
i\leqslant \ell}Q^i$, where $Q^i\in \mc{R}(Y_i)$ is the representative of $S_i$ for every
$1\leqslant i\leqslant \ell$. Finally, it holds that $P_{\Pi}(G,\tilde{X}^*)$
if and only if $P_{\Pi}(G,\tilde{X})$. It remains to observe that
$|Q^i|\leqslant |S_i|$, as we selected a minimum-sized set of an equivalence
class of $\sim_{\mc{F},i}$ as its representative.
\end{proof}

We are now ready to prove Proposition~\ref{prop:sol-decomp-single}.

\begin{reminder}{Proposition~\ref{prop:sol-decomp-single}.}
Let $(G,Y_0,k)$ be an instance of \DplanarF{} and let $\mc{P}=Y_0\uplus
Y_1\uplus \cdots \uplus Y_{\ell}$ be an $(\alpha,\beta)$-protrusion
decomposition of $G$, for some constant $\beta$. There exists an
$2^{O(\ell)}\cdot n$-time algorithm which computes a solution $\tilde{X}
\subseteq V(G)\setminus Y_0$ of size at most $k$ if it exists, or correctly
decides that there is no such solution.\end{reminder}
\begin{proof}
By Lemma~\ref{lem:sol-decomp}, either there exists a solution of size at most
$k$ which is decomposable or $(G,Y_0,k)$ is a \textsc{No}-instance. Assume that
the representatives $\mc{R}(Y_i)$ of the equivalence relations
$\sim_{\mc{F},i}$ are given for all $1\leqslant i \leqslant \ell$. Then, for a
decomposable set $\tilde{X}$, one can decide if $\tilde{X}$ is a solution or
not in time $O(h(t_{\mc{F}})\cdot n)$. Indeed, for $\tilde{X}$ to be a
solution, the treewidth of $G-\tilde{X}$ is at most $t_{\mc{F}} - 1$. Using the
algorithm of Bodlaender~\cite{Bod96}, one can decide in time
$2^{O(t_{\mc{F}}^3)}\cdot n$ whether a graph is of treewidth at most
$t_{\mc{F}}$ and if so, build a tree-decomposition of width at most
$t_{\mc{F}}$. Courcelle's theorem~\cite{Cou90} says that testing an
MSO-definable property on treewidth-$t_{\mc{F}}$ graphs can be done in linear
time, where the hidden constant depends solely on the treewidth $t_{\mc{F}}$ and the length of the MSO-sentence.
It follows that one can decide whether $G-\tilde{X}$ is $\mc{F}$-minor-free or
not in time $O(h(t_{\mc{F}})\cdot n)$. Here $h(t_{\mc{F}})$ is an additive
function resulting from Bodlaender's treewidth testing algorithm and
Courcelle's MSO-model checking algorithm, which depends solely on the treewidth
$t_{\mc{F}}$ and the MSO formula $|P_{\Pi}(G,\tilde{X})|$. It remains to observe that there are at
most $2^{O(\ell)}$ decomposable sets to consider. This is because an
MSO-definable graph property has finitely many equivalence classes on
$\mc{U}_t$ for every fixed $t$~\cite{Cou90,BvF01} (hence, also on
$\mc{U}_t^{\mbox{{\tiny small}}}$), and being $\mc{F}$-minor-free is an
MSO-definable property. \end{proof}

\paragraph{Constructing the sets of representatives.}
We remark that in the proof of Proposition~\ref{prop:sol-decomp-single} we have
assumed that the sets of representatives $\mc{R}(Y_i)$ are {\sl given} to the
algorithm. This makes the algorithm of Proposition~\ref{prop:sol-decomp-single}
(and therefore, also the algorithm of Theorem~II) {\sl non-constructive}, and
uniform on $k$ but non-uniform on the family $\mathcal{F}$. That is, for each
family $\mathcal{F}$ we deduce the existence of a different algorithm. In order
to make our algorithms constructive and uniform also on $\mathcal{F}$, we now
proceed to explain how to {\sl construct} the minimum-sized sets of
representatives $\mc{R}(Y_i)$ in linear time~\footnote{We note that Bodlaender
\emph{et al}.~\cite{BFLPST09} sketched how to compute such a set in linear
time, and full details will be given in the journal version of~\cite{BFLPST09}.
We provide our own proof here for completeness.}.

Our strategy can be summarized as follows. We will first define an equivalence
relation $\equiv_{\mc{F},t}$ on $\mc{U}_t^{\mbox{{\tiny small}}}$, with the
objective of capturing all possible behaviors of the graphs $G \ominus
G[Y_i^+]$ (we call such graphs the \emph{context} of the restricted protrusion
$Y_i$). From classic tree automaton theory~\cite{DF99,CE12}, Courcelle's
theorem~\cite{Cou90}, and the fact that $\equiv_{\mc{F},t}$ has finitely many
equivalence classes for each fixed $t$, it follows that a finite set $\mc{K}$
of representatives of $\equiv_{\mc{F},t}$ can be efficiently computed in
constant time. Once we have the set $\mc{K}$, we compute our desired set
$\mc{R}(Y_i)$ in each restricted protrusion $Y_i$ by solving a constant number
of \pmincmso{} graph problems on graphs of bounded treewidth, using the dynamic
programming linear-time algorithm of Borie \emph{et al}.~\cite{BPT92}. The
details follow.

For two $t$-boundaried graphs $K_1$ and $K_2$ from $\mc{U}_t^{\mbox{{\tiny
small}}}$, we say that $K_1 \equiv_{\mc{F},t} K_2$ if for every $Y\in
\mc{U}_t^{\mbox{{\tiny small}}}$, $K_1\oplus Y$ is $\mc{F}$-minor-free iff $K_2
\oplus Y$ is $\mc{F}$-minor-free. The intuition is that the graphs $K_1$ and
$K_2$ will correspond to contexts of a restricted protrusion $Y$. As the
property of being $\mc{F}$-minor-free is MSO-definable, $\equiv_{\mc{F},t}$ has
finitely many equivalence classes for each fixed $t$ and there is a finite set
$\mc{K}=\{K_1,\ldots, K_M\}$ of representatives of $\equiv_{\mc{F},t}$ (cf. for
instance~\cite{DF99,CE12}). Such a set of representatives $\mc{K}$ can be efficiently constructed on
the universe $\mc{U}_t^{\mbox{{\tiny small}}}$ from the given MSO formula and
the boundary size $t$.
Let us discuss the main line of the proof of this fact. Courcelle's theorem is
proved\footnote{In fact, there is more than one proof of Courcelle's theorem.
The one we depict in this article is as presented in~\cite{FG06}, which differs
from the original proof of Courcelle~\cite{Cou90}.} by converting an MSO
formula $\varphi$ on tree-decompositions of width $t$ into another MSO formula
$\varphi'$ on labeled trees. Trees, in which every internal node has bounded
fan-in and every node is labeled with an alphabet chosen from a fixed set, are
considered as a \emph{tree language}, which is a natural generalization of the
usual string language. It is well-known (as the analogue of B\"{u}chi's
theorem~\cite{Buc62} on tree languages) that the set of labeled trees for which
an MSO formula holds form a regular (tree) language\footnote{A regular tree
language is an analogue of a regular language on labeled trees. Appropriately
defined, most of the nice properties on string regular languages transfer
immediately to tree regular languages. It is beyond the scope of this article
to give details of tree languages and tree automatons. We invite the interested
readers to~\cite{DF99,CE12}.}. Moreover, based on its proof it is not difficult
to construct a finite tree automaton (cf. for instance~\cite{FG06}). In
particular, the number of states in the corresponding tree automaton is bounded
by a constant depending only on $\varphi$ and $t$. From this, it is possible to
prove (using a ``tree'' pumping lemma) that one can assume that the height of a
\emph{distinguishing extension} of two labeled trees is bounded by a constant
as well (in fact, the size of the tree automaton). Hence we can enumerate all
possible labeled trees of bounded height, which will be a test set to construct
the set of representatives $\mc{K}$. Now one can apply the so-called {\em
method of test sets} (basically implicit in the proof of the Myhill-Nerode
theorem~\cite{Ner58}, see~\cite{DF99,AF93} for more details) and retrieve the
set of representatives $\mc{K}$.


Once we have at hand the set $\mc{K}=\{K_1,\ldots, K_M\}$, we now proceed to
find the set of representatives $\mc{R}(Y_i)$ in time $O(|Y_i|)$ for every $1
\leqslant i \leqslant \ell$. We use a strategy inspired by the method of test
sets.
We consider the set of all binary vectors with $M$ coordinates, to which we
give the following interpretation. Each fixed such vector $\mathbf{v} =
(b_1,\ldots ,b_M)$ will correspond to a minimum-sized subset $Q_\mathbf{v}
\subseteq Y_i$ such that, for $1 \leqslant j \leqslant M$, the graph $G[(Y_i^+
\setminus Q_\mathbf{v}) \oplus K_j]$ is $\mc{F}$-minor-free iff $b_j = 1$.
Formally, for each binary vector $(b_1,\ldots ,b_M)$ of length $M$ we consider
the following optimization problem: find a set $Q\subseteq Y_i$ of minimum size
such that $\varphi(G[Y_i^+],Q)$ holds. Here $\varphi(G[Y_i^+],Q):=(Q\subseteq
Y_i) \wedge \left(\bigwedge_{j=1}^M \tilde{b}_j \right)$, where
$\tilde{b}_j:=\varphi_{K_j}(G[Y_i^+], Q)$ if $b_j=1$ and $\tilde{b}_j:=\neg
\varphi_{K_j}(G[Y_i^+], Q)$ if $b_j=0$, each $\varphi_{K_j}(G[Y_i^+], Q)$
stating that $G[(Y_i^+ \setminus Q) \oplus K_j]$ is $\mc{F}$-minor-free. For
each fixed $K_j\in \mc{U}_t^{\mbox{{\tiny small}}}$, whether $G[(Y_i^+
\setminus Q) \oplus K_j]$ is $\mc{F}$-minor-free or not depends only on
$G[Y_i^+]$ and $Q$, and moreover this property can be (tediously) expressed as
an MSO formula. As $\varphi$ is an MSO formula, we can apply the linear-time
dynamic programming algorithm of Borie \emph{et al}.~\cite{BPT92} on graphs of
bounded treewidth to solve the associated optimization problem. Note that the
running time is $O(|Y_i|)$, whose hidden constant depends solely on $|P_{\Pi}|$
and the treewidth $t$.

\begin{claim} Let $\mc{R}_i$ be the set of the optimal solutions over all $2^M$ binary
vectors of length $M$, obtained as explained above. Then $\mc{R}_i$
 is a set of minimum-sized representatives of~$\sim_{\mc{F},i}$.
\end{claim}
\begin{proof} We fix $t := 2 t_{\mathcal{F}} + r$, so we can assume
that all protrusions $Y_i^+$ belong to $\mc{U}_t^{\mbox{{\tiny small}}}$.

First note that in the definition the equivalence relation~$\sim_{\mc{F},i}$
(cf.~Definition~\ref{def:equivset}), one only needs to consider graphs $H \in
\mc{U}_t^{\mbox{{\tiny small}}}$. Indeed, if $H \in \mc{U}_t \setminus
\mc{U}_t^{\mbox{{\tiny small}}}$, then for any $Q \subseteq Y_i$ it follows
that $G[Y_i^+ \setminus Q]\oplus H$ is \emph{not} $\mc{F}$-minor-free (as
$\tw(G[Y_i^+ \setminus Q]\oplus H) \geqslant 2 t_{\mathcal{F}} + r
> t_{\mathcal{F}}$), so in order to define the equivalence classes of subsets of
$Y_i$ it is enough to consider $H \in \mc{U}_t^{\mbox{{\tiny small}}}$. In
other words, only the elements of $\mc{U}_t^{\mbox{{\tiny small}}}$ can {\sl
distinguish} the subsets of $Y_i$ with respect to~$\sim_{\mc{F},i}$.


Let $Q \subseteq Y_i$, and we want to prove that there exists $R_Q \in
\mc{R}_i$ such that $Q \sim_{\mc{F},i} R_Q$, that is, such that for any $H \in
\mc{U}_t$, $G[Y_i^+ \setminus Q]\oplus H$ is $\mc{F}$-minor-free iff $G[Y_i^+
\setminus R_Q]\oplus H$ is $\mc{F}$-minor-free. By the remark in the above
paragraph, we can assume that $H \in \mc{U}_t^{\mbox{{\tiny small}}}$, as
otherwise the statement is trivially true. Let $\mathbf{v}_Q =
(b_1,\ldots,b_M)$ be the binary vector on $M$ coordinates such that, for $1
\leqslant j \leqslant M$, $b_j = 1$ iff $G[(Y_i^+ \setminus Q) \oplus K_j]$ is
$\mc{F}$-minor-free. We define $R_Q$ to be the graph in $\mc{R}_i$
corresponding to the vector $\mathbf{v}_Q$. As $H \in \mc{U}_t^{\mbox{{\tiny
small}}}$, there exists $K_H \in \mathcal{K}$ such that $H \equiv_{\mc{F},t}
K_H$. Then, $G[Y_i^+ \setminus Q]\oplus H$ is $\mc{F}$-minor-free iff $G[Y_i^+
\setminus Q]\oplus K_H$ is $\mc{F}$-minor-free, which by construction is
$\mc{F}$-minor-free iff $G[Y_i^+ \setminus R_Q]\oplus K_H$ is
$\mc{F}$-minor-free, which is in turn $\mc{F}$-minor-free iff $G[Y_i^+
\setminus R_Q]\oplus H$ is $\mc{F}$-minor-free, as we wanted to prove.
\end{proof}


To summarize the above discussion, a set of representatives $\mc{R}(Y_i)$ can
be constructed in time $O(|Y_i|)$, and therefore all the sets of
representatives can be constructed in time $O(n)$.

\subsection{Proof of Theorem II}
\label{sec:proofBigTheorem}

We finally have all the ingredients to prove Theorem~II.

\begin{reminder}{Theorem~II.}
The parameterized \planarF{} problem can be solved in time $2^{O(k)}\cdot
n^2$.\end{reminder}
\begin{proof}
Lemma~\ref{lem:ic} states that \planarF{} can be reduced to \DplanarF{} so that
the former can be solved in single-exponential time solvable provided that the
latter is so, and the degree of the polynomial function just increases by one.
We now proceed to solve \DplanarF{} in time $2^{O(k)}\cdot n$. Given an
instance $(G,X,k)$ of \DplanarF{}, we apply Proposition~\ref{prop:protrusion
decomposition} to either correctly decide that $(G,X,k)$ is a
\textsc{No}-instance, or identify in time $2^{O(k)} \cdot n$ a set $I \subseteq
V(G)$ of size at most $k$ and a $(O(k),2t_{\mathcal{F}}+r)$-protrusion
decomposition $\mc{P}=Y_0\uplus Y_1\uplus\cdots\uplus Y_{\ell}$ of $G-I$, with
$X\subseteq Y_0$, such that there exists a set $X' \subseteq V(G) \setminus
Y_0$ of size at most $k-|I|$ such that $G-\tilde{X}$, with $\tilde{X}=X'\cup
I$, is $H$-minor-free for every graph $H\in\mc{F}$. Finally, using
Proposition~\ref{prop:sol-decomp-single} we can solve the instance
$(G_I,Y_0\setminus I, k-|I|)$ in time $2^{O(k)}\cdot n$.\end{proof}}

\section{Conclusions and further research}\label{sec:Conclusion}
\label{sec:conclusions}

\short{Concerning our kernelization algorithms, two main questions arise:
(1)~can similar results be obtained for an even larger class of (sparse)
graphs; and (2)~which other problems have linear kernels on
$H$-topological-minor free graphs. We would like to note that the degree of the
polynomial of the running time of our kernelization algorithm depends linearly
on the size of the excluded topological minor $H$. It seems that the recent
{\sl fast protrusion replacer} of Fomin \emph{et al}.~\cite{FLMS12} could
be applied to get rid of this dependency on $H$.

Concerning the \planarF{} problem, no single-exponential algorithm is known
when the family $\mathcal{F}$ does not contain any planar graph. Is it possible
to find such a family, or can it be proved that, under some complexity
assumption, a single-exponential algorithm is not possible? Very recently, a
randomized (Monte Carlo) constant-factor approximation algorithm for \planarF{}
has been given by Fomin \emph{et al}.~\cite{FLMS12}. Finding a
deterministic constant-factor approximation remains open.}

\journal{We presented a simple algorithm to compute protrusion decompositions
for graphs $G$ that come equipped with a set $X \subseteq V(G)$ such that the
treewidth of $G-X$ is at most some fixed constant $t$. Then we showed that this
algorithm can be used in order to achieve two different sets of results: linear
kernels on graphs excluding a fixed topological minor, and a single-exponential
parameterized algorithm for the \planarF{} problem.

Concerning our kernelization algorithm, two main questions arise: (1)~can
similar results be obtained for an even larger class of (sparse) graphs; and
(2)~which other problems have linear kernels on $H$-topological-minor free
graphs. In particular, it has been recently proved by Fomin \emph{et
al}.~\cite{FLST12} that \textsc{(Connected) Dominating Set} has a linear kernel
on $H$-minor-free graphs, but it remains open whether it is also the case on
$H$-topological-minor-free graphs. It would also be interesting to investigate
how the structure theorem by Grohe and Marx~\cite{GM11} can be used in this
context.

We would like to note that the degree of the polynomial of the running time of
our kernelization algorithm depends linearly on the size of the excluded
topological minor $H$. It seems that the recent {\sl fast protrusion replacer}
of Fomin \emph{et al}.~\cite{FLMS12} could be applied to get rid of the
dependency on $H$ of the running time.

Let us now discuss some further research related the our single-exponential
algorithm for \planarF{}. As mentioned in the introduction, no
single-exponential algorithm is known when the family $\mathcal{F}$ does not
contain any planar graph. Is it possible to find such a family, or can it be
proved that, under some complexity assumption, a single-exponential algorithm
is not possible? An ambitious goal would be to optimize the constants involved
in the function $2^{O(k)}$, possibly depending on the family $\mathcal{F}$, and
maybe even proving lower bounds for such constants, in the spirit of Lokshtanov
\emph{et al}.~\cite{LMS11} for problems parameterized by treewidth.

We showed (in Section~\ref{sec:colorvec}) how to obtain single-exponential
algorithms for \DplanarF{}  with a given linear protrusion decomposition. This
approach seems to be applicable to general vertex deletion problems to attain a
property expressible in CMSO (but probably, the fact of having bounded
treewidth is needed in order to make the algorithm constructive). It would be
interesting to generalize this technique to \pmaxcmso{} or \peqcmso{} problems,
as well as to edge subset problems.

Very recently, a randomized (Monte Carlo) constant-factor approximation
algorithm for \planarF{} has been given by Fomin \emph{et
al}.~\cite{FLMS12b}. Finding a deterministic constant-factor approximation
remains open. Also, the existence of linear or polynomial kernels for
\planarF{}, or even for the general $\mathcal{F}$-\textsc{Deletion} problem, is
an exciting avenue for further research. It seems that significant advances in
this direction, namely for \planarF{}, have been done by Fomin \emph{et
al}.~\cite{FLMS12b}. It is worth mentioning that very recently Fomin \emph{et
al}.~\cite{FJP12} have proved that $\mathcal{F}$-\textsc{Deletion} admits a
polynomial kernel when parameterized by the size of a vertex cover.

In the parameterized dual version of the $\mathcal{F}$-\textsc{Deletion}
problem, the objective is to find at least $k$ vertex-disjoint subgraphs of an
input graph, each of them containing some graph in $\mathcal{F}$ as a minor.
For $\mathcal{F} = \{K_3\}$, the problem corresponds to \textsc{$k$-Disjoint
Cycle Packing}, which does not admit a polynomial kernel on general
graphs~\cite{BTY09} unless $\coNP \subseteq \NP/poly$. Does this problem, for
some non-trivial choice of $\mathcal{F}$, admit a single-exponential
parameterized algorithm?}

\paragraph{Acknowledgements.} We thank Dimitrios M. Thilikos, Bruno Courcelle,
Daniel Lokshtanov, and Saket Saurabh for interesting discussions and helpful
remarks on the manuscript.

\short{\newpage}

\def\redefineme{
    \def\LNCS{LNCS}%
    \def\ICALP##1{Proc. of ##1 ICALP}%
    \def\FOCS##1{Proc. of ##1 FOCS}%
    \def\COCOON##1{Proc. of ##1 COCOON}%
    \def\SODA##1{Proc. of ##1 SODA}%
    \def\SWAT##1{Proc. of ##1 SWAT}%
    \def\IWPEC##1{Proc. of ##1 IWPEC}%
    \def\IWOCA##1{Proc. of ##1 IWOCA}%
    \def\ISAAC##1{Proc. of ##1 ISAAC}%
    \def\STACS##1{Proc. of ##1 STACS}%
    \def\IWOCA##1{Proc. of ##1 IWOCA}%
    \def\ESA##1{Proc. of ##1 ESA}%
    \def\WG##1{Proc. of ##1 WG}%
    \def\LIPIcs##1{LIPIcs}%
    \def\LIPIcs{LIPIcs}%
    \def\LICS##1{Proc. of ##1 LICS}%
}

\bibliographystyle{cplain}
\bibliography{cross,conf}

\appendix

\section{Edge modification problems are not minor-closed}
\label{ap:NotMinorClosed}

A graph problem $\Pi$ is \emph{minor-closed} if whenever $G$ is a
\textsc{Yes}-instance of $\Pi$ and $G'$ is a minor of $G$, then $G'$ is also a
\textsc{Yes}-instance of $\Pi$. It is easy to see that
\textsc{$\mc{F}$-(Vertex-)Deletion} is minor-closed, and therefore it is FPT by
Robertson and Seymour~\cite{RS95c}. Here we show that the edge modification
versions, namely, \textsc{$\mc{F}$-Edge-Contraction} and
\textsc{$\mc{F}$-Edge-Removal} (defined in the natural way), are not
minor-closed.

\paragraph{Edge contraction.} In this case, the problem $\Pi$ is whether one can contract at most $k$ edges from
a given graph $G$ so that the resulting graph does not contain any of the
graphs in $\mathcal{F}$ as a minor. Let $\mathcal{F}= \{K_5,K_{3,3}\}$, and let
$G$ be the graph obtained from $K_5$ by subdividing every edge $k$ times, and
adding an edge $e$ between two arbitrary original vertices of $K_5$. Then $G$
can be made planar just by contracting edge $e$, but if $G'$ is the graph
obtained from $G$ by deleting $e$ (which is a minor of $G$), then at least
$k+1$ edge contractions are required to make $G'$ planar.

\paragraph{Edge deletion.} In this case, the problem $\Pi$ is whether one can delete at most $k$ edges from
a given graph $G$ so that the resulting graph does not contain any of the
graphs in $\mathcal{F}$ as a minor. Let $G$, $G'$, and $H$ be the graphs
depicted in Figure~\ref{fig:edgeRemoval}, and let $k=1$. Then $G$ can be made
$H$-minor-free by deleting edge $e$, but $G'$, which is the graph obtained from
$G$ by contracting edge $e$, needs at least two edge deletions to be
$H$-minor-free.

\begin{figure}[h!]
    \center
    \journal{
        \vspace{-.05cm}
        \includegraphics[width=0.7\textwidth]{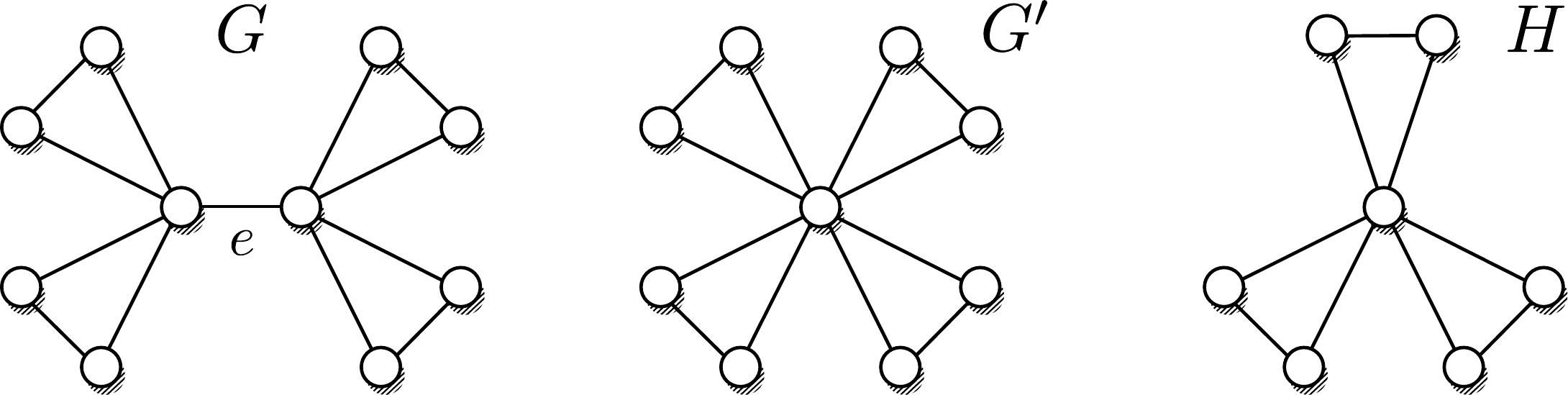}
    }
    \short{
        \includegraphics[width=0.48\textwidth]{edgeRemoval}
        \vspace{-.1cm}
    }
    \vspace{-.05cm} \caption{Example to show that \textsc{$\mc{F}$-Edge-Removal} is
    not minor-closed.\label{fig:edgeRemoval}}\vspace{-.15cm}
\end{figure}

\section{Disconnected \textsc{Planar-$\mathcal{F}$-Deletion} has not finite integer index}
\label{ap:notFII}

We proceed to prove that if $\mathcal{F}$ is a family of graphs containing some
disconnected graph $H$ (planar or non-planar), then the
\textsc{$\mc{F}$-Deletion} problem has not finite integer index (FII) in
general.

We shall use the equivalent definition of FII as suggested for graph
optimization problems, see~\cite{Flu97}. For a graph problem $o$-$\Pi$, the
equivalence relation $\sim_{o\text{-}\Pi,t}$ on $t$-boundaried graphs is
defined as follows. Let $G_1$ and $G_2$ be two $t$-boundaried graphs. We define
$G_1 \sim_{\Pi,t} G_2$ if and only if there exists an integer $i$ such that for
any $t$-boundaried graph $H$, it holds $\pi(G_1 \oplus H) = \pi(G_2 \oplus H) +
i$, where $\pi(G)$ denotes the optimal value of problem $o$-$\Pi$ on graph $G$.
We claim that $G_1 \sim_{\Pi,t} G_2$ if and only if $G_1 \equiv_{\Pi,t} G_2$
(recall Definition~\ref{def:finiteii} of canonical equivalence), where
$\Pi$ is the parameterized version of $o$-$\Pi$ with the solution size as a
parameter. Suppose $G_1 \sim_{\Pi,t} G_2$ and let $\pi(G_1 \oplus H) = \pi(G_2
\oplus H) + i$. Then
$$(G_1\oplus H, k)\in \Pi\ \rightleftharpoons \ \pi(G_1 \oplus H) \leqslant k\ \rightleftharpoons\ \pi(G_2 \oplus H) \leqslant k-i (G_2\oplus H, k-i)\in \Pi,$$ and thus the forward implication holds. The opposite direction is easy to see.



Let $F_1$ and $F_2$ be two incomparable graphs with respect to the minor
relation, and let $F$ be the disjoint union of $F_1$ and $F_2$. For instance,
if we want $F$ to be planar, we can take $F_1 = K_4$ and $F_2 = K_{2,3}$. We
set $\mathcal{F} = \{F\}$. Let $\Pi$ be the non-parameterized version of
\textsc{$\mc{F}$-Vertex Deletion}.

For $i \geqslant 1$, let $G_i$ be the 1-boundaried graph consisting of the
boundary vertex $v$ together with $i$ disjoint copies of $F_1$, and for each
such copy, we add and edge between $v$ and an arbitrary vertex of $F_1$.
Similarly, for $j \geqslant 1$, let $H_j$ be the 1-boundaried graph consisting
of the boundary vertex $u$ together with $j$ disjoint copies of $F_2$, and for
each such copy, we add and edge between $u$ and an arbitrary vertex of $F_2$.

By construction, if $i,j \geqslant 1$, it holds $\pi(G_i \oplus H_j) =
\min\{i,j\}$. Then, if we take $1 \leqslant n < m$,
\begin{eqnarray*}
\pi(G_n \oplus H_{n-1}) -  \pi(G_m \oplus H_{n-1}) & = & (n-1) - (n-1) \ =\   0,\\
\pi(G_n \oplus H_{m}) -  \pi(G_m \oplus H_{m}) & = & n - m\  <\  0.
\end{eqnarray*}
Therefore, $G_n$ and $G_m$ do not belong to the same equivalence class of
$\sim_{\Pi,1}$ whenever $1 \leqslant n < m$, so $\sim_{\Pi,1}$ has infinitely
many equivalence classes, and thus $\Pi$ has not FII.

In particular, the above example shows that if $\mathcal{F}$ may contain some
disconnected planar graph $H$, then \textsc{Planar-$\mathcal{F}$-Deletion} has
not FII in general.

\section{MSO formula for topological minor containment}
\label{ap:MSOformula}

\def\adj{\textsc{adj}\xspace}

For a fixed graph $H$ we describe and MSO$_1$-formula $\Phi_H$ over the usual
structure consisting of the universe $V(G)$ and a binary symmetric relation \adj modeling $E(G)$
such that $G \models \Phi_H$ iff $H \tminor G$.
\def\DISJ{\text{{\sc dis}}}
\def\CONN{\text{{\sc conn}}}
\begin{eqnarray*}
    \Phi_H(G) &:=& \exists x_{v_1} \dots \exists x_{v_r} \exists D_{e_1} \dots \exists D_{e_\ell}  \\
              && \Big( \,\enskip \bigwedge_{{\mathclap{1\leq i < j\leq r}}} x_{v_i} \neq x_{v_j}
                    \wedge \bigwedge_{{\mathclap{\stackrel{1\leq i\leq r}{1\leq j \leq \ell}}}} x_{v_i} \not \in D_{e_j}
                    \wedge \bigwedge_{\smash{\mathclap{1\leq i < j\leq r}}} \DISJ(D_{e_i},D_{e_j})
                    \wedge \bigwedge_{\mathclap{\stackrel{1\leq j\leq \ell}{e_j = v_i v_k}}} \CONN(x_{v_i}, D_{e_j},x_{v_k})
                            \, \Big) \\
    \text{with}~\DISJ(X,Y) &:=& \forall x( x \in X \rightarrow x \not \in Y ) \\
    \text{and}~\CONN(u,X,v) &:=& \exists w( \adj(u,w) \wedge w\in X)  \wedge \exists w(\adj(v,w) \wedge w\in X) \\
            && \wedge \forall A \forall B( (A \subseteq X \wedge B \subseteq X \wedge \DISJ(A,B)) \rightarrow
               \exists a \exists b (a \in A \wedge b \in B \wedge \adj(a,b)) )
\end{eqnarray*}

The subformula $\CONN(u,X,v)$ expresses that $u,v$ are adjacent to $X$ and that $G[X]$ is connected,
which implies that there exists a path from $u$ to $v$ in $G[X \cup \set{u,v}]$.
By negation we can now express that $G$ does not contain $H$ as a topological minor,  \ie $G \models \neg \Phi_H$ iff $G$ is $H$-topological-minor-free.

\end{document}